\documentclass[aps,pra,twocolumn,10pt,nofootinbib,floatfix,superscriptaddress]{revtex4-2}
\usepackage[utf8]{inputenc}
\usepackage[T1]{fontenc}
\usepackage{lmodern}

\usepackage{graphicx}
\usepackage{dcolumn}
\usepackage{bm}
\usepackage{multirow}
\usepackage{amssymb}
\usepackage{amsfonts}
\usepackage{amsmath}
\usepackage{longtable}
\usepackage{amsthm}
\usepackage{enumerate}
\usepackage{mathtools}
\usepackage{cancel}
\usepackage{xcolor}
\usepackage{marvosym}
\newcommand{\mathbbm}[1]{\text{\usefont{U}{bbm}{m}{n}#1}}
\usepackage{hyperref}

\usepackage{makecell}

\usepackage{tikz}
\usetikzlibrary{matrix,calc,shapes,positioning,decorations.pathreplacing}
\tikzset{
  treenode/.style = {shape=rectangle, rounded corners,
                     draw, anchor=center,
                     text width=10em, align=center,
                     inner sep=1ex},
  decision/.style = {treenode, diamond, inner sep=0pt,text width=8em,},
  root/.style     = {treenode},
  env/.style      = {treenode, font=\ttfamily\normalsize},
  finish/.style   = {root, fill=gray!25},
  finishnonnormal/.style   = {root, fill=gray!15,dashed},
  dummy/.style    = {circle,draw}
}
\newcommand{\yes}{edge node [above] {yes}}
\newcommand{\no}{edge  node [left]  {no}}

\newtheorem{theorem}{Theorem}

\newtheorem{observation}{Observation}

\newtheorem{lemma}{Lemma}
\newtheorem{definition}{Definition}

\newcommand{\tr}{\text{tr}}

\newcommand{\id}{\ensuremath{\mathbbm{1}}} 
\newcommand{\one}{\id}

\newcommand{\bra}[1]{\langle #1|}
\newcommand{\ket}[1]{|#1\rangle}

\newcommand{\C}{\ensuremath{\mathbbm C}}

\renewcommand\Re{\operatorname{Re}}
\renewcommand\Im{\operatorname{Im}}

\let\vec\mathbf

\newcommand{\bi}{\begin{itemize}}
\newcommand{\ei}{\end{itemize}}

\newcommand{\be}{\begin{equation}}
\newcommand{\ee}{\end{equation}}
\newcommand{\bea}{\begin{eqnarray}}
\newcommand{\eea}{\end{eqnarray}}

\newcommand{\kommentar}[1]{}

\newcommand{\identity}{\mathbbm{1}}

\newcommand{\forget}[1]{}

\newcommand{\Nto}{\overset{N}{\rightarrow}}

\usepackage{float}
\usepackage[caption=false]{subfig}

\begin{document}

\title{Symmetries and local transformations of translationally invariant Matrix Product States }

\date{\today}

\author{Martin Hebenstreit}
\email{martin.hebenstreit@uibk.ac.at}
\affiliation{Institute for Theoretical Physics, University of Innsbruck, 6020 Innsbruck, Austria}

\author{David Sauerwein}
\email{sauerwein.david@gmail.com; This work was done prior to joining AWS.}
\affiliation{Amazon Web Services Europe, Zürich, Switzerland}

\author{Andras Molnar}
\email{andras.molnar@univie.ac.at}
\affiliation{University of Vienna, Faculty of Mathematics, 1090 Vienna, Austria}

\author{J. Ignacio Cirac}
\email{ignacio.cirac@mpq.mpg.de}
\affiliation{Max-Planck-Institut für Quantenoptik, Hans-Kopfermann-Straße 1, 85748 Garching, Germany}

\author{Barbara Kraus}
\email{barbara.kraus@uibk.ac.at}
\affiliation{Institute for Theoretical Physics, University of Innsbruck, 6020 Innsbruck, Austria}

\begin{abstract}
We determine the local symmetries and local transformation properties of translationally invariant matrix product states (MPS). We focus on physical dimension $d=2$ and bond dimension $D=3$ and use the procedure introduced in D. Sauerwein {\it et al.}, Phys. Rev. Lett. {\bf 123}, 170504 (2019) to determine all (including non--global) symmetries of those states. We identify and classify the stochastic local transformations (SLOCC) that are allowed among MPS. We scrutinize two very distinct sets of MPS and show the big diversity (also compared to the case $D=2$) occurring in both, their symmetries and the possible SLOCC transformations. These results reflect the variety of local properties of MPS, even if restricted to translationally invariant states with low bond dimension. Finally, we show that states with non-trivial local symmetries are of measure zero for d = 2 and D > 3.

\end{abstract}

\maketitle


\section{Introduction}
Entanglement is a unique quantum property that is behind most modern quantum technologies, such as quantum computers \cite{Nielsen2010,Horodecki2009}, and is key to comprehend important features of quantum many-body systems \cite{Amico2008, Eisert2010}. The relevance of entanglement in many different branches of science has spurred significant research efforts to understand its properties  \cite{Horodecki2009}. While this has led to a clear understanding of bipartite entanglement, many questions are still open in the multipartite realm.

At the heart of entanglement theory lies the fact that entanglement is a resource under local operations assisted by classical communication (LOCC), which are the most general operations that spatially separated parties can use to manipulate a shared entangled state. Transformations of entangled states via LOCC induce a physically meaningful partial order on the set of entangled states: If $\Psi$ can be deterministically transformed into $\Phi$ via LOCC this means that $\Psi$ is at least as entangled as $\Phi$ \cite{Horodecki2009}. Furthermore, if two states cannot even probabilistically be transformed into each other, via so-called stochastic LOCC (SLOCC), they contain different, incomparable kinds of entanglement. This entails that they may be useful in different contexts of quantum information science \cite{Bennett2000}. Hence, the characterization of LOCC transformations is central in entanglement theory.

Bipartite LOCC transformations among pure states admit a simple characterization \cite{Nielsen1999, Vidal1999} that led to a clear understanding of bipartite entanglement and inspired a wide range of applications \cite{Horodecki2009}. A general characterization of multipartite LOCC is still elusive. This is, among other reasons, due to the notorious mathematical complexity of multipartite LOCC \cite{Chitambar2014}, the fact that there are multipartite entangled states that cannot even be transformed into each other via SLOCC \cite{Dur2000} and the exponential growth of the Hilbert space dimension as a function of the number of constituent subsystems.

However, most multipartite (pure) states are not particularily interesting from the perspective of state transformations. On the one hand, this is because most of them cannot even be reached in polynomial time, even if constant-size nonlocal quantum gates are allowed \cite{Poulin2011}. On the other hand, generic multiqudit states of $N>4$ $d$-dimensional subsystems cannot be transformed into nor be obtained from any inequivalent multipartite states of the same dimensions via LOCC  \cite{GoKr17, SaWa18} \footnote{Here and in the following we do not consider local unitary transformations, as they can always be applied locally and as they do not alter the entanglement contained in the system.}. This shows that investigations can be focused on transformations among states within a non-generic subset of physically relevant, i.e. naturally occurring in certain physical contexts, multipartite quantum states.

The starting point of the investigation of entanglement is the characterization of SLOCC classes: An $N$-partite state $\Psi$ is SLOCC--equivalent to a state $\Phi$ if there exist local invertible operators $g_j$ ($j=1,\ldots,N$) such that $|\Psi\rangle = \otimes_{j=1}^N g_j |\Phi\rangle$. Physically, this means that one can obtain $\Phi$ with a finite probability (i.e., for certain measurement outcomes) by applying  only local generalized measurements on the state $\Phi$. As SLOCC--inequivalent states are not related to each other via local operations, their entanglement, viewed as a resource cannot be compared. This is why, in the context of entanglement theory, the study of state transformations and SLOCC classes is central \cite{Horodecki2009}. 
Whereas the characterization of SLOCC classes becomes intractable with increasing the number of constituents and there exist in general infinitely many classes, the problem has been solved for various physically relevant sets of states. For instance, SLOCC classes have been characterized for symmetric states \cite{Bastin2015, MiRo13}, i.e. those that are invariant under permutations, or for certain tripartite ($N=3$) and four-partite ($N=4$) states \cite{Dur2000, Verstraete2002, Briand2004,Chitambar2010}. 

In Ref. \cite{SaMo19} we presented a systematic investigation of state transformations of translationally invariant Matrix Product States (MPS) with periodic boundary conditions. This family of states is physically relavant, as well as mathematically tractable. In physics, MPS efficiently describe the ground state of local, gapped Hamiltonians \cite{Hastings2007}, as well as critical systems. They also correspond one-to-one to the states that are prepared in the context of sequential generation \cite{Schon2005}, where one system sequentially interacts with a set of subsystems originally in a product state. From the mathematical point of view, they admit an efficient description in terms of $N$ tripartite (fiducial) states, or a single one if the state is translationally invariant. In Ref. \cite{SaMo19} we showed how local transformations and SLOCC classes of translationally invariant MPS can be characterized. We demonstrated that these properties can be inferred from the corresponding properties of the fiducial states and certain cyclic structures of operators acting on these fiducial states. We showed that these properties can be highly size dependent and revealed many interesting features of prominent many-body states, such as the cluster \cite{Raussendorf2001} and the AKLT state \cite{Affleck1987}. The methods introduced in \cite{SaMo19} can also be used to identify all local symmetries of MPS (not only corresponding to global unitary operators \cite{Sanz2009, Singh2010})\footnote{By local (global) symmetries we mean (in-)homogeneous symmetries, i.e., that a different (the same) action operates on each physical system.}. Such a characterization  induces a classification of zero temperature phases of matter \cite{Chen2011,Schuch2011, Pollmann2012}. 

Whereas we provided in \cite{SaMo19} a complete characterization of local symmetries and SLOCC classes of translationally invaritant MPS with bond dimension $D=2$, we extend these results here to the case $D=3$. Interestingly, this  increment leads to very different local properties of the corresponding MPS. This is not only seen in the local symmetries, but also in the possible SLOCC transformations and the SLOCC classes.

\subsection{Outline}
We summarize our results in Section \ref{sec:summary}, where we also introduce some notation. In Section \ref{sec:prelim} we discuss preliminaries. In the subsequent sections we characterize the symmetries of normal translationally invariant MPS (TIMPS) with physical dimension 2 and bond dimension $D=3$. There exist six SLOCC classes for the fiducial state in this case. We focus on two of them featuring considerably contrasting properties. In Section \ref{sec:M10M11M1inf} we discuss a fiducial state for which only a discrete number of operators $g$ acting on the qubit give rise to a symmetry of the state. In Section \ref{sec:L1L1T} we discuss a fiducial state for which any operator $g$ acting on the qubit may give rise to a symmetry of the state. We characterize the symmetries of the generated (normal) MPS for both of them. Regarding the SLOCC classification we scrutinize on the latter SLOCC class of the fiducial states (Section \ref{sec:L1L1T}). All the other classes can be treated similarly. We show that, in contrast to MPS with bond dimension $2$, a much larger variety of local and global symmetries occur in this class. Whereas more local symmetries should enable more SLOCC transformation, we show that this is not the case. In fact, we show that any SLOCC transformation which is possible within those states is realizable with a global SLOCC operation. Finally, in Section  \ref{sec:2DDdiagonal} we study fiducial states represented by diagonal matrix pencil for a bond dimension $D \geq 3$.
In Appendices \ref{app:M10M11M1infcycles} to \ref{app:L1L1TSLOCC} and \ref {sec:L1L1Ttable} we present additional details on the concepts used, as well as proofs of claims made in Sections \ref{sec:M10M11M1inf} and \ref{sec:L1L1T}.
In Appendix \ref{sec:otherSLOCCclasses} we discuss symmetries of MPS associated to fiducial states belonging to one of the four remaining SLOCC classes.

\section{Summary of results}
\label{sec:summary}

Before we summarize our findings, we introduce the following notation and basic concepts, which are needed in order to formulate our results. MPS are multipartite states defined in terms of three-partite tensors. We denote the physical dimension of the MPS by $d$ and the bond dimension by $D$. Given a rank-three tensor $A$ with respective index ranges $0\ldots d-1$, $0\ldots D-1$, and $0\ldots D-1$ we often write
\begin{equation}
 \label{Eq:fid}
  \ket{A} =  \sum_{i=0}^{d-1}\sum_{\alpha,\beta=0}^{D-1} A^{i}_{\alpha \beta} \ket{i} \otimes \ket{\alpha} \otimes \ket{\beta}.
\end{equation}
An MPS on $N$ subsystems is then defined in terms of the (in general site-dependent) tensors $A_i$ ($i \in \{0, \ldots, N-1\}$) by
\begin{equation}
    \label{eq:mps}
  \ket{\Psi} = \sum_{i_0,\ldots,i_{N-1}} \tr \left( A_0^{i_0} \dots A_{N-1}^{i_{N-1}} \right) \ket{i_0 \dots i_{N-1}}.
\end{equation}
If $A = A_0 = \ldots = A_{N-1}$ the MPS is translationally invariant (TI). In the TI case we may call $\ket{\Psi} = \ket{\Psi(A)}$ the MPS generated by the fiducial tripartite state $A$. In fact, $A$ generates a whole family of MPS of arbitrarily large system size, $N$.

Let us briefly recall some concepts from entanglement theory that are relevant for this work. Two $N$-partite states $\ket{\psi}$, $\ket{\phi}$ are said to be local unitary (LU) equivalent ($\ket{\psi} \sim_{LU} \ket{\phi}$) if there exists a local unitary operator $u = u_1 \otimes \ldots \otimes u_N$ such that $u\ket{\psi} = \ket{\phi}$. If $\ket{\psi}$ can be transformed into $\ket{\phi}$ via LOCC with finite probability of success, this transformation is said to be possible via stochastic LOCC (SLOCC) and we write $\ket{\psi} \to \ket{\phi}$.  Note that $\ket{\psi} \rightarrow \ket{\phi}$ holds if and only if there exists a local operator $g = g_0 \otimes \ldots \otimes g_{N-1}$ such that $\ket{\phi} = g\ket{\psi}$ \cite{Bennett2000}. If $\ket{\psi} \to \ket{\phi}$ and $\ket{\phi} \to \ket{\psi}$ the two states are said to be SLOCC equivalent. This is the case if and only if there exists an invertible local operator $g$ such that $\ket{\phi} = g\ket{\psi}$ \cite{Dur2000}. The corresponding equivalence classes are called SLOCC classes. We denote the group of local symmetries\footnote{Note that we consider here only invertible local symmetries. That is, we do not characterize those local operators, which annihilate the state. In the cases investigated here, such symmetries contain necessarily two rank one projectors.} of a state $\ket{\psi}$ by
\begin{equation*}
\mathcal{S}_{\ket{\psi}} = \{S: S \ket{\psi} = \ket{\psi}, S = S_0 \otimes S_1 \otimes \ldots \},
\end{equation*}
where $S_i\in {GL}({d_i},\C)$ and $d_i$ denotes the local dimensions of $\ket{\psi}$.

We will focus here on normal (for definition see Section \ref{sec:prelim}) Translationally Invariant MPS (TIMPS) with physical dimension $d=2$ and bond dimension $D=3$ and discuss the higher bond dimensional case in  Section \ref{sec:2DDdiagonal}. We study the local symmetries of MPS, i.e., for MPS $\ket{\Psi}$ we characterize the set $\mathcal{S}_{\ket{\Psi}}$. Moreover, we study the SLOCC classes of MPS. Note that both, the local symmetries as well as the SLOCC classification might depend on the number of subsystems. 

In Ref \cite{SaMo19} we showed that the local symmetries, possible SLOCC transformations, and the SLOCC classes of normal MPS, $\ket{\Psi(A)}$, are determined by certain cyclic structures of operators that are solely defined by its fiducial state, $\ket{A}$. More precisely, they are determined by the properties of the symmetry group of the fiducial state $\mathcal{S}_{\ket{A}}$, which we also denote by
 \be \label{eq:summaryGA}
 G_A =\{ h=g\otimes x\otimes y^T \ | \ h|A\rangle =|A\rangle \},
 \ee
where $T$ denotes the transpose in the standard basis. Given $G_A$, the only operators which can occur in a symmetry of the normal MPS are those, which act on the qubit system, i.e. the operators $g$ in Eq. (\ref{eq:summaryGA}). We call these symmetries the qubit symmetries. The symmetries of normal MPS are determined by specific properties of the symmetries of the fiducial state, so-called \emph{cycles} (see \cite{SaMo19} and Section \ref{sec:prelim}). To give an example, a symmetry of the form $g_0\otimes g_1 \otimes g_0\otimes g_1 \otimes g_0 \ldots \in \mathcal{S}_{\ket{\psi}}$ stems from a 2-cycle in the symmetry group of $\ket{A}$. It should be noted that the concept of cycles allows to characterize the full symmetry group of normal MPS. For non-normal MPS, the concept might yield only a subgroup of the symmetry group. Whether a TIMPS is normal or not depends only on properties of the fiducial state.

Note that if the two fiducial states generating normal MPS are SLOCC inequivalent, then the same holds for the MPS. 
The fiducial states corresponding to the MPS of interest ($D=3$) can be divided into six distinct SLOCC classes \cite{Chitambar2010}. We focus here on two of them, which show considerably contrasting properties with respect to their symmetry groups. As the two classes we focus on can be considered as the two extreme cases,  all the other classes can be treated similarly. The two considered SLOCC classes of the fiducial states are represented by (the notation used here will become clear afterwards): 
\begin{enumerate}[(i)]
\item  $\ket{ M(\omega) } = \ket{0}(\ket{00}+ \omega \ket{11}+ \omega^2\ket{22}) + \ket{1}(\ket{00}+\ket{11}+\ket{22})$, where $\omega = e^{i\frac{2 \pi}{3}}$ (see Section \ref{sec:M10M11M1inf})
\item $\ket{LLT} = \ket{0}(\ket{01} + \ket{22}) +  \ket{1}(\ket{00} + \ket{12})$, (see Section \ref{sec:L1L1T}).
\end{enumerate}

\begin{table*}[!t!]
\begin{tabular}{ |>{\centering\arraybackslash}m{0.48\linewidth}   |   >{\centering\arraybackslash}m{0.48\linewidth} | }
 
  \multicolumn{2}{c}{Considered SLOCC classes of the fiducial state}\\ \hline
 $M(\omega)$ & $LLT$  \\ \hline
  \multicolumn{2}{c}{Representatives}\\ \hline
  $\ket{0}(\ket{00}+ \omega \ket{11}+ \omega^2\ket{22}) + \ket{1}(\ket{00}+\ket{11}+\ket{22})$ & $\ket{0}(\ket{01} + \ket{22}) +  \ket{1}(\ket{00} + \ket{12})$ \\\hline
  \multicolumn{2}{c}{Symmetries of the fiducial state}\\ \hline
 discrete set of 6 possible operators $g$ appearing on physical site &  3-parametric family, any operator $g$ acting on the physical site gives rise to a symmetry \\ \hline
 $g$'s are unitary & $g$'s are non-compact\\ \hline
   \multicolumn{2}{c}{Possible (minimal) cycle lengths} \\ \hline
 1, 2, 3, 4, 6 & $N$ for any $N \in \mathbb{N}$ \\\hline
   \multicolumn{2}{c}{Symmetries of normal MPS} \\ \hline
unitary (finite) & both unitary (finite) and non-unitary (non-compact) \\ \hline
 diagonalizable (as unitary) & both diagonalizable and non-diagonalizable \\ \hline
   \multicolumn{2}{c}{SLOCC among normal MPS} \\ \hline
 \begin{itemize}
 \item examples of SLOCC equivalence via global operations
 \item examples where non-global operations required $\rightarrow$ equivalence is $N$-dependent \end{itemize} & global operations suffice $\rightarrow$ equivalence is not $N$-dependent \\ \hline
   \multicolumn{2}{c}{Results for non-normal MPS} \\ \hline
   no other than the already identified cycles appear among the non-normal MPS & The well-known Majumdar--Ghosh states appear as a special case, certain permutation-invariant states appear as another special case.
   We find SLOCC equivalences through non-global, but not global, operations (in contrast to normal MPS)\\ \hline
   \multicolumn{2}{c}{Comparison to $D=2$ (GHZ- and W- fiducial states, see \cite{SaMo19})} \\ \hline
   GHZ-states: symmetry group of order $1$, $2$, or $2^N$ & W-states: 2-cycle leading to 1-(complex)-parametric symmetry group\\ \hline
\end{tabular}
\caption{Highlights of the obtained results. Out of the six possible SLOCC classes for fiducial states with $D=3$, we consider two representatives with considerably contrasting properties. The representative of the class  $M(\omega)$ (see left column) leads to only six different, unitary, operators appearing as qubit symmetries. The representative of the class $LLT$ has a symmetry group in which any operator $g$ appears as a qubit symmetry. Note that the properties of the remaining four representatives lie between those two extremal cases, as do fiducial states for $D=2$. In this table we summarize and illustrate how the substantial differences in the symmetry groups lead to considerably contrasting properties in the generated MPS regarding the exhibited symmetries and SLOCC equivalence. We also compare to properties of ($D=2$)-MPS generated by W- and GHZ-states (left and right bottom). } 
\label{tab:summary}
\end{table*}

The SLOCC class represented by some state $\ket{A}$ is given by $a\otimes b\otimes c \ket{A}$, for any invertible operators $a,b,c$. In order to determine both, the symmetries as well as the SLOCC classes of all MPS which correspond to a fiducial state belonging to the SLOCC class represented by $\ket{A}$, it suffices to consider fiducial states of the form $\one \otimes b\otimes \one \ket{A}$ only \cite{SaMo19}. 

The main result of the present article is a full characterization of the local symmetries of normal MPS generated by fiducial states within these two SLOCC classes, as well as a full characterization of SLOCC equivalence among normal MPS corresponding to case (ii)---we only outline the procedure for case (i), as the SLOCC classification is much simpler in that case. In the case (i), the set of qubit symmetries are finite and unitary (in fact, they form a unitary representation of the symmetric group $S_3$)~\footnote{Note that in case there exist only finitely many symmetries, one can always chose a representative of the SLOCC class whose symmetry is unitary.}. In case (ii), any operator $g$ acting on the qubit system, i.e. on the physical system, leads to a symmetry of the fiducial state. Stated differently, there is as much freedom in the qubit symmetries as there could possibly be. A goal of this work is to illustrate how the contrasting properties of the two considered classes of fiducial states manifest also in the properties of the associated MPS. We highlight this in form of a comparison of selected properties of the MPS in Table \ref{tab:summary}. The properties of the symmetry group of the remaining four representatives of the fiducial states lie between the two considered extreme cases. Note that the symmetry groups of the fiducial states for $D=2$, the GHZ- and the W-state exhibiting 1- and 2-parametric qubit symmetries (we count complex parameters here and throughout the remainder of the article),  respectively, do not draw near the here considered extremal cases, either. This is reflected in the properties of the generated MPS (see Table \ref{tab:summary}).
The outlined main results are accompanied by several results on non-normal MPS, a discussion on bond dimension $D>3$, and results on MPS associated to the remaining SLOCC classes of fiducial states for $D=3$ (see Sections \ref{sec:L1L1Tnonnormal}, \ref{sec:2DDdiagonal}, and Appendix \ref{sec:otherSLOCCclasses}).

Let us now discuss the results obtained for the two respective fiducial states. 
In the case (i), $\ket{ M(\omega) }$, as mentioned above, the set of qubit symmetries is finite and unitary. This leads to finitely many local symmetries of the corresponding MPS (for arbitrary system size). We determine all possible local symmetries (see Table \ref{tab:M10M11M1infcycles}). All minimal cycles are of length $1,2,3,4,6$.
Furthermore, we characterize all fiducial states which generate normal MPS and which lead to the symmetries presented in Table \ref{tab:M10M11M1infcycles}. Despite the fact that deciding whether a family of states is normal or not can be cumbersome, a complete parametrization of the corresponding fiducial states, i.e. the operators $b$ can be found in Appendix \ref{app:M10M11M1infcycles}. 

Given that the set $G_A$ is finite it is straightforward to determine all possible SLOCC transformations. Stated differently, it is easy to determine which pairs of fiducial states lead to MPSs $\ket{\Psi}$ and $\ket{\Phi}$ such that $\ket{\Psi} \propto g_0\otimes g_1 \ldots \otimes g_{N-1} \ket{\Phi}$, where $g_i$ denote local regular matrices. We outline how this question can be answered and how the SLOCC classification can be derived. In fact, many potential SLOCC transformations can be ruled out due to the incompatibility of the symmetry groups of the corresponding states. Within the remaining cases we provide examples of states which are not SLOCC--equivalent, those which are via global transformations and those, which require non--global transformations. 

The reason for investigating the case (ii), i.e. TIMPS which correspond to fiducial states within the SLOCC class represented by $\ket{LLT}$, is that the properties of the fiducial states are in stark contrast to the previously considered case. Thus, also the properties of the corresponding MPS can be expected to be. In fact, in this class, the fiducial state has infinitely many qubit symmetries. As mentioned above, actually, any operator $g$ acting on the physical system leads to a symmetry of the fiducial state. This is why the study of this class is particularly interesting. As we show, this leads to a huge variety of local symmetries of the corresponding MPS. Any possible cycle length actually appears, moreover, finite as well as infinite symmetry groups emerge. Among them, we find diagonalizable as well as non-diagonalizable symmetries. A complete characterization of all possible symmetries for normal MPS is presented in Figure \ref{fig:L1L1Tflowchart}. Although such  large symmetry groups could lead one to believe that there are many possible SLOCC transformation, we show that, surprisingly, the opposite is true. All possible SLOCC transformations can be performed via global operations, i.e. for any two states, $\ket{\Psi}$ and $\ket{\Phi}$, which are SLOCC equivalent, there exits a global operator $g^{\otimes N}$ such that $\ket{\Psi} \propto g^{\otimes N} \ket{\Phi}$.
Note that, in contrast to case (i) and also the case of TIMPS with bond dimension $D=2$ MPS generated by GHZ states (see \cite{SaMo19}), this implies that the existence of a SLOCC transformation among two TIMPS does not depend on the system size. 
In contrast to case (i), the characterization of the SLOCC classes is more challenging due to the large symmetry groups. We provide a parametrization of representatives of each SLOCC class in Figure \ref{fig:L1L1TflowchartSLOCC}. The fact that two states are in the same SLOCC class if and only if they are related to each other via a global operation leads to a huge variety of SLOCC classes.

As mentioned above, only for normal MPS it holds that the whole local symmetry group of the TIMPS can be determined via the local symmetries of the corresponding fiducial state. For non-normal MPS, the methodology of cycles is still useful, but the determined symmetries might form a subgroup of the symmetry group of the MPS, only. Interestingly, in the case (ii), the whole symmetry group can also be determined for certain non normal MPS. This is due to the fact that those states correspond to permutationally invariant states, for which the local symmetry groups are known \cite{BaKr09,MaKr10,MiRo13}. Moreover, in case (ii) it turns out that certain non-normal MPS are well--known states, the Majumdar--Ghosh states. We show that non--normal MPS show a distinct behaviour compared to  normal ones (see Section \ref{sec:L1L1Tnonnormal}). Not only do we provide examples of non--normal states which possess a much larger symmetry group than the one determined by the fiducial state, but also examples of states which are SLOCC--equivalent, but cannot be transformed into each other via a global transformation, in contrast to all normal MPS corresponding to case (ii). 

\subsection{Particularly interesting states}
\label{sec:examples}

In this subsection we present selected fiducial states, which generate MPS with particularly interesting properties with respect to their symmetry group. In light of these examples it is evident that bond dimension $D=3$ allows for much more diverse symmetry groups than is the case for bond dimension $D=2$ \cite{SaMo19}. For the examples presented here we use a notation to label the fiducial states reflecting the properties of the symmetry groups as follows. We use the notation $\ket{1/2 \,G_{(\infty)} \, D^{l_1,l_2,\ldots}_{(\infty)}}$ to label a fiducial state (primarily) according to the properties of the symmetry group of the generated MPS. The notation should be read in three parts: "$1/2$", "$G_{(\infty)}$", and "$D^{l_1,l_2,\ldots}_{(\infty)}$". It should be understood as follows. The first part, "$1/2$", should read either "1", or "2", and indicates whether the fiducial state belongs to the SLOCC class represented by $\ket{M(\omega)}$ (in case of "1"), or  $\ket{LLT}$ (in case of "2"). The second part, "$G_{(\infty)}$", describes the global symmetries of the generated MPS. It should read "$G$" ("$G_\infty$"), in case finitely many (infinitely many) nontrivial global symmetries are present. The third part, "$D^{l_1,l_2,\ldots}_{(\infty)}$", describes the local symmetries of the MPS. As before, the presence of the subscript $\infty$ indicates whether there are finitely of infinitely such symmetries. Moreover, the integers $l_1,l_2, \ldots$ are used to indicate that the local part of the symmetry group changes depending on whether $l_1, l_2, \ldots$ divide the particle number $N$. In case the MPS does not posses any nontrivial local (any nontrivial global) symmetry, we simply omit the "$D$" ("$G$") part. Clearly, the naming scheme does not allow to unambiguously identify MPS, but it suffices to distinguish the examples considered here.

\subsubsection{Examples of fiducial states within the SLOCC class represented by $\ket{ M(\omega) }$}
The first example is given by the fiducial state $\ket{1 \, G}=  \ket{012} + \ket{021} + \ket{101} +\ket{102} + \omega (\ket{001} +\ket{010} +\ket{110}+\ket{112}) + \omega^2 (\ket{002} +\ket{020} +\ket{120}+\ket{121})$.
The corresponding tensor $A$ reads
\begin{equation*}
    A^0 =  
    \begin{pmatrix} 
      0 & \omega & \omega^2 \\ 
      \omega & 0 & 1 \\ 
      \omega^2 & 1 & 0
    \end{pmatrix}   
    \quad \text{and} \quad  
    A^1 =  
    \begin{pmatrix} 
      0 & 1 & 1 \\ 
      \omega & 0 & \omega \\ 
      \omega^2 & \omega^2 & 0
    \end{pmatrix}  .
\end{equation*}
The generated MPS $\ket{\Psi_{1 \, G}}$ has global symmetries only. The symmetry group is finite and unitary and given by $g^{\otimes N}$, where $g$ is any of the six operators generated by $\sigma_x = \begin{pmatrix} 0&1 \\1&0\end{pmatrix}$ and $\operatorname{diag}(\omega,1)$.

The second example is given by the fiducial state $\ket{1\, G \, D^3} = \ket{000} + \ket{002} + \ket{012} +\ket{100}+\ket{101} +\ket{121} + \omega (\ket{001} + \ket{021} +\ket{022}+\ket{102} +\ket{111} +\ket{112} ) +\omega^2 (\ket{010} + \ket{011} +\ket{020} +\ket{110} +\ket{120} +\ket{122})$. The corresponding tensor $A$ reads
\begin{equation*}
    A^0 =  \begin{pmatrix} 1 & \omega & 1 \\ \omega^2 & \omega^2 & 1 \\ \omega^2 & \omega & \omega \end{pmatrix}   \quad \text{and} \quad  A^1 =  \begin{pmatrix} 1 & 1 & \omega \\ \omega^2 & \omega & \omega \\ \omega^2 & 1 & \omega^2 \end{pmatrix}  .
\end{equation*}
The corresponding MPS $\ket{\Psi_{1\, G \, D^3}}$ has the same global symmetries as the first example. If the particle number of the MPS, $N$, is no multiple of 3, these are the only symmetries and the symmetry group is thus identical to the first example. If, however, $N$ is a multiple of 3, then additional local symmetries emerge. Then, the symmetry group is comprised of 18 elements and is generated by repeating sequences (what we will later on call cycles) of $g_0 \otimes g_1 \otimes g_2$, as well as by repeating sequences of $g_0 \otimes g_2 \otimes g_1$ (and translations thereof), where $g_0 = \sigma_x$, $g_1 = \operatorname{diag}(\omega,1)\sigma_x$, and $g_2 = \operatorname{diag}(\omega,1)^2\sigma_x$.

The third example is given by the fiducial state $\ket{1\, G  \, D^{2,6}} = \ket{000} + \ket{002} + \ket{022} +\ket{100}+\ket{101} +\ket{111} + \omega (\ket{001} + \ket{011} +\ket{012}+\ket{102} +\ket{121} +\ket{122} ) +\omega^2 (\ket{010} + \ket{020} +\ket{021} +\ket{110} +\ket{112} +\ket{120})$. The corresponding tensor $A$ reads
\begin{equation*}
    A^0 =  \begin{pmatrix} 1 & \omega & 1 \\ \omega^2 & \omega & \omega \\ \omega^2 & \omega^2 & 1 \end{pmatrix}   \quad \text{and} \quad  A^1 =  \begin{pmatrix} 1 & 1 & \omega \\ \omega^2 & 1 & \omega^2 \\ \omega^2 & \omega & \omega \end{pmatrix}.  
\end{equation*}
The generated MPS $\ket{\Psi_{1\, G  \, D^{2,6}}}$ possesses the global symmetry $g^{\otimes N}$ for $g = \sigma_x$. In case of an even particle number $N$, the MPS possesses local symmetries and the symmetry group is generated by repeating sequences of $g_0 \otimes g_1$, where $g_0 = \operatorname{diag}(\omega,1)\sigma_x$ and $g_1 = \operatorname{diag}(\omega,1)^2\sigma_x$. Moreover, in case that $N$ is divisible by 6, additional local symmetries emerge, repeating sequences of $g_0 \otimes g_1 \otimes g_2 \otimes g_3 \otimes g_4 \otimes g_5$, where $g_0 = g_5 = \operatorname{diag}(\omega,1)\sigma_x$,  $g_1 = g_4 = \sigma_x$, and $g_2 = g_3 = \operatorname{diag}(\omega^2,1)\sigma_x$ generate the symmetry group, then.  
This example is particularly interesting as the MPS $\ket{\Psi_{1\, G  \, D^{2,6}}}$ is SLOCC equivalent to the previous example, $\ket{\Psi_{1\, G \, D^3}}$, for an even particle number $N$. However, this is not the case if $N$ is odd. This leads to the following possible situations.  If 2 divides $N$, but 3 does not, then $\ket{\Psi_{1\, G \, D^3}}$ and $\ket{\Psi_{1\, G  \, D^{2,6}}}$ are SLOCC equivalent and their symmetry groups are of order 6. If both 2 and 3 divide $N$, then $\ket{\Psi_{1\, G \, D^3}}$ and $\ket{\Psi_{1\, G  \, D^{2,6}}}$ are SLOCC equivalent and their symmetry groups are of order 18. If 3 divides $N$, but 2 does not, then $\ket{\Psi_{1\, G \, D^3}}$ and $\ket{\Psi_{1\, G  \, D^{2,6}}}$ are not SLOCC equivalent. Remarkably, the orders of the corresponding symmetry groups differ. The symmetry group of $\ket{\Psi_{1\, G \, D^3}}$ is of order 18, while the symmetry group of $\ket{\Psi_{1\, G  \, D^{2,6}}}$ is of order 2. Finally, if neither 2, nor 3 divides $N$, then $\ket{\Psi_{1\, G \, D^3}}$ and $\ket{\Psi_{1\, G  \, D^{2,6}}}$ are not SLOCC equivalent, again the orders of the symmetry groups differ. The symmetry group of $\ket{\Psi_{1\, G \, D^3}}$ is of order 6, while the symmetry group of $\ket{\Psi_{1\, G  \, D^{2,6}}}$ is of order 2.

The fourth example is given by the fiducial state $\ket{1\,D^3} = \ket{012} - \ket{021}+\ket{101} + \ket{102} + \omega (\ket{001}-\ket{010} +\ket{022} -\ket{110} -\ket{111} + \ket{112} ) +\omega^2 (\ket{002} -\ket{011}+\ket{020} +\ket{120}-\ket{121}+\ket{122})$. The corresponding tensor $A$ reads
\begin{equation*}
    A^0 =  \begin{pmatrix} 0 & \omega & \omega^2 \\ -\omega & -\omega^2 & 1 \\ \omega^2 & -1 & \omega \end{pmatrix}   \quad \text{and} \quad 
    A^1 =  \begin{pmatrix} 0 & 1 & 1 \\ -\omega & -\omega & \omega \\ \omega^2 & -\omega^2 & \omega^2 \end{pmatrix}  .
\end{equation*}
The generated MPS $\ket{\Psi_{1\,D^3}}$ does not exhibit any global symmetry (except the trivial symmetry), however, it does possess local symmetries if $N$ is a multiple of 3. Its symmetry group then is of order 4 and is generated by repeating sequences of $g_0 \otimes g_1 \otimes g_2$ (and translations thereof), where $g_0 = \identity$ and $g_1 = g_2 = \sigma_x$. 

Additional possible symmetry groups (in fact, all possible symmetry groups for normal MPS generated by a fiducial state that is in the SLOCC class of $\ket{ M(\omega) }$) are displayed in Figure \ref{fig:sets}. Examples of MPS exhibiting the corresponding symmetry groups may be easily constructed with the help of Table \ref{tab:M10M11M1infcycles}.

\subsubsection{Examples of fiducial states within the SLOCC class represented by $\ket{LLT}$}
The fifth example is given by the fiducial state $\ket{2\, G_\infty}= \ket{002} - \ket{012} + \ket{021} + \ket{022} -  \ket{112} + \ket{120}$. The corresponding tensor $A$ reads
\begin{equation*}
    A^0 =  \begin{pmatrix} 0 & 0 & 1 \\ 0 & 0 & -1 \\ 0 & 1 & 1\end{pmatrix}  
    \quad \text{and} \quad
    A^1 =  \begin{pmatrix} 0 & 0 & 0 \\ 0 & 0 & -1 \\ 1 & 0 & 0\end{pmatrix}  .
\end{equation*}
The generated MPS $\ket{\Psi_{2\, G_\infty}}$ has a one-parametric symmetry group (counting complex parameters) of global symmetries $g^{\otimes N}$, where $g = \begin{pmatrix} 1 &x \\ 0 & 1\end{pmatrix}$ for any $x\in \mathbb{C}$.

The sixth example is given by the fiducial state $\ket{2\, D^m_\infty} = e^{i\frac{\pi}{m}}\ket{002} + \ket{021} + \ket{022}  - e^{-i\frac{\pi}{m}} \ket{112} +\ket{120}$ for some $m \in \mathbb{N}$. The corresponding tensor $A$ reads
\begin{equation*}
    A^0 =  \begin{pmatrix} 0 & 0 & e^{i\frac{\pi}{m}} \\ 0 & 0 & 0 \\ 0 & 1 & 1\end{pmatrix}   
    \quad \text{and} \quad
    A^1 =  \begin{pmatrix} 0 & 0 & 0 \\ 0 & 0 & -e^{-i\frac{\pi}{m}} \\ 1 & 0 & 0\end{pmatrix} . 
\end{equation*}
The generated MPS $\ket{\Psi_{2\, D^m_\infty}}$ has a non-trivial symmetry group only if the particle number $N$ is a multiple of $m$. Then, the MPS exhibits a one-parameter symmetry group of local symmetries $g_0 \otimes g_1 \otimes \ldots$, where $g_k = \begin{pmatrix}1 & x e^{i \frac{2 k \pi}{m}}\\ 0&1 \end{pmatrix}$, $x \in \mathbb{C}$.

The seventh example is given by the fiducial state $\ket{2\, G  \, D^{2}_\infty} = i \ket{002} + \ket{021} + \ket{022}  + i \ket{112} +\ket{120} + \ket{122}$. The corresponding tensor $A$ reads
\begin{equation*}
    A^0 =  \begin{pmatrix} 0 & 0 & i \\ 0 & 0 & 0 \\ 0 & 1 & 1\end{pmatrix}   
    \quad \text{and} \quad
    A^1 =  \begin{pmatrix} 0 & 0 & 0 \\ 0 & 0 & i \\ 1 & 0 & 1\end{pmatrix}  .
\end{equation*}
In case of an odd particle number $N$, the generated MPS $\ket{\Psi_{2\, G  \, D^{2}_\infty}}$ has a single non-trivial symmetry $g^{\otimes N}$, where $g=i \sigma_x$. In case of an even particle number the MPS exhibits a continuous symmetry group with symmetries $g_0 \otimes g_1 \otimes g_0 \otimes g_1 \otimes \ldots$, where $g_0=\begin{pmatrix}1 & x \\ x & 1\end{pmatrix}$ and  $g_1 = g_0^{-1}$ for any $x \in \mathbb{C} \setminus \{1\}$.

As an eighth example we give the well-known Ma\-jum\-dar-Ghosh states, which are non-normal MPS possessing the symmetry group $g^{\otimes N}$ for any $g$. They appear as a particular case within the SLOCC class of fiducial states represented by $\ket{LLT}$ (see Section \ref{sec:L1L1Tnonnormal}). One such state may be obtained considering the fiducial state $\ket{A} = \ket{002} +\ket{021} -\ket{112} + \ket{120}$. The corresponding tensor $A$ reads
\begin{equation*}
    A^0 =  \begin{pmatrix} 0 & 0 & 1 \\ 0 & 0 & 0 \\ 0 & 1 & 0\end{pmatrix}   
    \quad \text{and} \quad
    A^1 =  \begin{pmatrix} 0 & 0 & 0 \\ 0 & 0 & -1 \\ 1 & 0 & 0\end{pmatrix}  .
\end{equation*}

A list of all representatives of normal MPS generated by fiducial states within the $\ket{LLT}$ class is presented in Figure \ref{fig:L1L1TflowchartSLOCC}, the emerging symmetry groups of the generated MPS are displayed in Figure \ref{fig:L1L1Tflowchart}.

We briefly discuss the remaining SLOCC classes of the fiducial states in Appendix \ref{sec:otherSLOCCclasses}. Moreover, we show that for generic fiducial states of higher bond dimension, there exists no non--trivial local symmetry. That is the set $S_{\ket{
\Psi}}=\{\one\}$, which also implies that the SLOCC classification is trivial. 

\section{Preliminaries}

\label{sec:prelim}
In this section we review relevant concepts from  the theory of MPS (Section \ref{sec:revMPS}) and of Ref. \cite{SaMo19} (Section \ref{sec:revMPSsym}).

\subsection{Matrix Product States}
\label{sec:revMPS}

As injective and normal MPS (not necessarily TI) play a crucial role in the theory of MPS, we recall here their definitions.
\begin{definition}
  A MPS tensor $A$ is injective if the following map is injective:
  \begin{equation*}
    X \mapsto \sum_i \tr \left(X A^i\right) \ket{i}.
  \end{equation*}
  A MPS is injective if all the defining tensors are injective. An MPS is normal if there exists an $L$ such that the contraction of any $L$ consecutive tensors are injective, i.e.,\ the following maps are injective:
  \begin{equation*}
    X \mapsto \sum_j \tr \left(X A_i^{j_1} A_{i+1}^{j_2} \dots A_{i+L-1}^{j_L}\right) \ket{j_1 \dots j_L}.
  \end{equation*}
\end{definition}
In the following we only consider MPS with $N \geq 2L+1$, if not stated differently. In this case the following Fundamental Theorem of MPS characterizes when two normal MPS generate the same state.

\begin{theorem}[Fundamental Theorem of MPS \cite{Molnar2018}]
\label{thm:fund}
 Two normal MPS given by tensors $A_0,\ldots,A_{N-1}$ and $B_0,\ldots,B_{N-1}$ generate the same state $\Psi$ iff there exist regular matrices $x_0,\ldots,x_{N-1}$ such that $A_k^j = x_k^{-1}B_k^jx_{k+1}$ for all $k$ and $j$, with $x_{N} \equiv x_0$; that is, iff
\begin{align}
\ket{A_k} = \one \otimes x_k^{-1} \otimes x_{k+1}^T\ket{B_k} \ \forall \ k. \label{eq:fundtheom}
\end{align}
The matrices $x_0,\ldots,x_{N-1}$ are unique up to a multiplicative constant.
\end{theorem}

Whenever we refer to MPS in the remainder of this paper we refer to normal TI MPS, if not stated differently. We call
${\cal N}_{N,D}$ the set of normal, translationally invariant MPS
with bond dimension $D$ and $N \geq 2L + 1$ sites.

\subsection{Review of results on symmetries and local transformations of MPS}
\label{sec:revMPSsym}

As mentioned before, the local symmetries of a normal MPS, $\Psi(A)$, are determined by certain cyclic structures of operators that are solely defined by its fiducial state, $A$ \cite{SaMo19}, $G_A$ (see Eq. (\ref{eq:summaryGA}). Two operators $h_0,h_1\in G_A$ with $h_i=g_i\otimes x_i\otimes y_i^T$ can be concatenated, denoted as $h_0\to h_1$, if $y_0 x_1 \propto \one$. A sequence $\{h_i\}_{i=0}^{k-1} \subseteq G_A$ of $k$ elements in $G_A$ with
 \be
 \label{concat}
 h_0\to h_1\to \ldots \to h_{k-1}\to h_0
 \ee
 is called a $k$-cycle. More explicitly, we have that the sequence $\{h_i\}_{i=0}^{N-1} \subseteq G_A$ of $N$ elements in $G_A$ form a $N$--cycle if the following conditions hold for any $0\leq k \leq N-2$: 
 \bea \label{eq:ConcatSym}
  y_k x_{k+1} \propto \one \nonumber\\
  y_{N-1} x_{0} \propto \one.
 \eea 
 We showed the following theorem.

\begin{theorem}[\cite{SaMo19}]
The local (global) symmetries of $\Psi(A)\in {\cal N}_{N,D}$ are in one-to-one correspondence with the $N$-cycles (1-cycles) in $G_{A}$.
\label{Thm1}
\end{theorem}
The symmetry of the state corresponding to the cycle $h_0\to h_1\to \ldots \to h_{N-1}\to h_0$ is $g_0\otimes \dots \otimes g_{N-1}$.  Hence, one solely has to determine $G_A$ and find all $N$-cycles in this set to characterize the local symmetries of $\Psi(A)$. In fact, this yields the symmetries of all states in the family of normal MPS generated by $A$. In practice, it is sufficient to characterize all minimal cycles of $G_{A}$ from which all others can be obtained by concatenation. For example, a 3-cycle can always be concatenated with itself to an $N$-cycle if 3 divides $N$. A symmetry of the form $g^{\otimes N}$ is called global. The global symmetries are defined in terms of 1-cycles, and thus require that there is a regular $x$ such that $g\otimes x^{-1} \otimes x^T |A\rangle=|A\rangle$ \cite{SaMo19}.
If $g$ is unitary, this reduces to the well-known characterization of global unitary symmetries of MPS \cite{Sanz2009, Singh2010}. However, minimal cycles of length $N>1$ yield local symmetries of the TI MPS $\Psi(A)$ that are not global and which are generally not considered. 

We often consider fiducial states $\one\otimes b\otimes \one \ket{A}$. The concatenation conditions [see Eq. (\theequation)] then read
 \bea \label{eq:ConcatSymb}
  y_k b x_{k+1} b^{-1} \propto \one, \nonumber\\
  y_{N-1} b x_{0}b^{-1} \propto \one.
 \eea 

In order to characterize SLOCC transformations among normal MPS, one first notices that the corresponding fiducial states need to be SLOCC equivalent. We considered in Ref. \cite{SaMo19} the set
\begin{equation*}
    G_{A,B} = \left\{ h= g\otimes x \otimes y^T \ | \ h \ket{A} = \ket{B} \right\},
 \end{equation*}
As in the case of $G_{A}$ we can define $k$-cycles on $G_{A,B}$. Using the notation $A \Nto B$ if an $N$--partite state $A$ can be transformed via local operations into an $N$--partite state $B$, we proved the following theorem in Ref. \cite{SaMo19}.
\begin{theorem}[\cite{SaMo19}] \label{ThSlocc}$A \Nto B$ with local (global) transformations iff there exists
	 an $N$-cycle
	(1-cycle) in $G_{A,B}$.
\label{Thm2}
\end{theorem}

In this theorem the operators which transform $A$ to $B$ are not necessarily regular. Here, we focus on SLOCC transformations, i.e. on invertible matrices on the physical (as well as the virtual) systems. As shown in \cite{SaMo19}, in order to solve the problem of SLOCC--equivalence (and also the symmetries), it is sufficient to consider fiducial state of the form 
$\ket{A_b}=\one\otimes b\otimes \one \ket{A}$, where $\ket{A}$ denotes a representative of the SLOCC class of the fiducial states.  Let us briefly recall the reason for that. First, it is clear that the two fiducial states corresponding to SLOCC--equivalent normal TIMPS must be SLOCC--equivalent. Second, $g^{\otimes N} \Psi(A)$ is obviously SLOCC--equivalent to $\Psi(A)$ and therefore the operator on the qubit system does not need to be taken into account. And third, due to the fundamental theorem (see Theorem \ref{thm:fund}), an operator on the third system can be mapped to an operator on the second. Clearly, the same argument applies when considering local symmetries. 

According to Theorem \ref{ThSlocc}, two normal TIMPS corresponding to the fiducial states $\ket{A_b}$, $\ket{A_c}$, respectively, are SLOCC--equivalent iff there exists an $N$-cycle in $G_{A_b,A_c}$ (or, equivalently, in $G_{A_c,A_b}$). Using that  $G_{A_b,A_c}= (\one \otimes c\otimes \one)G_A (\one\otimes b^{-1}\otimes \one)$ the existence of such a cycle can be formulated in terms of the symmetries of the fiducial state representing the SLOCC class as follows. The operators $h_0,h_1\in G_{A}$, with $h_i = g_i \otimes x_i \otimes y_i^T$, are called $(b \to c)$-concatenated, if $y_0 b x_1 \propto c$. In this case we write $h_0 \xrightarrow{b\to c} h_1$. A sequence $\{h_i\}_{i=0}^{k-1} \subseteq G_{A}$ is called a $(b\to c)$-$k$-cycle if
  \begin{equation}\label{eq:bc_cycle}
    h_0 \xrightarrow{b\to c} h_1 \xrightarrow{b\to c} \dots \xrightarrow{b\to c} h_{k-1} \xrightarrow{b\to c} h_0.
  \end{equation}
Stated explicitly, $\{h_i\}_{i=0}^{N-1} \subseteq G_{A}$ is a $(b\to c)$-$N$-cycle if the following concatenation rules are fulfilled for any $k$ such that $0\leq k\leq N-2$
  \bea\label{eq:bc_cycleExpl}
  y_k b x_{k+1} \propto c \\
    y_{N-1} b x_{0} \propto c.
  \eea

As in the case of symmetries it might be possible that $G_{A,B}$ only contains $k$-cycles with $k \geq 2$. Then $\Psi(A) \to \Psi(B)$ only holds if $k$ divides $N$ and the corresponding SLOCC operator is not global, i.e., not of the form $g^{\otimes N}$. Note that Theorem \ref{Thm2} was used in \cite{SaMo19} to characterize all SLOCC classes of normal MPS.

\subsection{The fiducial states of TIMPS with physical dimension $d=2$ and bond dimension $D=3$}
\label{sec:revFiducial}

As we reviewed above, the symmetries and the SLOCC classes of TIMPS can be characterized by considering sets of operators, which are determined by the three--partite fiducial states. 

The fiducial states of TIMPS with physical dimension $d$ and bond dimension $D$ are $d\times D\times D$ states. We focus on the case $d=2$ and $D=3$ and discuss extensions of the results presented here in the appendix. For $d=2$ (and arbitrary $D$) the SLOCC classes have been determined using the theory of Matrix Pencils (MPs) ~\cite{Kronecker, Chitambar2010, HeGa18}. As this theory is also well suited to determine the symmetries of the states, we briefly review it here. To the three-partite state given in Eq. (\ref{Eq:fid}) we associate the homogeneous matrix polynomial (matrix pencil),
\begin{align}
\label{def:pencils}
 \mathcal{P}_{A} \equiv P_{(A^0,A^1)} \equiv \mu A^0 + \lambda A^1
\end{align}
where $\mu, \lambda$ are complex variables and $A^0,A^1 \in \C^{D \times D'}$. In Ref. \cite{Kronecker} Kronecker showed that every matrix pencil is strictly equivalent to a matrix pencil in so-called \emph{Kronecker Canonical Form (KCF)}, which is characterized by a set of invariants (e.g., the finite and infinite eigenvalues, eigenvalue size signatures and minimal indices). Stated differently, for each matrix pencil $\mathcal{P}_{A}$ there exist regular matrices, $x,y$ such that $x \mathcal{P}_{A} y$ has canonical form. The KCF together with the results presented in \cite{Chitambar2010, HeGa18} can be used to determine both, the SLOCC classes of the states as well as the symmetries of the state, as we will show in the following.\\

As shown in \cite{Chitambar2010} there is always an operation on the qubit that transforms a $2 \times D \times D$ state $A$ into a state whose matrix pencil has only finite eigenvalues \cite{Chitambar2010}. Furthermore, the operation on the qubit cannot change the structure of the matrix pencil, but only its eigenvalues.

More precisely, the action of 
\begin{equation*} 
  w = \begin{pmatrix}
      \alpha & \beta \\
      \gamma & \delta
     \end{pmatrix} \in GL(2,\C)
 \end{equation*}
on the qubit changes the eigenvalues of the resulting MP from $\{x_i\}$ to 
\begin{equation*}
  x_i'=\frac{\alpha x_i + \beta}{\gamma x_i + \delta}. 
\end{equation*}
  
Using all that it is then easy to see that the six SLOCC classes of $2\times 3\times 3$ systems are represented by the following states: 

\begin{enumerate}[(i)]
\item $\ket{ M(\omega) } = \ket{0}(\ket{00} + \omega \ket{11}+ \omega^2 \ket{22}) + \ket{1}(\ket{00}+\ket{11}+\ket{22})$, with corresponding MP $M^1(1)\oplus M^1(\omega)\oplus M^1(\omega^2)$, i.e., a MP with three distinct eigenvalues 
\item $\ket{D}=\ket{0}(D\otimes \one)(\ket{00}+\ket{11}+\ket{22}) + \ket{1} (\ket{00}+\ket{11}+\ket{22})$, with $D$ a diagonal matrix with degenerate eigenvalues, which corresponds to a MP with degenerate eigenvalues (disregarding biseparable states, there must be one eigenvalue with degeneracy 1 and one eigenvalue with (algebraic and geometric) multiplicity 2.)
\item  $\ket{J} = \ket{0}(J\otimes \one)(\ket{00}+\ket{11}+\ket{22}) + \ket{1} (\ket{00}+\ket{11}+\ket{22})$, with $J$ a non-diagonalizable matrix in Jordan normal form, which corresponds to a MP with degenerate eigenvalues (this case comprises three distinct SLOCC classes)
\item  $\ket{LLT} = \ket{0}(\ket{01} + \ket{22}) +  \ket{1}(\ket{00} + \ket{12})$, with MP $L_1 \oplus L_1^T$. In this case the MP does not have any eigenvalue.
\end{enumerate}

We will mainly focus here on the cases (i) and (iv) and will discuss the symmetries of the remaining cases in Appendix \ref{sec:otherSLOCCclasses}.

Let us already mention here that in order to determine the symmetries in case (i)-(iii), one first has to ensure that the eigenvalues are at most permuted by the action on the qubit and then choose the operators $x,y$ such that the state is again transformed into KCF \footnote{Such a transformation is then always possible.}. This leads to the following Lemma (see \cite{ChMi10}).

\begin{lemma}
 \label{cor:sym2DD}
Let $A$ be a $2\times D \times D$ state with only finite eigenvalues, $\{x_i\}_{i=1}^l$. Then $G_A$ is characterized as follows. For an invertible matrix $g \in GL(\C,2)$ on the qubit, with
  \begin{align}
  g = \begin{pmatrix}
      \alpha & \beta \\
      \gamma & \delta
     \end{pmatrix}, \label{eq:smat}
 \end{align}
  there exist invertible matrices $x,y \in GL(\C,D)$ on the qudits such that $s \otimes x \otimes y^T \in G_A$ iff there exists a permutation $\sigma \in S(l)$ of the eigenvalues $\{x_i\}_{i=1}^l$ such that
  \begin{align} \label{eq:EigenvaluesMPsym}
  \frac{\alpha x_i + \beta}{\gamma x_i + \delta} = x_{\sigma(i)} \ \forall i 
 \end{align}
 and such that $\sigma$ only permutes eigenvalues of matching multiplicities (i.e., $x_i$ and $x_{\sigma(i)}$ have coinciding size signatures for all $i$).
\end{lemma}

Let us from now on refer to $g\in GL(\C,2)$ in $g \otimes x \otimes y^T \in G_A$ as a qubit symmetry and to $x,y$ as qudit symmetries. From Lemma \ref{cor:sym2DD} we have that for a given fiducial state, $A$, the qubit symmetry can be easily determined via Eq. (\ref{eq:EigenvaluesMPsym}). The corresponding qudit symmetry $x$ and $y$ can then be computed as explained above (see, e.g., Ref. \cite{Gantmacher}).

In case the matrix pencil does not possess any eigenvalue (see case (iv)), it has been shown in \cite{HeGa18} that for any operator $g$ there exist operators $x,y$ such that $g \otimes x \otimes y^T \in G_A$.

Note that the symmetry group of the fiducial state (as of any state) is generated by symmetries of the form $\identity \otimes B \otimes C$ as well as by symmetries $g \otimes B_g \otimes C_g$, for predefined operators $B_g$ and $C_g$.

In the subsequent sections we will use the results reviewed here to determine the symmetries of the fiducial states, which we then use to determine the symmetries and the SLOCC classes of the corresponding TIMPS.

\section{Symmetries and SLOCC classes of the TIMPS $\Psi(M(\omega))$}
\label{sec:M10M11M1inf}

We first determine the symmetries of the fiducial state using MP-theory and then use the results summarized above to determine the symmetries of the corresponding MPS. As mentioned before, a representative of this SLOCC class is the state $\ket{ M(\omega) } = \ket{0}(\ket{00}+\omega \ket{11}+ \omega^2 \ket{22}) + \ket{1}(\ket{00}+\ket{11}+\ket{22})$, where $\omega = e^{\frac{i 2 \pi}{3}}$. The corresponding matrix pencil reads
\begin{align*}
\mathcal{P} &= M^1(1) \oplus M^1(\omega) \oplus M^1(\omega^2) \\ 
&= \begin{pmatrix}\mu + \lambda & 0 &0\\0& \omega \mu + \lambda   & 0\\ 0 & 0 & \omega^2 \mu + \lambda \end{pmatrix}.
\end{align*}

Let us remark that an alternative representative is the state $\ket{0}(\ket{11}+\ket{22}) + \ket{1}(\ket{00}+\ket{11})$. We consider the alternative representative when discussing normality in Appendix \ref{app:M10M11M1infnormality}, as it leads to a more sparsely populated tensor $A$. Here, we will stick to the representative $\ket{ M(\omega) }$ though, as the group structure of the local symmetries of the generated TIMPS will appear more natural for this representative.

\subsection{Symmetries of the fiducial state}
In this subsection we discuss the symmetries of the fiducial state $\ket{ M(\omega) }$.

As any symmetry of the form  $\identity \otimes B \otimes C$ must fulfill that $B\mathcal{P} C^T \propto \mathcal{P}$ (see Sec \ref{sec:prelim}), we obtain for any such symmetry 
(choosing a convenient normalization) that $B = \operatorname{diag}(1,B_{11},B_{22})$,  $C = \operatorname{diag}(1,1/B_{11},1/B_{22})$.

The matrix pencil has three distinct eigenvalues, $(x_0, x_1, x_2) = (1, \omega, \omega^2)$. It follows that there is a discrete set of six operators $g$ appearing as the first local operator in the symmetries of the fiducial state. These operators correspond to all possible permutations of the eigenvalues of the matrix pencil. We will index these operators by $\sigma \in S_3$, where $\sigma$ describes the permutation of the eigenvalues. We will use the notation $\sigma = \begin{pmatrix}0 & 1 & 2\\ \sigma(0) & \sigma(1) & \sigma(2) \end{pmatrix} = (\sigma(0), \sigma(1), \sigma(2))$. Moreover, we use permutation matrices
\begin{align*}
P_\sigma = \sum_i \ket{\sigma_i}\bra{i}.
\end{align*}
Defining
\begin{align}
\label{eq:diag3_A}
&{g_{(0,1,2)}} = \identity,\ {g_{(0,2,1)}} =\begin{pmatrix}0 &  1 \\1  &0 \end{pmatrix}, \nonumber\\ & {g_{(2,1,0)}} =\begin{pmatrix}0 &  \omega^2 \\1  & 0 \end{pmatrix}, \ {g_{(1,0,2)}} =\begin{pmatrix}0 &  \omega \\1  & 0 \end{pmatrix}, \nonumber\\  &{g_{(2,0,1)}} =\begin{pmatrix}\omega &  0 \\0  & 1 \end{pmatrix}, \ {g_{(1,2,0)}} =\begin{pmatrix}\omega^2 &  0 \\0  & 1 \end{pmatrix} 
\end{align}
we find symmetries ${g} \otimes B_{g} \otimes C_{g}$ of $\ket{ M(\omega) }$ for
\begin{align}
\label{eq:EigenvaluesMPsymM10M11M1inf}
{g} &= g_\sigma \nonumber \\
B_{g} &= P_\sigma^{-1} D_{\sigma}^{-1}\nonumber \\
C_{g}^T &= P_\sigma,
\end{align}
where $D_\sigma = \begin{cases}\identity & \operatorname{sgn}(\sigma)=1 \\ \operatorname{diag}(1, \omega,\omega^2) & \operatorname{sgn}(\sigma)=-1\end{cases}$, where $\sigma \in S_{3}$, and where $\operatorname{sgn}$ denotes the signum of $\sigma$. We will refer to $\sigma \in \{(0,2,1),(2,1,0),(1,0,2)\}$ ($\sigma$ with $\operatorname{sgn}(\sigma)=-1$) as transpositions and to  $\sigma \in \{(1,2,0),(2,0,1)\}$ ($\sigma$ with $\operatorname{sgn}(\sigma)=1$, but $\sigma \neq \identity$) as cyclic permutations of length 3.
The symmetry group of the state $\ket{ M(\omega) }$ is given by (see Section \ref{sec:prelim})
\begin{align}
\label{eq:symM10M11M1inf}
{g} \otimes B_{g} B \otimes C_{g} C.
\end{align}
 
We will use these symmetries to determine the local symmetries of the corresponding (normal) TIMPS. As mentioned above, to determine then the TIMPS which are SLOCC equivalent it is sufficient to consider the fiducial states of the form $\identity \otimes b \otimes \identity \ket{ M(\omega) }$. The tensor $A_b$ associated to this state reads
\begin{align}
\label{eq:M10M11M1inftensor}
A_b^0 = \begin{pmatrix} b_{00}& b_{01} & b_{02} \\  \omega b_{10} & \omega b_{11} &  \omega b_{12}\\ \omega^2 b_{20} & \omega^2 b_{21}& \omega^2 b_{22}\end{pmatrix}, \ 
A_b^1 = b.
\end{align}

The corresponding symmetries are obviously of the form
\begin{align*}
{g} \otimes b B_{g} B b^{-1} \otimes C_{g} C.
\end{align*}

\subsection{Local symmetries of the TIMPS \texorpdfstring{$\Psi(M(\omega))$}{Psi(Momega)}}

Let us now characterize the symmetries of normal TIMPS generated by $\identity \otimes b \otimes \identity \ket{ M(\omega) }$. To ease notation, we denote the local symmetry group of this fiducial state by $G_b$ throughout this whole section. As explained in the preliminaries, to identify the symmetries of the TIMPS, we first characterize all possible $N$-cycles $({g}_{\sigma_k})_{k=0}^{N-1}$ in $G_b$. After identifying all possible cycles, we  characterize all normal TIMPS (i.e. all $b$) admitting each of the identified cycles.

Using the symmetry of the representative $\ket{ M(\omega) }$, $x= B_{g} B$, $y=\left( C_{g} C \right)^T$ [see Eq. (\ref{eq:symM10M11M1inf})] and inserting in the concatenation conditions for $\identity \otimes b \otimes \identity \ket{ M(\omega) }$ [see Eq. (\ref{eq:ConcatSymb})] we obtain 
 \begin{align}
\label{eq:lhsrhsM10M11M1inf}
b   P_{\sigma_{k+1}}^{-1} B_{k+1}  D_{\sigma_{k+1}}^{-1}   b^{-1} \propto P_{\sigma_k}^{-1}  B_k,
\end{align} 
where $B_k$ are arbitrary diagonal matrices stemming from the symmetries of the form  $\identity \otimes B \otimes C$, and $P_{\sigma_k}$ as well as $D_{\sigma_k}$ stem from the symmetries ${g} \otimes B_{g} \otimes C_{g}$ as in Eq. (\ref{eq:EigenvaluesMPsymM10M11M1inf}).

Note that Eq. (\theequation) comprises a similarity transformation among so-called \emph{monomial matrices}, which are also called \emph{generalized permutation matrices}. These are (invertible) matrices that can be written as a product of a permutation matrix and a diagonal matrix.
In the following, we use the Fourier transform $\mathcal{F} = \begin{pmatrix}1 & 1 & 1 \\ 1 & \omega & \omega^2 \\ 1 & \omega^2 & \omega \end{pmatrix}$, as well as Fourier transforms acting on subspaces $\mathcal{F}_{01} = \begin{pmatrix}1 & 1 & 0 \\ 1 & -1 & 0 \\ 0 & 0 & 1 \end{pmatrix}$ ($\mathcal{F}_{02}$ and $\mathcal{F}_{12}$ analogously).
\begin{observation}
\label{obs:genpermutation3}
Let $P_\sigma$ be a permutation matrix of dimension 3 and $D=\operatorname{diag}(d_0,d_1,d_2)$, $d_i \in \mathbb{C}\setminus\{0\}$. Then the eigenvalues and eigenvectors of $P_\sigma D$ can be determined as follows. 
In case $\sigma$ is trivial, the eigenvalues are the entries of $D$ and the eigenvectors are the computational basis vectors.
In case $\sigma$ is a 3-cycle, the eigenvalues read $d$, $d\omega$, $d\omega^2$, where $d=(d_0 d_1 d_2)^{1/3}$. The eigenvectors are given by $\tilde{D} \mathcal{F}$ with $\tilde{D} = \operatorname{diag}(\tilde{d}_0,\tilde{d}_1,\tilde{d}_2)$, where $\tilde{d}$ may be determined via the recurrence relation $\tilde{d}_{\sigma(i)} = \tilde{d}_{i}d_{i}/d$.
Finally, in case $\sigma$ is the transposition $(0,2,1)$, the eigenvalues read $d_{0}$ and $\pm \sqrt{d_{1}d_{2}}$, and similarly for the remaining transpositions. The eigenvectors are given by  $\tilde{D} \mathcal{F}_{12}$ with $\tilde{D} = \operatorname{diag}(1,\sqrt{d_2},\sqrt{d_1})$, and similarly for the remaining transpositions.
\end{observation}

As an immediate consequence, considering Eq. (\theequation), we observe that if $\sigma_k$ is a transposition for some $k$, then  $\sigma_l$ cannot be a three-cyclic permutation for any $l$ (and vice versa), due to the mismatch in the eigenvalues of the right-hand side and the left-hand side of Eq. (\theequation).
\begin{observation}
\label{obs:permutationsgnfixed}
Any cycle in  $G_b$ involving a transposition, $\sigma_k$, cannot involve a cyclic permutation of length 3, $\sigma_l$, and vice versa.
\end{observation}
\begin{proof}
We make use of Observation \ref{obs:genpermutation3}. Let us assume that $P_{\sigma_k}$ is a transposition. Then, two eigenvalues of $P_{\sigma_k}^{-1} B_k$ differ by multiplication with $-1$. However, if $P_{\sigma_l}$ is a cyclic permutation of length 3, then any pair of eigenvalues of $P_{\sigma_l}^{-1} B_l$ differs by multiplication with $\omega=e^{i\frac{2 \pi}{3}}$ (or $\omega^2$). Suppose wlog. that $l>k$ and that $\sigma_q$ is identity for $q \in \{k+1, \ldots, l-1\}$ (one may always find such a sequence within any potential cycle involving a transposition as well as a cyclic permutation of legnth 3). Recall that $D_{\sigma_q}=\identity$ unless $\sigma_q$ is a transposition. Then, due to Eq. (\ref{eq:lhsrhsM10M11M1inf}), the eigenvalues of $P_{\sigma_k}^{-1} B_k$ and $P_{\sigma_l}^{-1} B_l$ must coincide (up to a common scaling factor), leading to a contradiction.
\end{proof}

Let us now prove a lemma excluding non-trivial cycles of a certain form. Later on we will resort to this lemma in order to exclude non-trivial cycles in a much broader scope.
\begin{lemma}
\label{lemma:identityattwosites}
Any cycle in  $G_b$ involving $g_k = g_{k+1} = \identity$ for some $k$ must be trivial, i.e., $g_k = \identity$ for all $k$.
\end{lemma}
\begin{proof}
Let us assume wlog. that $g_1 = g_2 = \identity$. In order to prove the lemma it suffices to show that Eq. (\ref{eq:lhsrhsM10M11M1inf}) implies $g_0 = \identity$, as then the argument may be iterated in order to show $g_k = \identity$ for all $k$. 

Let us now show that $g_0 = \identity$.
As $g_1 = g_2 = \identity$, Eq. (\ref{eq:lhsrhsM10M11M1inf}) for $k=1$ reads
 \begin{align*}
b    B_{2}   b^{-1} \propto   B_1.
\end{align*} 
As Eq. (\theequation) displays a similarity transformation and both $B_2$ and $B_1$ are diagonal, we have that $B_2 \propto \tilde{P} B_1 \tilde{P}^{-1}$ for some permutation matrix $\tilde{P}$. Let us now distinguish three cases depending on the degeneracies of the eigenvalues of $B_1$. In case $B_1 \propto \identity$, considering Eq. (\ref{eq:lhsrhsM10M11M1inf}) for $k=0$ immediately yields $P_{\sigma_0}^{-1}B_0 \propto \identity$ and thus $g_0 = \identity$. Let us now consider the case that all the eigenvalues of $B_1$ are non-degenerate. Then, due to the uniqueness of the spectral decomposition (see Observation \ref{obs:genpermutation3} for spectral decompositions of the matrices involved in Eq. (\theequation)), $b$ must be a monomial matrix too. Using this fact in Eq. (\ref{eq:lhsrhsM10M11M1inf}) for $k=0$ shows that $P^{-1}_{\sigma_{0}} = \identity$.

Let us finally consider the case that $B_1$ has two distinct eigenvalues with multiplicities one and two, respectively. Let us assume wlog that the degenerate subspace is spanned by $\ket{0},\ket{1}$. Then, due to Eq. (\theequation) and the uniqueness of the spectral decomposition, $b \tilde{P}$ must be block-diagonal in the subspace spanned by $\{\ket{0},\ket{1}\}$ and $\ket{2}$. Let us now consider Eq. (\ref{eq:lhsrhsM10M11M1inf}) for $k=0$. If $\tilde{P} \ket{2} = \ket{2}$, then $b$ commutes with $B_1$, and thus $P^{-1}_{\sigma_{0}} = \identity$ follows immediately.
If $\tilde{P} \ket{2} \neq \ket{2}$, then considering Eq. (\ref{eq:lhsrhsM10M11M1inf}) for $k=0$ shows that $\sigma_{0} \in \{\identity,(1,0,2)\}$. This can be seen as follows. The left-hand side of  Eq. (\ref{eq:lhsrhsM10M11M1inf}) for $k=0$ reads $b B_1 b^{-1} = \left( b\tilde{P}\right) \left( \tilde{P}^{-1}B_1\tilde{P}\right) \left(b\tilde{P}\right)^{-1}$. Combining the fact that $\tilde{P}^{-1}B_1\tilde{P}$ is diagonal with the block-diagonal structure of $\left( b\tilde{P}\right)$ shows that the right-hand side of Eq. (\ref{eq:lhsrhsM10M11M1inf}) for $k=0$ must have the same block-diagonal structure and thus $\sigma_{0} \in \{\identity,(1,0,2)\}$. Let us now argue that the case $\sigma_{0} = (1,0,2)$ cannot occur. To this end, let us assume $\sigma_{0} = (1,0,2)$ and show that this leads to a contradiction. Let us consider Eq. (\ref{eq:lhsrhsM10M11M1inf}) for $k=N-1$ and rewrite the left-hand side as $\left( b\tilde{P}\right) \left( \tilde{P}^{-1} P^{-1}_{(1,0,2)}\tilde{P}\right)D \left(b\tilde{P}\right)^{-1}$ for some diagonal $D$. Due to the right-hand side of Eq. (\ref{eq:lhsrhsM10M11M1inf}), $\left( b\tilde{P}\right) \left( \tilde{P}^{-1} P^{-1}_{(1,0,2)}\tilde{P}\right)D \left(b\tilde{P}\right)^{-1}$ must be a monomial matrix, which is only possible if $b$ is a monomial matrix (this can be seen by considering the last row and column of the matrix expression). However, if $b$ is a monomial matrix, then Eq. (\ref{eq:lhsrhsM10M11M1inf}) for $k=0$ implies that $P_{\sigma_0}^{-1} = \identity$ (a diagonal matrix conjugated by a monomial matrix remains diagonal), which is a contradicition. This completes the proof of the lemma.
\end{proof}

In the following, we will make use of the group structure of the symmetries. Obviously, whenever $S_1$ and $S_2$ are symmetries of a state, then also $S_1 S_2$ is. As we consider TIMPS, we additionally have that $S_1 \mathcal{T} S_2 \mathcal{T}^{-1}$ is a symmetry of the MPS, where $\mathcal{T}$ denotes the translation operator. It can be easily verified that this structure carries over to cycles. If  $G_b$ exhibits an $N$-cycle $g_0, \ldots, g_{N-1}$, there must also exist an $N$-cycle $g_0 g_1, g_1 g_2,  \ldots, g_{N-1} g_{0}$, etc. Subsequently, we will also deal with situations in which we have partial information about a cycle. Given a string of operators $g_0, \ldots, g_{N-1}$, we denote by \emph{substring} any list of operators which consecutively appear within the string $g_0, \ldots, g_{N-1}$.
With the considerations above, one may easily convince oneself of the following. Given a string of operators $g_0, \ldots, g_{k}$ and the promise that one may append operators $g_{k+1}, \ldots, g_{N-1}$ in order to obtain an $N$-cycle in $G_b$, the promise must also hold for any element-wise product of substrings of coinciding length such as e.g. the length-3 string
$g_1g_3, g_2 g_4, g_3 g_5$. Concretely, there must exist $N-3$ additional operators such that an $N$-cycle starting with $g_1g_3, g_2 g_4, g_3 g_5$ is obtained.
In the following, it will be helpful to call the set of all such element-wise products of substrings the set of \emph{generated substrings}.
Using Lemma \ref{lemma:identityattwosites}, these considerations allow to make the following observation.

\begin{observation}
\label{obs:substrings}
A given string of operators $g_0, \ldots, g_{N-1}$, which contains the substring $A,B,C$ as well as the substring $A,B,D$ for some operators $A,B,C,D$ such that $C \neq D$, cannot form an $N$-cycle in $G_b$. 

More generally, the same holds if $A,B,C$ and $A,B,D$ are generated substrings of $g_0, \ldots, g_{N-1}$.
\end{observation}
\begin{proof}
The proof is by contradiction. Suppose that  $g_0, \ldots, g_{N-1}$ forms an $N$-cycle in  $G_b$, and $A,B,C,D$ are such as in the statement of the observation. Due to the group structure, there must also exist an $N$-cycle described by a string of operators $g'_0, \ldots, g'_{N-1}$ containing $A^{-1}A, B^{-1}B, C^{-1}D$ as a substring. Due to Lemma \ref{lemma:identityattwosites}, $C^{-1}D = \identity$, which contradicts the assumption $C \neq D$.
\end{proof}

We are now in the position to characterize all possible non-trivial cycles within $G_b$. For brevity we use the following shorthand notation for permutations, $\tau = (0,2,1)$, $\epsilon = (1,0,2)$, $\kappa = (2,1,0)$ (transpositions), and $S = \circlearrowright=(1,2,0)$, $S^2 = \circlearrowleft=(2,0,1)$ (cyclic permutations of length 3). We assign labels $T_0^{\tau/\epsilon/\kappa/\circlearrowleft/\circlearrowright}, \ldots, T_7^{\tau/\epsilon/\kappa/\circlearrowleft/\circlearrowright}$ and $C_0, \ldots, C_4$ to specific cycles as in Table \ref{tab:M10M11M1infcycles} ("$T$" indicating that the cycle involves transpositions, and "$C$" indicating that the cycle involves cyclic permutations of length 3. The superindex differentiates between subgroups that are of a similar structure, e.g., $T_1^\tau$ refers to the 2-cylce $\tau \otimes \identity$, while $T_1^\epsilon$ refers to the 2-cylce $\epsilon \otimes \identity$).

\begin{theorem}
\label{thm:M10M11M1infcycles}
The possible $N$-cycles ${g}_0, \ldots, {g}_{N-1}$ in  $G_b$ are given by $T_0^{\tau/\epsilon/\kappa/\circlearrowleft/\circlearrowright}, \ldots, T_7^{\tau/\epsilon/\kappa/\circlearrowleft/\circlearrowright}$ and $C_0, \ldots, C_4$ as in Table \ref{tab:M10M11M1infcycles} and have lengths $N \in \{1,2,3,4,6\}$.
\end{theorem}
\begin{proof}
It turns out that the necessary conditions for a string of operators forming a cycle as given in Lemma \ref{lemma:identityattwosites} and Observations \ref{obs:permutationsgnfixed} and \ref{obs:substrings} are very restrictive. In fact, all strings of operators satisfying the conditions may be exhaustively enumerated. We will now argue why this is the case and, in the course of that, provide such an enumeration. 
Consider the following tree exploration protocol. Starting from $N=1$, strings of operators ${g}_0, \ldots, {g}_{N-1}$ of increasing length $N$ are generated by appending operators to the previously considered strings of length $N-1$ (and thus, a tree is formed). A string (branch) is discarded if it, or any of its generated substrings, violates the conditions in Lemma \ref{lemma:identityattwosites} or Observations \ref{obs:permutationsgnfixed} and \ref{obs:substrings}. Moreover, one may stop further exploring a branch once the substring consisting of the last two operators ${g}_{N-2},{g}_{N-1}$ has appeared previously as a substring within the considered branch. The reason for this is that the only possibility to continue from there on (without violating conditions in Lemma \ref{lemma:identityattwosites} or Observations \ref{obs:permutationsgnfixed} and \ref{obs:substrings}) is repeating the sequence starting from the first appearance of the substring ${g}_{N-2},{g}_{N-1}$ over and over again.
Whenever, in some branch, this point is reached, one then either has obtained a candidate for an $N$-cycle, if the sequence ${g}_{N-2},{g}_{N-1}$ coincides with ${g}_{0}, {g}_{1}$, or one discards the string (and abandons the branch), otherwise, as in the latter case it is impossible to close the cycle.
As the number of operators to choose from is finite (in fact, 6), this is guaranteed to happen at a finite $N$ \footnote{A very naive bound would be $N \leq 6^2+2=38$ using that a string of length $N$ has $N-1$ substrings of length 2 and there exist $6^2$ distinct strings of length 2. Much better bounds can be obtained, though.}. One may also skip exploring a branch whenever the currently considered string contains a substring that has been considered already. For instance, if all branches starting with ${g_0}, {g_1}= \tau, \epsilon$ have been handled, then one may skip investigating the branch ${g_0}, {g_1}, {g_2} = \tau, \tau, \epsilon$, as all possibly emerging cycles will have been identified already.

Following this procedure, one obtains candidates\footnote{Recall that one obtains strings of operators that satisfy necessary conditions for forming an $N$-cycle for some MPS.} for cycles as in Table \ref{tab:M10M11M1infcycles} as well as the following additional candidates: $\identity \otimes S$ $(\tilde{C}_0)$, $\identity \otimes S \otimes S$ $(\tilde{C}_1)$, $\identity \otimes S \otimes S \otimes S^2 \otimes \identity \otimes S^2\otimes S^2 \otimes S$  $(\tilde{C}_2^{\circlearrowleft})$, $\identity \otimes S \otimes S^2 \otimes S^2 \otimes \identity \otimes S^2\otimes S \otimes S$  $(\tilde{C}_2^{\circlearrowright})$, $\tau \otimes \tau \otimes \epsilon \otimes \epsilon \otimes \kappa \otimes \kappa ({\tilde{T}}_0^{\circlearrowleft})$,  ${\tilde{T}}_0^{\circlearrowright}$ analogously, $\epsilon \otimes \epsilon \otimes \tau \otimes \epsilon \otimes \kappa \otimes \kappa \otimes \tau \otimes \kappa$ (${\tilde{T}}_1^{\tau, \circlearrowleft}$), and ${\tilde{T}}_1^{\epsilon, \circlearrowleft}$, ${\tilde{T}}_1^{\kappa, \circlearrowleft}$, ${\tilde{T}}_1^{\tau, \circlearrowright}$, ${\tilde{T}}_1^{\epsilon, \circlearrowright}$, ${\tilde{T}_1}^{\kappa, \circlearrowright}$ analogously. Note that any ${\tilde{T}}_0$ has $\tilde{C}_0$ as a subgroup and moreover, any ${\tilde{T}_1}$ has some ${\tilde{C}}_2$ as a subgroup. Finally, it is straightforward to show that for $\tilde{C}_0$, $\tilde{C}_1$, and $\tilde{C}_2$, there does not exist any $b$ satisfying the concatenation conditions in Eq. (\ref{eq:lhsrhsM10M11M1inf}). This shows that $\tilde{C}^{\circlearrowleft/\circlearrowright}_i$ and $\tilde{T}_i^{\tau/\epsilon/\kappa/\circlearrowleft/\circlearrowright}$ cannot be cycles, in contrast to the cycles shown in Table \ref{tab:M10M11M1infcycles}. 
\end{proof}

\begin{table}[h]
\begin{tabular}{l|l|l|l}
Label & Subgroup(s) & Cycle & Length \\ \hline
  $C_0$&- & $S$& 1 \\
 $T_0^\tau$ & - & $\tau$ & 1\\
 $T_0^\epsilon$ & - & $\epsilon$ & 1\\
 $T_0^\kappa$ & - & $\kappa$ & 1\\ \hline
  $C_1$ &-& $S \otimes S^2$  & 2 \\
 $T_1^\tau$ & $T_0^\tau$ & $\identity \otimes \tau$& 2  \\
 $T_1^\epsilon$ & $T_0^\epsilon$ & $\identity \otimes \epsilon$  & 2\\
 $T_1^\kappa$ & $T_0^\kappa$ & $\identity \otimes \kappa$ & 2 \\
 $T_2^\tau$ &  $T_0^\tau$, $C_1$ & $\epsilon \otimes \kappa$ & 2\\
 $T_2^\epsilon$ &  $T_0^\epsilon$, $C_1$ & $\tau \otimes \kappa$ & 2\\
 $T_2^\kappa$ &  $T_0^\kappa$, $C_1$ & $\tau \otimes \epsilon$ & 2\\ \hline
  $C_2$& $C_0$ & $\identity \otimes S \otimes S^2$& 3 \\
  $T_3^\tau$& - & $\identity \otimes \tau \otimes \tau$& 3 \\
  $T_3^\epsilon$& - & $\identity \otimes \epsilon \otimes \epsilon$& 3 \\
  $T_3^\kappa$& - & $\identity \otimes \kappa \otimes \kappa$& 3 \\
  $T_4^{\circlearrowleft}$& $C_0$ & $\tau \otimes \kappa \otimes \epsilon$& 3 \\
  $T_4^{\circlearrowright}$& $C_0$ & $\tau \otimes \epsilon \otimes \kappa$& 3 \\
  $T_5^{\circlearrowleft}$& $C_2$, $C_0$ & $\tau \otimes \tau \otimes \kappa $& 3 \\
  $T_5^{\circlearrowright}$& $C_2$, $C_0$ & $\tau \otimes \tau \otimes \epsilon $& 3 \\ \hline
  $C_3$& - & $\identity \otimes S \otimes \identity \otimes S^2$ & 4\\
  $T_6^\tau$& $C_3$, $T_0^\tau$ & $\kappa \otimes \tau \otimes \epsilon \otimes \tau $ & 4\\
  $T_6^\epsilon$& $C_3$, $T_0^\epsilon$ & $\tau \otimes \epsilon \otimes \kappa \otimes \epsilon $ & 4\\
  $T_6^\kappa$& $C_3$, $T_0^\kappa$ & $\epsilon \otimes \kappa \otimes \tau \otimes \kappa $ & 4\\ \hline
   $C_4$& $C_1$& $\identity \otimes S \otimes S \otimes \identity \otimes S^2 \otimes S^2$  & 6\\
  $T_7^\tau$& $C_4$, $C_1$, $T_2^\tau$, $T_0^\tau$ & $\kappa \otimes \tau \otimes \epsilon \otimes \epsilon \otimes \tau \otimes \kappa$  & 6 \\
  $T_7^\epsilon$& $C_4$, $C_1$, $T_2^\epsilon$, $T_0^\epsilon$ & $\tau \otimes \epsilon \otimes \kappa \otimes \kappa \otimes \epsilon \otimes \tau$  & 6 \\
  $T_7^\kappa$& $C_4$, $C_1$, $T_2^\kappa$, $T_0^\kappa$ & $\epsilon \otimes \kappa \otimes \tau \otimes \tau \otimes \kappa \otimes \epsilon$  & 6 \\\hline
\end{tabular}
\caption{Characterization of all possible cycles in $G_b$ (considering fiducial states $\identity \otimes b \otimes \identity \ket{ M(\omega) }$), and thus, all possible symmetries of associated normal MPS. Cycles of lengths 1, 2, 3, 4, and 6 are possible. The third column gives the physical operators, here we use the shorthand notation $\tau$ in place of ${g}_\tau$, etc. Note that the given cycles should be understood as generators, i.e., the cycle $C_3$, $\identity \otimes S \otimes \identity \otimes S^2$, also comprises the cycle $S \otimes S \otimes S^2 \otimes S^2$, or $T_6^{\tau}$ also comprises $\epsilon \otimes \epsilon \otimes \kappa \otimes \kappa$, etc. In the second column, we indicate which cycles emerge as a subgroup of the cycle at hand.}
\label{tab:M10M11M1infcycles}
\end{table}

Before we characterize the $b$ possessing the cycles mentioned in the theorem, a few remarks are in order. First, note that some of the identified cycles lead to symmetry groups that are subgroups of the symmetry group corresponding to another of the identified cycles, as indicated in the second column of Table \ref{tab:M10M11M1infcycles}. Thus, e.g., it is not possible to find $b$ such that the fiducial state exhibits the 2-cycle $T_1^{\tau}$, but not $T_0^\tau$. We illustrate the emerging group structure in Figure \ref{fig:groupstructure}. However, we will see later on that not every combination of cycles that is compatible with the group structure is actually possible. We will see, e.g., that there does not exist a $b$ exhibiting $C_3$, but not any of $T_6^{\tau/\epsilon/\kappa}$.

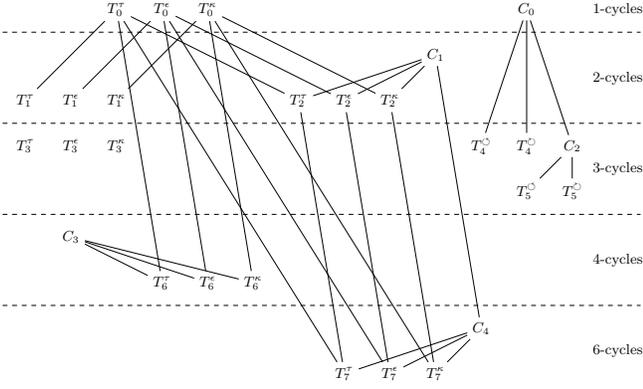
\begin{figure}
\resizebox{\linewidth}{!}{%
\begin{tikzpicture}
  \node (T0tau) at (0,0) {$T_0^\tau$};
  \node (T0eps) at (1,0) {$T_0^\epsilon$};
  \node (T0kap) at (2,0) {$T_0^\kappa$};
  \node (C0) at (9,0) {$C_0$};
  \node (C1) at (7,-1) {$C_1$};
  \node (T1tau) at (-2,-2) {$T_1^\tau$};
  \node (T1eps) at (-1,-2) {$T_1^\epsilon$};
  \node (T1kap) at (0,-2) {$T_1^\kappa$};
  \node (T2tau) at (4,-2) {$T_2^\tau$};
  \node (T2eps) at (5,-2) {$T_2^\epsilon$};
  \node (T2kap) at (6,-2) {$T_2^\kappa$};
  \node (T3tau) at (-2,-3) {$T_3^\tau$};
  \node (T3eps) at (-1,-3) {$T_3^\epsilon$};
  \node (T3kap) at (0,-3) {$T_3^\kappa$};
  \node (T4cc) at (8,-3) {$T_4^\circlearrowleft$};
  \node (T4c) at (9,-3) {$T_4^\circlearrowright$};
  \node (C2) at (10,-3) {$C_2$};
  \node (T5cc) at (9,-4) {$T_5^\circlearrowleft$};
  \node (T5c) at (10,-4) {$T_5^\circlearrowright$};
  \node (C3) at (-1,-5) {$C_3$};
  \node (T6tau) at (1,-6) {$T_6^\tau$};
  \node (T6eps) at (2,-6) {$T_6^\epsilon$};
  \node (T6kap) at (3,-6) {$T_6^\kappa$};
  \node (C4) at (8,-7) {$C_4$};
  \node (T7tau) at (5,-8) {$T_7^\tau$};
  \node (T7eps) at (6,-8) {$T_7^\epsilon$};
  \node (T7kap) at (7,-8) {$T_7^\kappa$};
  \draw (T1tau) -- (T0tau);
  \draw (T1eps) -- (T0eps);
  \draw (T1kap) -- (T0kap);
  \draw (T2tau) -- (T0tau);
  \draw (T2eps) -- (T0eps);
  \draw (T2kap) -- (T0kap);
  \draw (T2tau) -- (C1);
  \draw (T2eps) -- (C1);
  \draw (T2kap) -- (C1);
  \draw (T4cc) -- (C0);
  \draw (T4c) -- (C0);
  \draw (C2) -- (C0);
  \draw (T5cc) -- (C2);
  \draw (T5c) -- (C2);
  \draw (T6tau) -- (T0tau);
  \draw (T6eps) -- (T0eps);
  \draw (T6kap) -- (T0kap);
  \draw (T6tau) -- (C3);
  \draw (T6eps) -- (C3);
  \draw (T6kap) -- (C3);
  \draw (C4) -- (C1);
  \draw (T7tau) -- (T0tau);
  \draw (T7eps) -- (T0eps);
  \draw (T7kap) -- (T0kap);
  \draw (T7tau) -- (C4);
  \draw (T7eps) -- (C4);
  \draw (T7kap) -- (C4);
  \draw (T7tau) -- (T2tau);
  \draw (T7eps) -- (T2eps);
  \draw (T7kap) -- (T2kap);
  \draw[dashed] (-2.5,-0.5) -- (11.5,-0.5);
  \draw[dashed] (-2.5,-2.5) -- (11.5,-2.5);
  \draw[dashed] (-2.5,-4.5) -- (11.5,-4.5);
  \draw[dashed] (-2.5,-6.5) -- (11.5,-6.5);
  \node (1c) at (11,0) {1-cycles};
  \node (2c) at (11,-1.5) {2-cycles};
  \node (3c) at (11,-3.5) {3-cycles};
  \node (4c) at (11,-5.5) {4-cycles};
  \node (6c) at (11,-7.5) {6-cycles};
\end{tikzpicture}
}
\caption{Illustration of the group structure of the cycles for normal MPS generated by the fiducial states $\identity \otimes b \otimes \identity \ket{M(\omega)}$. The depicted labels correspond to the cycles as given in Table \ref{tab:M10M11M1infcycles}. A line connecting two labels indicates that the symmetry group associated to the higher elevated cycle is a subgroup of the symmetry group associated to the cycle below. }
\label{fig:groupstructure}
\end{figure}

Second, note that we have not restricted our attention to TIMPS that are normal, so far. In particular, the characterization of allowed cycles in Theorem \ref{thm:M10M11M1infcycles} holds for normal as well as for non-normal TIMPS. Recall that in case of normal TIMPS, a characterization of cycles yields a full characterization of the symmetries of the TIMPS, while non-normal MPS might possess additional symmetries, which are not captured by the study of cycles. Thus, Table \ref{tab:M10M11M1infcycles} exhaustively lists all the possible symmetries of normal MPS and moreover, it exhaustively lists those symmetries of non-normal MPS, that may be identified via the study of cycles. In the following we will focus on normal TIMPS, though. We observe that $b$ of a certain form cannot be normal.
\begin{observation}
\label{obs:permutationnotnormal}
If $b$ is such that in any row or column $i$, the entry $b_{ii}$ is the only non-vanishing entry, or if $b$ is a generalized permutation matrix, then the TIMPS generated by $\identity \otimes b \otimes \identity \ket{ M(\omega) }$ is non-normal. 
\end{observation}
\begin{proof}
Let us first consider the case that $b$ is such that in any row or column $i$, the entry $b_{ii}$ is the only non-vanishing entry. Note that this property is retained when taking products of matrices of such a form. Moreover, note that if $b$ is of such a form, then both $A_b^0$ and $A_b^1$ as in Eq. (\ref{eq:M10M11M1inftensor}) are of this form too. Hence, it is impossible to find more than seven linearly independent products of $A_b^0$ and $A_b^1$ and the tensor $A_b$ cannot be normal.

Let us now consider the case that $b$ is a generalized permutation matrix, $b = P_\sigma D$. Then, any product of $A_b^0$ and $A_b^1$ comprised of $L$ factors can be written as $P_{\sigma^L} \tilde{D}$ for some diagonal $\tilde{D}$ (as $A_b^{0,1} = b A^{0,1}$ and $A^{0,1}$ is diagonal). Hence, for any $L$ it is impossible to find more than three linearly independent products of $A_b^0$ and $A_b^1$ with $L$ factors. Hence, the tensor $A_b$ cannot be normal.
\end{proof}
 Thus, we restrict our attention to $b$ which are not of the form given in observation in the following. In fact, as we will see (see also Observation \ref{obs:M10M11M1infnormality}), the remaining $b$'s are either normal or $G_b$ does not exhibit any non-trivial cycles.

With Theorem \ref{thm:M10M11M1infcycles}, it is now a straightforward calculation to characterize all $b$, for which the  $G_b$ exhibits any specific cycle listed in Table $\ref{tab:M10M11M1infcycles}$. To this end, one considers Eq. (\ref{eq:lhsrhsM10M11M1inf}) for $\sigma_k$ given by the considered cycle. One may then utilize Observation \ref{obs:genpermutation3} in order to determine $b$.
We present a full characterization of those $b$ (disregarding $b$ of the form given in Observation \ref{obs:permutationnotnormal}) for which  $G_b$ exhibits non-trivial cycles (see Table \ref{tab:M10M11M1infcycles}) in Table \ref{tab:M10M11M1infcycleslong} in Appendix \ref{app:M10M11M1infcycles}.
Let us exemplarily display the calculation for the cycle $C_0$. We use Eq. (\ref{eq:lhsrhsM10M11M1inf}) with $N=1$, $\sigma_0 = S$ and obtain 
\begin{align}
\label{eq:lhsrhsM10M11M1infC0}
b   P_{S}^{-1} B_{0}  D_{S}^{-1}   b^{-1} \propto P_{S}^{-1}  B_0.
\end{align} 
As $D_S = \identity$, making use of Observation \ref{obs:genpermutation3} we have  $P_{S}^{-1} B_{0}  D_{S}^{-1} = P_{S}^{-1} B_{0} =  \tilde{D} \mathcal{F} \operatorname{diag}(1,\omega,\omega^2) \left(\tilde{D} \mathcal{F}\right)^{-1}$ with $\tilde{D} = \operatorname{diag}(1,\frac{B^{(0)}_{22}}{\sqrt[3]{B^{(0)}_{11} B^{(0)}_{22} }^2},\frac{1}{\sqrt[3]{B^{(0)}_{11} B^{(0)}_{22} }})$. Due to the uniqueness of the spectral decomposition we then have
\begin{align*}
    b &= \tilde{D} \mathcal{F} P_{\tilde{\sigma}} \operatorname{diag}(x_0,x_{1},x_{2})  \mathcal{F}^{-1}  \tilde{D}^{-1} \\
    &=  \tilde{D} \operatorname{diag}(1,\omega,\omega^2)^l\mathcal{F}  \operatorname{diag}(x_0,x_{1},x_{2})  \mathcal{F}^{-1}  \tilde{D}^{-1}
\end{align*}
for some $x_i \in \mathbb{C}$, $l \in \{0,1,2\}$, and some permutation $\tilde{\sigma}$ with $\operatorname{sgn}(\tilde{\sigma})=1$. Here, the additional permutation matrix $P_{\tilde{\sigma}}$ comes from the fact that the proportionality factor in Eq. (\ref{eq:lhsrhsM10M11M1infC0}) allows to cyclically permute the eigenvalues. Equivalently, we may write 
\begin{align*}
b= D \operatorname{diag}(1,\omega,\omega^2)^l \begin{pmatrix} b_{0} & b_{1}& b_{2} \\ b_{2} & b_{0} & b_{1} \\ b_{1} & b_{2} & b_{0} \end{pmatrix}D^{-1},
\end{align*}
where $b_{0}, b_{1}, b_{2}\in \mathbb{C}$, $l \in \{0,1,2\}$, and $D$ is an arbitrary diagonal matrix. Let us remark that the diagonal matrix $D$ is actually irrelevant. More precisely, the fiducial states $\identity \otimes b \otimes \identity \ket{ M(\omega) }$ and $\identity \otimes c \otimes \identity \ket{ M(\omega) }$ give rise to the same TIMPS, if $c= D b D^{-1}$ for any diagonal $D$, due to the Fundamental Theorem (Theorem \ref{thm:fund}) and the symmetries of the seed state\footnote{This may also be easily seen by noting that such $b$ and $c$ are related by a $(b \rightarrow c)$ 1-cycle  as in Eq. (\ref{eq:bc_cycleExpl}) with $g = \identity$ for any diagonal matrix $D$.}. 

It is straightforward to perform the calculation for all cycles, leading to the results in Table \ref{tab:M10M11M1infcycleslong} within Appendix \ref{app:M10M11M1infcycles}. We find that there are continuous families of $b$'s leading to the cycles $C_0$, $T_0^{\tau, \epsilon,\kappa}$, $C_1$, $T_1^{\tau, \epsilon,\kappa}$, $T_2^{\tau, \epsilon,\kappa}$, $T_3^{\tau, \epsilon,\kappa}$, and $T_4^{\circlearrowleft,\circlearrowright}$, while there is a discrete number of $b$'s leading to the cycles $C_2$, $T_5^{\circlearrowleft,\circlearrowright}$, $C_3$, $T_6^{\tau, \epsilon,\kappa}$, $C_4$, and $T_7^{\tau, \epsilon,\kappa}$.  

While we defer details on the concrete parametrizations of the sets of operators $b$ for which  $G_b$ exhibits the respective cycles to Appendix \ref{app:M10M11M1infcycles}, the relations among these sets are important in order to know which different symmetry groups can occur simultaneously. Hence, we will discuss these relations in the following.
Let us denote the set of $b$'s leading to a certain cycle, say $T_0^\tau$, by $b(T_0^\tau)$, etc.
Obviously, the families of $b$'s satisfy relations imposed by the group structure of the cycles mentioned above (see Figure \ref{fig:groupstructure}), e.g., $b(T_1^\tau)$ must be a subset of $b(T_0^\tau)$. Note, however, that it is not guaranteed that for every possible combination of cycles, which is compatible with the group structure, there exists a $b$ such that  $G_b$ exhibits this combination of cycles. On the contrary, additional set-theoretic restrictions emerge, which one may not immediately conclude from the group structure of the cycles. We display all the relations within an Euler diagram in Figure \ref{fig:sets}.

\begin{figure}
\subfloat[]{\label{main:a}\includegraphics[width=1.0\linewidth]{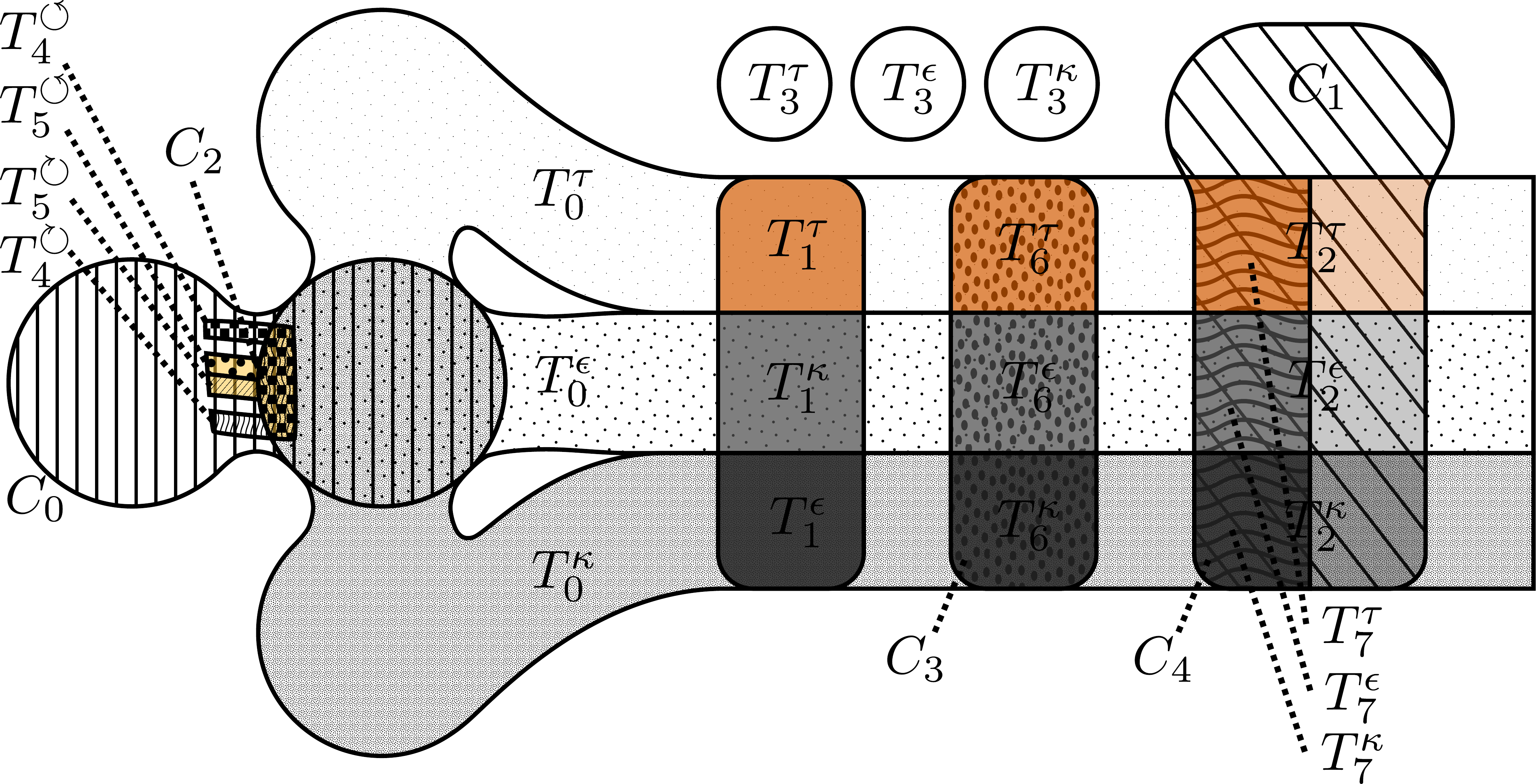}}\par\medskip
\subfloat[]{\label{main:b}\includegraphics[width=1.0\linewidth]{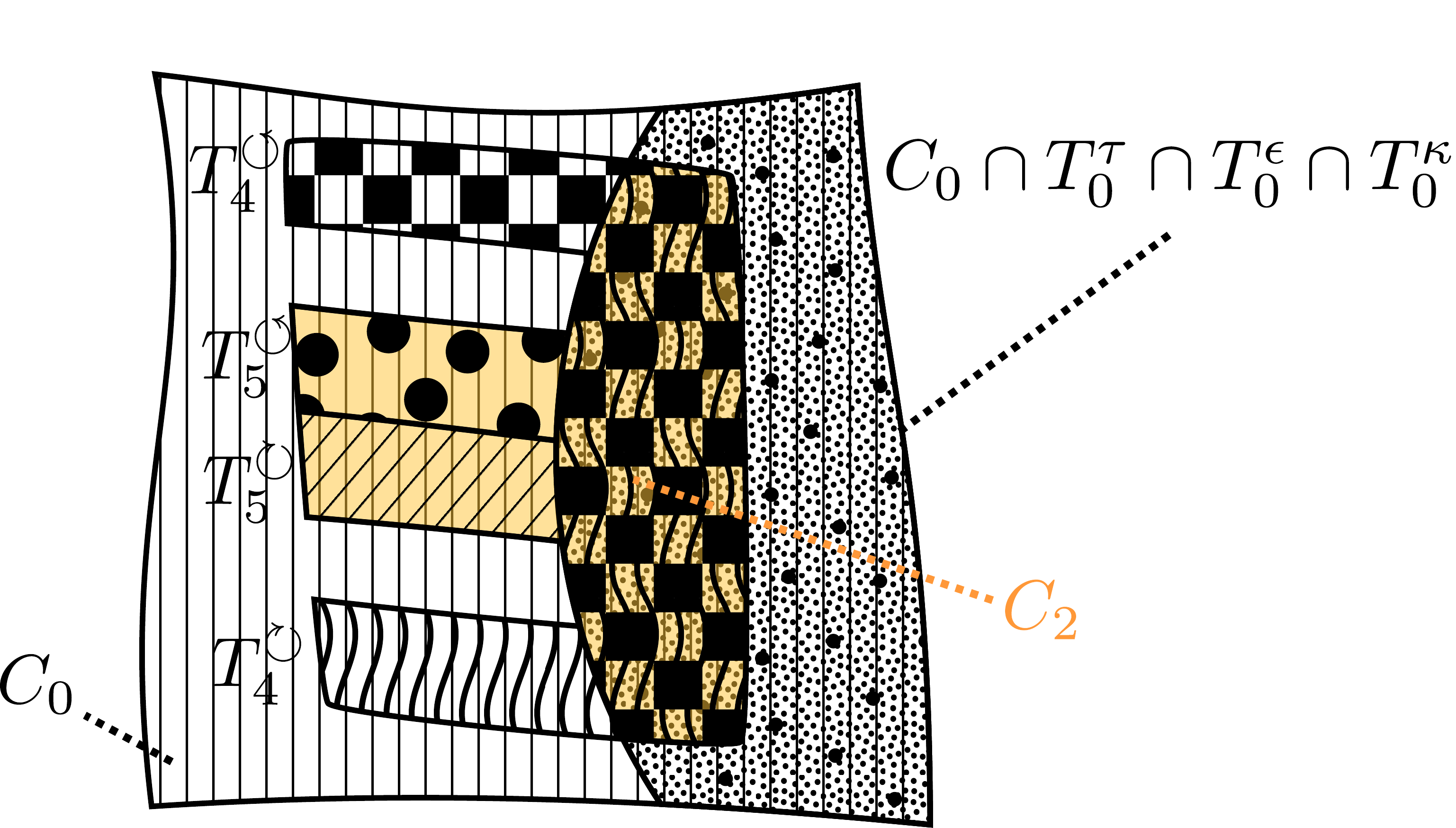}}
\caption{(a) Set-theoretic sketch (Euler diagram) of $b$'s leading to non-trivial cycles in  $G_b$ (considering the fiducial state $\identity \otimes b \otimes \identity \ket{ M(\omega) }$). The sketch illustrates the inclusion relations and intersections between the different sets of $b$'s leading to certain cycles as labelled in Table \ref{tab:M10M11M1infcycles}. Let us point out a few of them. For instance, one sees that the set $b(C_3)$ is a union of the (pairwise) non-intersecting sets $b(T_6^\tau)$, $b(T_6^\epsilon)$, and $b(T_6^\kappa)$, and each of $b(T_6^\sigma)$ is a subset of $b(T_0^\sigma)$ for $\sigma \in \{\tau, \epsilon, \kappa\}$. This illustrates that, e.g., there does exist a $b$ s.t.  $G_b$ exhibits the three cycles $T_6^\tau$, $T_0^\tau$, and $C_3$, but there does not exist any $b$ leading to the cycle $C_3$ only. 
We find that $b(T_0^\tau)$, $b(T_0^\epsilon)$, $b(T_0^\kappa)$, and $b(C_0)$ intersect. However, the three sets $b(T_i^\tau)$, $b(T_j^\epsilon)$, and $b(T_k^\kappa)$ are disjoint for any $i,j,k \in \{1,2,3,6,7\}$. In fact, $b(T_3^{\tau,\epsilon,\sigma})$ are completely isolated in the sense that they do not intersect with any other family. The set $b(T_1^\sigma)$ is a subset of $b(T_0^\sigma)$ (for $\sigma \in\{\tau, \epsilon, \kappa\}$), but does not intersect with any other family. 
Note that the sketch does not correctly represent the geometry or sizes of the sets.
As explained in the main text, the diagram gives as a complete characterization of symmetries of all normal TIMPS generated by fiducial states $\identity \otimes b \otimes \identity \ket{ M(\omega) }$. 
While some of the displayed inclusion relations follow immediately considering the group structure of the cycles as in Figure \ref{fig:groupstructure}, additional relations only reveal themselves considering the full characterization of $b$ as in Table \ref{tab:M10M11M1infcycleslong}. An instance of the latter case would be one of the example mentioned above---the fact that there exists no $b$ s.t.  $G_b$ possesses the cycle $C_3$, but does not possess any of the cycles $T_6^\tau$, $T_6^\epsilon$, and $T_6^\kappa$. (b) Detailed view of an aspect of subfigure (a).}
\label{fig:sets}
\end{figure}

In Observation \ref{obs:permutationnotnormal}, we have given a few conditions on $b$ under which the generated TIMPS is not normal. In contrast to that, we find that any $b$ s.t.  $G_b$ possesses non-trivial cycles and s.t. $b$ does not satisfy one of the mentioned conditions for non-normality, is actually normal.
\begin{observation}
\label{obs:M10M11M1infnormality}
All $b$ s.t. there exist non-trivial cycles in  $G_b$ (see Table \ref{tab:M10M11M1infcycleslong}) are normal with injectivity length $L=4$, $5$, or $6$, unless $b$ fulfills the prerequisites of Observation \ref{obs:permutationnotnormal}. 
\end{observation}
We present a proof of the observation, as well as additional details on proving normality in general, in Appendix \ref{app:M10M11M1infnormality}.

As a consequence, except for $b$ satisfying the conditions in Observation \ref{obs:permutationnotnormal}, the characterization of cycles directly yields the symmetries of the generated TIMPS.  Figure \ref{fig:sets} hence shows all possible (non-trivial) symmetries of normal TIMPS generated by fiducial states of the type $\identity \otimes b \otimes \identity \ket{ M(\omega) }$ and hence for any normal TIMPS corresponding to a fiducial state in the SLOCC class of $\ket{M(\omega)}$.
For any given normal $b$, the symmetry group of the generated TIMPS may be determined by comparing $b$ with the results in Table \ref{tab:M10M11M1infcycleslong} in Appendix \ref{app:M10M11M1infcycles}. Conversely, in order to decide whether there exists a normal TIMPS (generated by some $\identity \otimes b \otimes \identity \ket{ M(\omega) }$) possessing a desired symmetry group, one may simply look up Figure \ref{fig:sets} and see whether the corresponding intersection is non-empty. If it exists, then an appropriate $b$ may be constructed with the help of Table  \ref{tab:M10M11M1infcycleslong} in Appendix \ref{app:M10M11M1infcycles}.

\subsection{SLOCC Classification}
As we are dealing here with finitely many symmetries, it is straightforward to determine the SLOCC classes of the TIMPS.
Whereas we will determine all the classes with infinitely many symmetries in the subsequent section, we will only outline here the procedure and discuss some examples.

To this end, we consider the concatenation conditions presented in Eq. (\ref{eq:bc_cycleExpl}). 
To determine for instance all $(b\rightarrow c)$ 1-cycles, we have to only consider one equation, namely
\begin{align*}
   b   P_{\sigma}^{-1} B  D_{\sigma}^{-1}   \propto P_{\sigma}^{-1}  B c,
\end{align*}
which immediately lets one construct all $c$ connected to a given $b$ via an  $(b\rightarrow c)$ 1-cycle. 
Stated differently, normal TIMPSs corresponding to the fiducial states $\identity \otimes b \otimes \identity \ket{ M(\omega) }$ and $\identity \otimes c \otimes \identity \ket{ M(\omega) }$ respectively are related to each other via a global operation iff $b$ and $c$ fulfill the equation above.

For 2-cycles one would proceed as follows. First, one considers the necessary condition
\begin{align}
   b   P_{\sigma_1}^{-1} D_{\sigma_1}^{-1} B_1^{-1} B_0  D_{\sigma_0}P_{\sigma_0} b^{-1}  \propto P_{\sigma_0}^{-1}  B_0^{-1} B_1 P_{\sigma_1}.
\end{align}
Obviously, the tools utilized throughout this section so far are applicable here. Once $g_0$ and $g_1$ satisfy the necessary conditions for some given $b$, all $c$ connected to $b$ via the  $(b\rightarrow c)$ 2-cycle given by $g_0, g_1$ may be straightforwardly characterized. 

Let us remark that an obvious necessary condition for SLOCC equivalence of two states $\ket{\psi}$ and $\ket{\phi}$ is that their symmetry group must be compatible, i.e., $\mathcal{S}_{\ket{\psi}}$  equals $\mathcal{S}_{\ket{\phi}}$ up to conjugation. For the MPS considered here this must be a conjugation by some tensor product of $g_\sigma$ as in Eq. (\ref{eq:diag3_A}). Thus, not only the order of the full symmetry group must coincide, in fact, but also the number of symmetries involving transpositions as well as the number of symmetries involving cyclic permutations of length 3 must be retained, each. This immediately rules out SLOCC equivalence among many of the families of $b$s as in Figure \ref{fig:sets}.

Let us conclude with some examples. It may be easily verified that the MPS associated to the family $b(T_i^{\tau})$ are SLOCC-equivalent to some MPS associated to the family $b(T_i^{\epsilon})$ and $b(T_i^{\kappa})$ (and vice versa) for any $i \in \{0,1,2,3,6,7\}$. This is witnessed by the $(b \rightarrow c)$ 1-cycles given by $g_\epsilon$, or $g_\kappa$, respectively. Moreover, for even $N$, any MPS generated by some fiducial state belonging to $b(C_0) \cap b(T_0^{\tau}) \cap b(T_0^{\epsilon}) \cap b(T_0^{\kappa})$ is SLOCC equivalent to some MPS associated to $b(C_1) \cap b(T_0^{\kappa}) \cap b(T_2^{\kappa})$ and vice versa. This is witnessed by the $(b\rightarrow c)$ 2-cycle $\identity \otimes \kappa$. Conversely, there exist examples in $b(C_0)$, which are not related to any $b$ in $b(C_1)$ despite compatibility of the stabilizer.

\section{Symmetries and SLOCC classes of the TIMPS \texorpdfstring{$\Psi(LLT)$}{Psi(LLT)}}
\label{sec:L1L1T}
In this section we discuss MPS generated by fiducial states that are represented by $\ket{LLT} = \ket{0}(\ket{01} + \ket{22}) +  \ket{1}(\ket{00} + \ket{12})$. First, we present the symmetries of the fiducial states. Then, we characterize the symmetries of normal MPS, which, in contrast to the previous section, are potentially infinitely many. In the course of that, we give a characterization of those fiducial states that generate normal MPS. Finally, we characterize SLOCC equivalence among normal MPS. We conclude with a few remarks on some non-normal MPS generated by fiducial states represented by $\ket{LLT}$.

The matrix pencil corresponding to $\ket{LLT}$ reads
\begin{align*}
\mathcal{P} = L_1\oplus L_1^T = \begin{pmatrix}\lambda & \mu &0\\0& 0  & \lambda\\ 0 & 0 & \mu \end{pmatrix}.
\end{align*}
Note that the tensor $A$ corresponding to $\ket{LLT}$ is simply given by  $A^0 = \mathcal{P}|_{\mu=1,\lambda=0}$ and $A^1 = \mathcal{P}|_{\mu=0,\lambda=1}$.
\subsection{Symmetries of the fiducial state}
The symmetries of the state $\ket{LLT}$ are special in the sense that any invertible operator acting on the first site forms a local symmetry of $\ket{LLT}$ with appropriate operators acting on the remaining two sites. That is, for any operator $g$ on the qubit, there exist $3\times 3$ matrices, which are uniquely determined by $g$, $B_g$, and $C_g$ such that $g\otimes B_g\otimes C_g$ is a symmetry of the state. Moreover, as mentioned above, any symmetry of ${\ket{LLT}}$ can be written as a product of one symmetry of the form ${g} \otimes B_{g} \otimes C_{g}$ and symmetries of the form $\identity \otimes B \otimes C$.
For 
\begin{equation*}
{g} = \begin{pmatrix}\alpha & \beta \\ \gamma & \delta \end{pmatrix}
\end{equation*} 
it is easy to show that 
\begin{align*}
B_{g} &= \frac{1}{det(g)} \begin{pmatrix}det(g) & 0 & 0 \\0 &\alpha & -\beta \\ 0 & -\gamma & \delta \end{pmatrix} \\
C_{g}^T &= \frac{1}{det(g)} \begin{pmatrix}\alpha & -\gamma &0 \\  -\beta & \delta & 0 \\ 0 & 0 & det(g) \end{pmatrix} .
\end{align*} 
The symmetries of the form  $\identity \otimes B \otimes C$ are given by
\begin{align*}
B &=  \begin{pmatrix}1 &  B_{01} & B_{02} \\0 & B_{11} & 0 \\ 0 & 0 & B_{11} \end{pmatrix} \\
C^T &=  \frac{1}{B_{11}} \begin{pmatrix}B_{11} & 0 & -B_{01} \\0 & B_{11} & -B_{02} \\ 0 & 0 & 1 \end{pmatrix},
\end{align*} 
where we use the normalization $B_{00}=1$.
The symmetries of $\ket{LLT}$ are thus given by
\begin{align*}
{g} \otimes B_{g} B \otimes C_{g} C
\end{align*}
for any ${g} \in GL(2,\mathbb{C})$, $B_{00}, B_{01},B_{02},B_{11} \in \mathbb{C}$.

Clearly, the symmetries of $\identity \otimes b \otimes \identity \ket{LLT}$ are given by
\begin{align} \label{Eq:xy}
g \otimes b\underbrace{ B_{g} B }_{x}b^{-1} \otimes \underbrace{C_{g} C}_{y^T},
\end{align}
using the notation $x$ and $y$ for symmetries of the representative as introduced in Section \ref{sec:prelim}.  We denote the local symmetry group of this fiducial state by $G_b$ throughout this whole section. The tensor $A_b$ associated to $\identity \otimes b \otimes \identity \ket{LLT}$ reads
\begin{align*}
A_b^0 = \begin{pmatrix} 0 & b_{00} & b_{02} \\ 0 & b_{10} & b_{12}\\0 & b_{20} & b_{22}\end{pmatrix} 
\quad \text{and} \quad 
A_b^1 = \begin{pmatrix} b_{00} &0 & b_{01} \\  b_{10} & 0 & b_{11}\\b_{20}& 0 & b_{21}\end{pmatrix}.
\end{align*}

Let us now introduce the following  parametrization for any $3 \times 3$ operators $b$ with $b_{20} = 1$ (it will become clear later on that $b_{20} \neq 0$ is required in order to obtain normal MPS) in terms of a matrix $T$ and a vector $\vec{v}$, as well as two complex numbers  $b_{00}, b_{10}$,
 \begin{align}
 \label{eq:L1L1Tparam}
 b= \begin{pmatrix} b_{00} \\ b_{10} \\ 1\end{pmatrix}. \begin{pmatrix} 1 \\ v_1 - b_{00}  \\ v_0 - b_{10} \end{pmatrix}^T -\det T \begin{pmatrix}&  \sigma_x T^{-1} \sigma_z \\0& \end{pmatrix}.
 \end{align}
 It may be easily veryfied that for a given $b$, $T$ and $\vec{v}$ may be obtained via 
\begin{align}
\label{eq:Tdef}
T &= \begin{pmatrix} b_{02}- b_{00}b_{22} &  b_{10}b_{22} -b_{12} \\ b_{01} - b_{00}b_{21} & b_{10}b_{21} -b_{11} \end{pmatrix}\\
\label{eq:vdef}
\vec{v}&=\begin{pmatrix}
b_{10} + b_{22}\\
b_{00} + b_{21}
\end{pmatrix}.
\end{align}
Note that $\det b = \det T$. Despite the fact that this parameterization might seem a bit arbitrary, we will see that it is particularly useful to characterize the local symmetries and the SLOCC classes of TIMPS corresponding to fiducial states of the form $\one \otimes b \otimes \one \ket{L_1 \otimes L_1^T}$

\subsection{Concatenation conditions}
To obtain a physical symmetry, (normal) MPS need to fulfill the conditions given in Eq. (\ref{eq:ConcatSymb}), which we restate here,
\begin{align*}
y_k b x_{k+1}b^{-1} &\propto \identity \quad \forall  \ k \in \{0, \ldots, N-2\}, \nonumber\\ y_{N-1} b x_{0}b^{-1}  &\propto \identity,
\end{align*}
where $x,y$ are such that $g\otimes x\otimes y^T$ are a symmetry of the fiducial state as in Eq. (\ref{Eq:xy}).
For simplicity, we will in the following only consider the upper equation for $k \in \{0, \ldots, N-1\}$ and consider addition modulo $N$ in this context, i.e., $N \equiv 0$.

Inserting for $x_{k+1},y_k$ the expressions given in Eq. (\ref{Eq:xy}) and using the normalization  $B_{00}^{(k)} = 1$ as well as $\det {g}_k = 1$ for all $k$, the concatenation conditions  $b x_{k+1} b^{-1} \propto y_k^{-1}$ read
 
\begin{align}
\label{eq:lhsrhs}
b \underbrace{\begin{pmatrix}1&  B_{01}^{(k+1)} & B_{02}^{(k+1)} \\0 & \alpha_{k+1} B_{11}^{(k+1)} & -\beta_{k+1} B_{11}^{(k+1)} \\ 0 & -\gamma_{k+1} B_{11}^{(k+1)} & \delta_{k+1} B_{11}^{(k+1)} \end{pmatrix}}_{LHS_{k+1}} b^{-1} \propto \nonumber \\
\underbrace{\begin{pmatrix}\delta_{k}   &  \gamma_{k} & \gamma_{k} B_{02}^{(k)} + \delta_{k} B_{01}^{(k)}\\ \beta_{k}   & \alpha_{k}  & \alpha_{k} B_{02}^{(k)} + \beta_{k} B_{01}^{(k)} \\ 0 &0 &  B_{11}^{(k)} \end{pmatrix}}_{RHS_k}.
\end{align} 

We denote the eigenvalues of $g_k$ by $\chi_k$ and $1/\chi_k$. We use the convention 
$|\chi_k| \geq 1$ and additionally in case $|\chi_k| > 1$ (i), we choose $\Im{\chi_k} >0$, or $\Im{\chi_k} = 0$ and $\Re{\chi_k} > 0$, while in case $|\chi_k| = 1$ (ii) we choose both $\Re{\chi_k} \geq 0$ and $\Im{\chi_k} \geq 0$. We denote the domain of $\chi_k$ by $\mathcal{D}$ and show a sketch of $\mathcal{D}$ in Figure \ref{fig:domain}.
This normalization may be achieved by ordering the eigenvalues appropriately and by a freedom of multiplying ${g}_k$ by $-1$ which still remains after fixing  $\det {g}_k = 1$.  
For each $g_k$ we then consider the Jordan decomposition 
\begin{align}
\label{eq:jordandecomp}
{g}_k = S_k J_k S_k^{-1},
\end{align}
where either $J_k = \operatorname{diag}(\chi_k, 1/\chi_k)$ (in case ${g}_k$ is diagonalizable), or $J_k = \begin{pmatrix} 1 & 1 \\ 0 & 1 \end{pmatrix}$ (in case ${g}_k$ is not diagonalizable, here $\chi_k=1$).

\begin{figure}[t]
\includegraphics[width=0.85\linewidth]{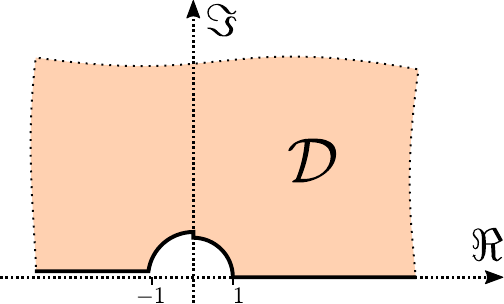}
\caption{Sketch of the domain of $chi_k$, $\mathcal{D}$. $\mathcal{D}$ is comprised of complex numbers in the upper half of the complex plane which have absolute value larger than or equal to one, excluding the negative real axis, and, moreover excluding numbers with negative real part whose absolute value equals 1. The imaginary unit $i$ is included in $\mathcal{D}$.}
\label{fig:domain}
\end{figure}

As a simple necessary condition, we see that the set of eigenvalues of $LHS_{k+1}$ must match the set of eigenvalues of $RHS_{k}$ up to a common proportionality factor.
The eigenvalues read\footnote{Note that eigenvalues may coincide and that $LHS_{k}$, $RHS_{k}$ may be not diagonalizable, even though $g_k$ is diagonalizable.}
\begin{align*}
\sigma(LHS_{k+1}) &= \left\{1, B_{11}^{(k+1)} \chi_{k+1}, B_{11}^{(k+1)}/\chi_{k+1}   \right\}  , \\
\sigma(RHS_{k}) &= \left\{B_{11}^{(k)}, \chi_k, 1/\chi_k \right\}.
\end{align*}

Let us remark here that considering the concatenation conditions, it is immediately clear that for any ${g}$, one can find an MPS that has the global symmetry ${g}^{\otimes N}$. The reason for that is that the concatenation condition ($b LHS_{0} b^{-1} \propto RHS_0$), in case of a global symmetry, i.e., 1-cycle) reduces to matching the set of eigenvalues (up to a common proportionality factor) \footnote{This can be easily seen choosing $B_{01}^0 =B_{02}^0=0$ and $B_{11}^0=0$.} For all matrices ${g}$ it is possible to find a proportionality factor and a choice of $B_{11}$  such that the eigenvalues match.

\subsection{Local symmetries of the TIMPS \texorpdfstring{$\Psi(LLT)$}{Psi(LLT)}}
In this subsection, we present a characterization of the symmetries of normal MPS generated by $\identity\otimes b \otimes \identity \ket{\Psi_0}$ and also discuss some of the non-normal MPS. Certain details of the derivation will be deferred to Appendix \ref{app:L1L1Tsym}.

Let us right away distinguish between the two cases $b_{20} = 0$ and $b_{20} \neq 0$. In the former case, $b_{20} = 0$, the generated MPS cannot be normal as we will see in Observation \ref{obs:L1L1Tnormality}. However, despite the fact that the fundamental theorem does not apply, actually much can be said about the symmetries of the corresponding MPS as we show in the following observation.
\begin{observation}
\label{obs:L1L1Tb20vansihes}
$N$-qubit MPS associated to $\identity\otimes b \otimes \identity \ket{LLT}$ with $b_{20} = 0$ are either SLOCC-equivalent to $\ket{0}^{\otimes N}$, $\ket{\mathrm{GHZ}_N}$, or they possess only global symmetries.
\end{observation}
Let us remark here that using $b= \identity$, i.e., using the seed state $\ket{LLT}$ as fiducial state gives rise to an MPS that is a product state.
\begin{proof}
In order to prove the observation, we consider the definition of an MPS as in Eq. (\ref{eq:mps}) and note that for $b_{20}=0$ it holds that 
\begin{align*}
\tr\{A^{j_1} A^{j_2} \ldots A^{j_N}\} = b_{22}^{|j|} b_{21}^{N - |j|} +  b_{00}^{|j|} b_{10}^{N - |j|},
\end{align*}
where $j = (j_1 ,\ldots, j_N) \in \{0,1\}^N$ and $|j|$ denotes the Hamming weight of $j$. In particular, the expression in Eq. (\theequation) does not depend on the order of the operators $A_{j_l}$. Thus, the MPS is, for $b_{20}=0$ not only translationally invariant, but actually invariant under any particle permutation. A permutation invariant $N$-qubit state is either SLOCC equivalent to $\ket{0}^{\otimes N}$, $\ket{\mathrm{GHZ}_N}$, or $\ket{W_N}$\footnote{Note that the W-state is not representable by a TIMPS of bond dimension three.}, or the state is what was called non--exceptionally symmetric \cite{sym}, meaning that all its symmetries are of the form $S^{\otimes N}$ \cite{BaKr09,MaKr10, MiRo13}.   In the former cases, the symmetries of the MPS are well-known and in the latter case (by definition of non--exceptionally symmetric states) the MPS possesses global symmetries only.
\end{proof}

We will focus on the case $b_{20} \neq 0$ for the remainder of this section. For simplicity, in the following we will assume $b_{20} = 1$, as an overall scaling factor within $b$ is irrelevant. Note, however, that $b_{20} \neq 0$ is not a sufficient condition to have normal MPS. In fact, normality additionaly depends on $\vec{v}$ as in Eq. (\ref{eq:vdef}), as the following observation shows. We prove the observation in Appendix \ref{app:L1L1Tnormality} (see also Appendix \ref{app:M10M11M1infnormality} for a few general remarks on proving normality).
\begin{observation}
\label{obs:L1L1Tnormality}
$N$-qubit MPS associated to $\identity\otimes b \otimes \identity \ket{LLT}$ are normal if and only if $b_{20}\neq 0$ and $\vec{v} \neq \vec{0}$.
\end{observation}
 We keep this fact in mind, however, in the following we will continue without narrowing down the considered set of fiducial states any further and postpone a more detailed discussion on normality till Section \ref{sec:L1L1Tnonnormal}. 

Let us now analyze the concatenation equations in more depth. The (2,0)-matrix element of the concatenation condition with proportionality factors $\lambda_k$, $b LHS_{k+1}  - \lambda_k RHS_{k} b = 0$, reads $b_{20} (1 - B_{11}^{(k)} \lambda_k) = 0$. As $b_{20} = 1$, this matrix element thus fixes the proportionality factors. We then obtain equality of the following two sets of eigenvalues as necessary condition,
\begin{align*}
 \left\{B_{11}^{(k)},B_{11}^{(k)} B_{11}^{(k+1)} \chi_{k+1}, B_{11}^{(k)} B_{11}^{(k+1)}/\chi_{k+1}  \right\}   \nonumber \\
 = \left\{B_{11}^{(k)}, \chi_k, 1/\chi_k \right\}.
\end{align*}
Particularly interesting are the trace and the determinant of the matrix equation in the concatenation condition, i.e., the sum and the product of the elements in the two sets in Eq. (\theequation). One obtains
\begin{align*}
\left(B_{11}^{(k)}\right)^2  \left( B_{11}^{(k+1)}\right)^2 &= 1,\\
 B_{11}^{(k)}  B_{11}^{(k+1)} (\chi_{k+1} + 1/\chi_{k+1}) &=   \chi_{k} + 1/\chi_{k}.
\end{align*}
Considering the chosen normalization, this implies that  the sets of eigenvalues of ${g}_k$ must coincide for all $k$, i.e., $\chi_{k} =  \chi_{l}$ for all $k, l$. We will thus drop the index $k$ in $\chi_k$, in the following.
If $\chi \neq i$, i.e., if $\tr {g}_k$ does not vansish, then $B_{11}^{(k)}  B_{11}^{(k+1)}=1$ for all $k$. Considering a cycle of odd length, one moreover has that either $B_{11}^{(k)}  = 1$ for all $k$, or $B_{11}^{(k)}  = -1$ for all $k$. To see this, note that using $B_{11}^{(k)}  B_{11}^{(k+1)}=1$ recursively yields $B_{11}^{(k)}  =1/B_{11}^{(k)}$. In the case that  $\tr {g}_k= 0$, we only have  $B_{11}^{(k)}  B_{11}^{(k+1)}= \pm 1$, instead. We summarize these findings in the following observation
\begin{observation}
\label{obs:L1L1Teigenvalues}
Suppose that ${g}_0, \ldots, {g}_{N-1}$ is an $N$-cycle in $G_b$. Then, the eigenvalues of ${g}_k$, $\chi$ and $1/\chi$, coincide for all $k$.  If $\chi \neq i$, we have $B_{11}^{(k)}  B_{11}^{(k+1)}=1$. If $\chi = i$, we have $B_{11}^{(k)}  B_{11}^{(k+1)} = \pm1$ in Eq. (\ref{eq:lhsrhs}).
\end{observation}

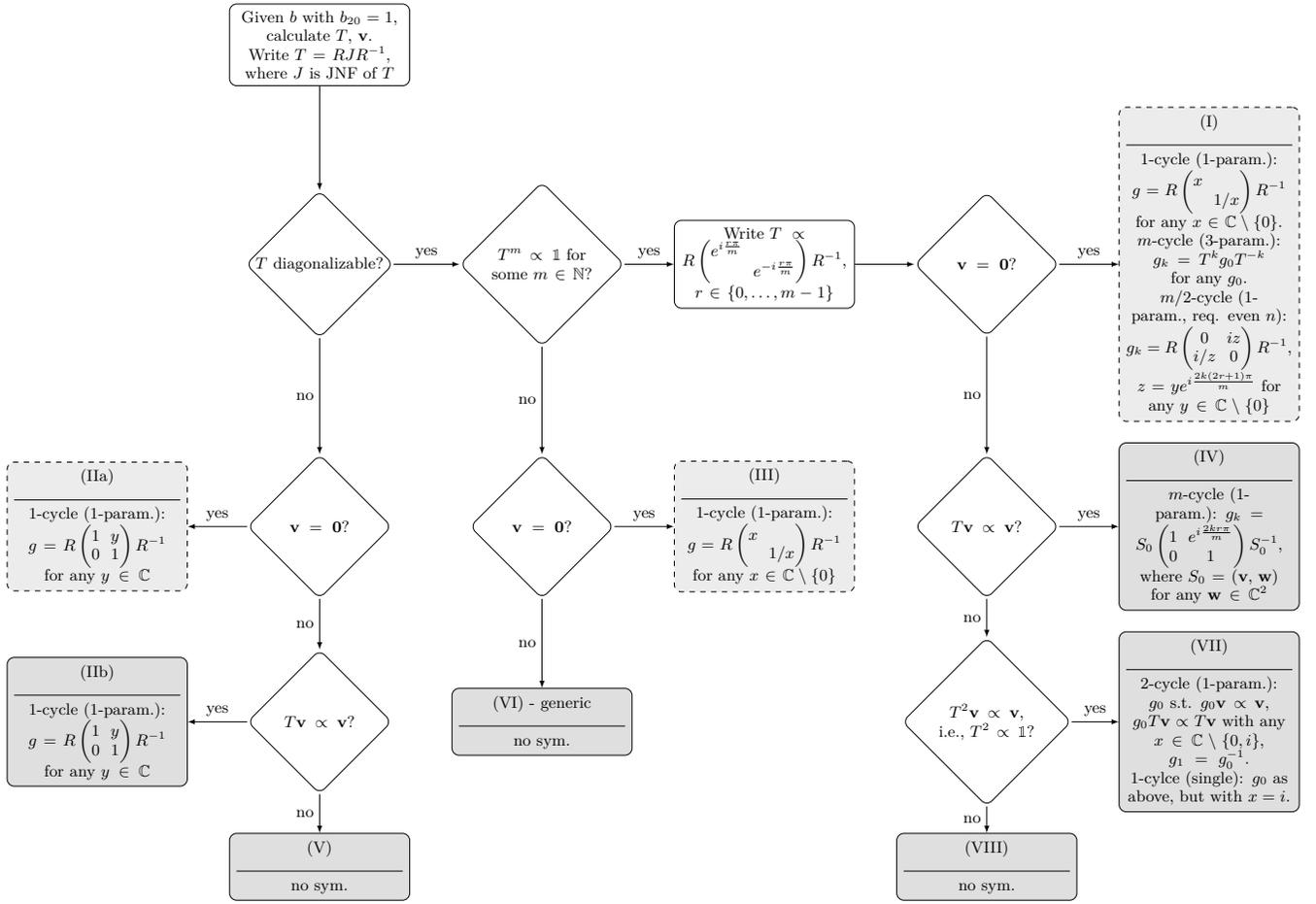
\begin{figure*}[t]
\resizebox{\textwidth}{!}{%
\begin{tikzpicture}[-latex]
  \matrix (chart)
    [
      matrix of nodes,
      column sep      = 2.5em,
      row sep         = 3ex,
      ampersand replacement=\&
    ]
    {
                               \& |[root]| Given $b$ with $b_{20}=1$, calculate $T$, $\vec{v}$.  Write $T = R J R^{-1}$, where $J$ is JNF of $T$     \& \&  \&         \&    \\
                               \& |[decision]| $T$ diagonalizable?                                    \& |[decision]| $T^m \propto \identity$ for some $m\in\mathbb{N}$? \&   |[treenode]| Write $T \propto R \begin{pmatrix} e^{i \frac{r \pi}{m}} & \\ &  e^{-i \frac{r \pi}{m}}\end{pmatrix} R^{-1}$, $r \in \{0, \ldots, m-1\}$  \& |[decision]| $\vec{v}=\vec{0}$?  \&  |[finishnonnormal]| (I) \rule{\textwidth}{0.2pt} 1-cycle (1-param.): ${g} = R \begin{pmatrix}x&\\&1/x \end{pmatrix} R^{-1}$ for any $x\in \mathbb{C}\setminus\{0\}$. \linebreak $m$-cycle (3-param.): ${g}_k= T^k g_0 T^{-k}$ for any ${g}_0$. \linebreak $m/2$-cycle (1-param., req. even $n$): ${g}_k = R  \begin{pmatrix}0 & i z\\ i/z & 0 \end{pmatrix} R^{-1}$, $z = y e^{i \frac{2k (2r+1) \pi}{m}}$ for any $y \in \mathbb{C}\setminus\{0\}$\\
      |[finishnonnormal]| (IIa) \rule{\textwidth}{0.2pt} 1-cycle (1-param.): ${g} = R \begin{pmatrix}1&y\\0&1 \end{pmatrix} R^{-1}$ for any $y \in \mathbb{C}$   \&  |[decision]| $\vec{v}=\vec{0}$?  \&  |[decision]| $\vec{v}=\vec{0}$? \&  |[finishnonnormal]| (III)   \rule{\textwidth}{0.2pt} 1-cycle (1-param.): ${g} = R \begin{pmatrix}x&\\&1/x \end{pmatrix} R^{-1}$ for any $x\in \mathbb{C}\setminus\{0\}$ \&  |[decision]| $T \vec{v}\propto \vec{v}$? \&  |[finish]| (IV)  \rule{\textwidth}{0.2pt} $m$-cycle (1-param.): ${g}_k = S_0 \begin{pmatrix}1&e^{i \frac{2 k r \pi}{m}}\\0&1 \end{pmatrix} S_0^{-1}$, where $S_0 = (\vec{v},\, \vec{w})$ for any $\vec{w} \in \mathbb{C}^2$ \\
          |[finish]| (IIb) \rule{\textwidth}{0.2pt} 1-cycle (1-param.): ${g} = R \begin{pmatrix}1&y\\0&1 \end{pmatrix} R^{-1}$ for any $y \in \mathbb{C}$                 \&   |[decision]| $T \vec{v}\propto \vec{v}$? \&   |[finish]| (VI) - generic \rule{\textwidth}{0.2pt} no sym.   \&  \&  |[decision]| $T^2 \vec{v} \propto \vec{v}$, i.e., $T^2 \propto \identity$?  \&  |[finish]| (VII)  \rule{\textwidth}{0.2pt} 2-cycle (1-param.): ${g}_0$ s.t. ${g}_0 \vec{v} \propto \vec{v}$, ${g}_0 T\vec{v} \propto T\vec{v}$ with any $x \in \mathbb{C}\setminus\{0,i\}$, ${g}_1 = {g}_0^{-1}$. \linebreak 1-cylce (single): ${g}_0$ as above, but with $x=i$. \\
                               \&         |[finish]| (V) \rule{\textwidth}{0.2pt} no sym.                      \& \&  \&  |[finish]| (VIII) \rule{\textwidth}{0.2pt} no sym. \&   \\
    };
  \draw
    (chart-1-2) edge (chart-2-2)    
    (chart-2-2) \no (chart-3-2)    
    (chart-2-2) \yes (chart-2-3)  
    (chart-3-2) \yes (chart-3-1)
    (chart-3-2) \no (chart-4-2) 
    (chart-2-3) edge node [above] {yes} (chart-2-4)
    (chart-2-4) edge  (chart-2-5)
    (chart-2-5) edge node [above] {yes} (chart-2-6)
    (chart-2-5) \no (chart-3-5)    
    (chart-2-3) \no (chart-3-3)
    (chart-3-3) \yes (chart-3-4)
    (chart-3-3) \no (chart-4-3)
    (chart-3-5) \no (chart-4-5)    
    (chart-3-5) \yes (chart-3-6)
    (chart-4-2) \yes (chart-4-1)
    (chart-4-5) \yes (chart-4-6)
    (chart-4-2) \no (chart-5-2)    
    (chart-4-5) \no (chart-5-5)    ;
\end{tikzpicture}
}
\caption{Flowchart showing the characterization of the symmetries of all normal MPS generated by $\identity \otimes b \otimes \identity \ket{LLT}$, i.e., for any $b$ with $b_{20} \neq 0$ (wlog $b_{20} = 1$) and $\vec{v}\neq \vec{0}$ (shaded rectangles with solid contour), cf. Observation \ref{obs:L1L1Tnormality}. For any such $b$, the symmetries of the corresponding MPS may be determined by calculating $T$ and $\vec{v}$ as in Eqs. (\ref{eq:Tdef}) and (\ref{eq:vdef}) and then following the procedure described in the main text, which is shown in the flowchart. Additionally, the flowchart also shows the cycles in $G_b$ obtained for non-normal MPS generated by $b$ such that $b_{20} = 1$ and $\vec{v} = \vec{0}$ (shaded rectangles with dashed contour). Note that for non-normal MPS the symmetry group might be larger than displayed, as the utilized methods may fail to identify the full symmetry group (and yield a subgroup instead).
Here, ``no sym.'' indicates that the corresponding MPS possesses only the trivial symmetry. Generic $b$ belong to box (VI), as indicated in the flowchart.}
\label{fig:L1L1Tflowchart}
\end{figure*}

Building on the observations above and making use of $T$ and $\vec{v}$ as in Eqs. (\ref{eq:Tdef}) and (\ref{eq:vdef}), we derive the following theorem, which gives necessary and sufficient conditions for ${g}_0, \ldots, {g}_{N-1}$ forming an $N$-cycle in $G_b$. We prove the theorem in Appendix \ref{app:L1L1Tsym}. Note that, as we will see below, this leads to a rich variety of situations involving 1-cycles as well as $N$-cycles, diagonalizable $g_k$ as well as non-diagonalizable $g_k$ and single cycles, as well as continuous families of cycles.

\begin{theorem}
\label{thm:L1L1Tconditions}
${g}_0, \ldots, {g}_{N-1}$ is an $N$-cylce in $G_b$ \footnote{Recall that we consider here $b_{20} = 1$. We have dealt with the case $b_{20} = 0$ in Observation \ref{obs:L1L1Tb20vansihes}.} if and only if there exist $B_{11}^{(k)} \in \mathbb{C}$ such that for all $k \in \{0, \ldots, N-1\}$,
\begin{align}
\label{eq:conjugationrule}
{g}_{k+1} &= \frac{1}{B_{11}^{(k)}B_{11}^{(k+1)}} T {g}_k T^{-1},\\
\label{eq:vrule}
\left[ {g}_k - {B_{11}^{(k)}} \identity \right]  \vec{v} &= \vec{0},\\
B_{11}^{(k)}B_{11}^{(k+1)} &= \begin{cases} \pm1 & \text{if } \chi = i \\ 1 &\text{otherwise}\end{cases}
\end{align}
\end{theorem}

With the help of the conditions provided in Theorem \ref{thm:L1L1Tconditions}, the $N$-cycles in $G_b$ may be determined for any given $b$ with $b_{20} =1$.  In the following, we describe the procedure to do so. We defer the details on the derivation of the procedure to Appendix \ref{app:L1L1Tsym}. Recall that the considered family of $b$'s also involves non-normal MPS. For this reason, we have formulated the theorem in terms of cycles in $G_b$, although for normal MPS the theorem directly characterizes the symmetries of the associated MPS.

 First, one calculates the matrix $T$ according to Eq. (\ref{eq:Tdef}) as well as the vector $\vec{v}$ according to Eq. (\ref{eq:vdef}). The symmetries will be completely determined by $T$ and $\vec{v}$, which, as we would like to stress here again, are merely properties of $b$. Let us denote the similarity transformation bringing $T$ into its Jordan Normal Form (JNF) $J$ by $R$, i.e., we have $T = R J R^{-1}$.
 We now distinguish two cases. We have the case that $T$ is diagonalizable and the case that $T$ is not diagonalizable. In the latter case, we obtain only trivial cycles if $\vec{v} \neq \vec{0}$ and $T \vec{v} \not\propto \vec{v}$. In contrast to that, $G_b$ exhibits a one-parametric family of 1-cycles with ${g}= R \begin{pmatrix} 1& \eta\\ 0 &1\end{pmatrix} R^{-1}$ for any $\eta \in \mathbb{C}$ if $T \vec{v} \propto \vec{v}$ (or $\vec{v}= \vec{0}$).

Let us now discuss the case that $T$ is diagonalizable. We now distinguish two further cases depending on whether there exists an $m \in \mathbb{N}$ such that $T^m \propto \identity$, or not. In case such an $m$ does not exist, we distinguish several subcases depending on $\vec{v}$. If $\vec{v} = \vec{0}$ we obtain a one-parametric family of 1-cycles  with ${g}= R \begin{pmatrix} \chi& \\  &1/\chi \end{pmatrix} R^{-1}$ for any $\chi \in \mathbb{C} \setminus \{0\}$. In contrast, if $T \vec{v} \not\propto \vec{v}$, but $T^2 \vec{v} \propto \vec{v}$ we 
obtain a single 1-cycle with ${g}= R \begin{pmatrix} i& \\  &-i \end{pmatrix} R^{-1}$. We obtain only trivial cycles for all other $\vec{v}$. Note that a generic $b$ falls into this category.

Let us now discuss the case that there exists an $m \in \mathbb{N}$ such that $T^m \propto \identity$. In this case, we may write $T \propto R \begin{pmatrix} e^{i \frac{r \pi}{m}} & \\ & e^{i \frac{r \pi}{m}} \end{pmatrix} R^{-1}$ for some $r \in \{0, \ldots, m-1\}$. Again, we distinguish several subcases depending on the vector $\vec{v}$. First, let us consider the case that $\vec{v} = \vec{0}$. In this case we obtain a rich set of cycles in $G_b$. Actually, we obtain $m$-cycles with ${g}_k = T^k {g}_0 T^{-k}$ for any ${g}_0$. This is effectively a three-parametric family of cycles including both, instances in which the ${g}_k$ are diagonalizable, as well as instances in which ${g}_k$ are not diagonalizable. In case $m$ is even, in addition to that a one-parametric family of $m/2$-cycle emerges. There, we have ${g}_k = R  \begin{pmatrix}0 & i \eta e^{i \frac{2k (2r+1) \pi}{m}}\\ i/ \eta e^{-i \frac{2k (2r+1) \pi}{m}} & 0 \end{pmatrix} R^{-1}$, for any $y \in \mathbb{C}\setminus\{0\}$. Note that $\tr {g}_k = 0$. Let us now discuss the case that we have a non-vanishing $\vec{v}$ with $T \vec{v} \propto \vec{v}$. In this case we obtain an effectively one-parametric $m$-cycle of non-diagonalizable cycles with ${g}_k = S_0 \begin{pmatrix}1 & e^{i \frac{2 k r \pi}{m}} \\ 0 & 1 \end{pmatrix} S_0^{-1}$, where $S_0 = \left(\vec{v}, \vec{w} \right)$ for any $\vec{w} \in \mathbb{C}^2$. In other words, ${g}_0$ may be chosen as any non-diagonalizable matrix whose eigenvector is given by $\vec{v}$, the remaining matrices are then determined.
 Let us now discuss the case that we have a non-vanishing $\vec{v}$ with $T \vec{v} \not\propto \vec{v}$, but $T^2 \vec{v} \propto \vec{v}$. Note that this implies $T^2\propto \identity$. In this case, we obtain cycles with diagonalizable ${g}_k$. The eigenvectors of each ${g}_k$ are given by $\vec{v}$ and $T\vec{v}$. We obtain global cycles, in which the eigenvalues of ${g}$ are $\pm i$. Moreover, in case of an even particle number, we obtain two-cycles with ${g}_1 = {g}_0^{-1}$ and a freely choosable eigenvalue $\chi \neq 0,\pm i$. Finally, in case of a non-vanishing $\vec{v}$ with $T^2 \vec{v} \not\propto \vec{v}$ we obtain no non-trivial cycles. This completes the characterization of cycles within $G_b$ considering the fiducial states $\identity \otimes b \otimes \identity \ket{LLT}$. We present a summary of the findings in terms of a flowchart in Figure \ref{fig:L1L1Tflowchart}.
 
 Let us conclude with remarking that for any specified $T$ and $\vec{v}$, there is a two-parametric family of $b$'s (with $b_{20}=1$) leading to the specified $T$, $\vec{v}$, as in Eq. (\ref{eq:L1L1Tparam}).  Thus, it is possible to construct a $b$  possessing any desired symmetry presented in Figure \ref{fig:L1L1Tflowchart} using the appropriate $T$ and $\vec{v}$. Moreover, for normal MPS, Theorem \ref{thm:L1L1Tconditions} characterizes all possible symmetries.

\subsection{SLOCC Classification}
\label{sec:L1L1T_SLOCC}
In order to identify the different SLOCC classes emerging within the normal MPS associated to $\ket{LLT}$, we consider $(b \rightarrow c)$ cycles within the symmetry group of the fiducial state. More precisely, we study the relation
\begin{align}
y_{k} b x_{k+1} \propto c \text{ for all } k \in \{0, \ldots, N-1\}
\end{align}
in order to decide whether the MPS generated by $\identity \otimes b \otimes \identity \ket{LLT}$ and $\identity \otimes c \otimes \identity \ket{LLT}$ are SLOCC equivalent to each other. As shown in \cite{SaMo19} (see also Sec. \ref{sec:prelim} and Sec. \ref{sec:M10M11M1inf}), they are SLOCC equivalent to each other iff it is possible to identify an $N$-cycle (or an $M$-cycle, where $M$ divides the total particle number $N$). 
We will first characterize 1-cycles. Then, we will introduce a (non-unique) standard form for $b$ and $c$ up to global SLOCC operations. Finally, we complete the classification by considering non-global operations. Note that we charcterize the $(b\rightarrow c)$ cycles for all $b$, $c$ with $b_{20},c_{20} \neq 0$, however, we keep in mind that certain such $b,c$ lead to non-normal MPS. We present the SLOCC classification in the flowchart shown in Figure \ref{fig:L1L1TflowchartSLOCC}, following the same structure as in Figure \ref{fig:L1L1Tflowchart}.

Recall that two states $\ket{\psi}$ and $\ket{\phi}$ can only be SLOCC equivalent if their symmetry group is compatible, i.e., $\mathcal{S}_{\ket{\psi}}$  equals $\mathcal{S}_{\ket{\phi}}$ up to conjugation. This immediately shows that, e.g., states belonging to the box (IV) cannot be SLOCC equivalent to states belonging to box (V) in Figure \ref{fig:L1L1Tflowchart}. However, this necessary condition is not strong enough to reveal anything about SLOCC equivalence between, e.g., states belonging to boxes (V) and (VI) within the figure, yet.

\subsubsection{Global SLOCC operations and standard form}
As a first step, we investigate $(b \rightarrow c)$ 1-cycles, which allow to characterize equivalence of normal MPS under global operations. For normal MPS, stated differently, we characterize here all $b$ for which there exists an operator $g$ such that $\Psi_b(LLT)=g^{\otimes n} \Psi_c(LLT)$ for a given $c$. Using the symmetry of the fiducial state [see Eq. (\ref{Eq:xy})] this leads to the following. 
A $(b \rightarrow c)$ 1-cycle exists if
\begin{align*}
b \propto x^{-1} c y^{-1},
\end{align*}
where 
\begin{align*}
x^{-1}&= \begin{pmatrix} \delta & \gamma & B_{02} \gamma +  B_{01} \delta \\ \beta & \alpha &  B_{02} \alpha +  B_{01} \beta \\ 0&0&B_{11}\end{pmatrix}\\
y^{-1}&=\frac{1}{B_{11}}\begin{pmatrix} B_{11} & -(B_{02} \gamma +  B_{01} \delta)  & -(B_{02} \alpha +  B_{01} \beta) \\ 0 & \delta &  \beta \\ 0& \gamma & \alpha\end{pmatrix},
\end{align*}
where $\alpha,\beta,\gamma,\delta,B_{11},B_{01},B_{02} \in \mathbb{C}$ such that $\det g = 1$, where $g = \begin{pmatrix}\alpha & \beta\\ \gamma & \delta \end{pmatrix}$.
We use the parametrization and normalization of $b$ and $c$ as in  Eq. (\ref{eq:L1L1Tparam}) and write
\begin{align*}
b = b[T_b, \vec{v}_b,\begin{pmatrix}b_{10}\\b_{00}\end{pmatrix}]
\end{align*}  
and similarly for $c$. We obtain all $b$ that are connected to $c$ via a $(b \rightarrow c)$ 1-cycle through 
\begin{align} \label{Eq:BSL}
b = b[\frac{1}{B_{11}^2} g T_c g^{-1},\frac{1}{B_{11}}  g \vec{v}_c,\frac{1}{B_{11}} g \begin{pmatrix}c_{10}+B_{02}\\c_{00}+B_{01}\end{pmatrix} ]
\end{align}
for $g$ with $\det g = 1$ and $B_{11},B_{01},B_{02} \in \mathbb{C}$. Here, $g$ is the global physical operation relating the two MPS. Note that even $g = \identity$ leads to a freedom in $b$, which  is due to symmetries of the fiducial state that have the form $\identity \otimes B \otimes C$. Note also that in Eq. (\theequation) we have equality and not proportionality as $b_{20} = c_{20} = 1$ fixes the proportionality factor to 1. Note further that Eq. (\ref{Eq:BSL}) allows to easily identify global LU-invariant quantities. 

We introduce a standard form for $b$, $c$ up to global operations (1-cycles). It then suffices to study SLOCC equivalence for MPS associated to $b$, $c$ which are in standard form in order to provide a full characterization of SLOCC equivalence. We choose the following standard form:
\begin{align*}
b = b[T_b, \vec{v}_b, 0],
\end{align*}
where $T_b$ is in JNF, $\det T_b = 1$. Moreover, we use the same convention for the ordering and possible sign flip of the eigenvalues as earlier in this section. More precisely, for diagonalizable $T_b$ we write $T_b = \operatorname{diag}(\sigma_b,\sigma_b^{-1})$, where $\sigma_b \in \mathcal{D}$ (see Figure \ref{fig:domain}).  Note that the standard form is not unique as we may flip the direction of $\vec{v}$ via a sign change in $B_{11}$ and moreover, special forms of $T_b$ such as e.g. $T_b = \identity$ leave even more freedom to choose the direction of $\vec{v}$. Let us stress here that $b$ with coinciding $T$ and non-vanishing $\vec{v}$ whose directions coincide, but whose norms differ, lead to MPS that share the same symmetry group, but are not necessarily related by a global SLOCC operation. 
Clearly, if $b$ and $c$ which are connected by a $(b \rightarrow c)$ 1-cycle are in standard form, we necessarily have that $T_c = T_b$.

For normal tensors, the characterization of $(b \rightarrow c)$ 1-cycles allows to characterize equivalence of the associated MPS under global SLOCC operations (for non-normal tensors, additional MPS might turn out to be equivalent, which are not identified as equivalent by considering $(b\rightarrow c)$ cycles) \cite{SaMo19}. Due to the considerations of $(b \rightarrow c)$ 1-cycles above, we obtain such a characterization as stated in the following lemma. 
\begin{lemma}
\label{lemma:L1L1T_global_SLOCC}
Consider fiducial states  $\identity \otimes b \otimes \identity \ket{LLT}$ and $\identity \otimes c \otimes \identity \ket{LLT}$ which correspond to normal MPS (i.e., $b_{20} = c_{20} = 1$ and additionally $\vec{v}_b,\vec{v}_c \neq \vec{0}$). Then, the MPS are related via a global SLOCC operation if and only if there exists a $g \in SL(2,\mathbb{C})$ such that $\vec{v}_b \propto g \vec{v}_c$  and $T_b = \frac{\vec{v}_b^T \vec{v}_b}{(g \vec{v}_c)^T g \vec{v}_c} g T_c g^{-1}$. 
\end{lemma}
\begin{proof}
The statement follows from the considerations of $(b \rightarrow c)$ 1-cycles above.
\end{proof}

\subsubsection{Non-global SLOCC operations}
Let us now also take non-global SLOCC operations into account. 
Considering
\begin{align*}
b \propto x_k^{-1} c y_{k+1}^{-1},
\end{align*}
and imposing that both $b$ and $c$ are in standard form we obtain $B_{01}^{(k)} = B_{02}^{(k)} = 0$ for all $k$ and
\begin{align}
\label{eq:L1L1Tslocctrafo}
b = b[\frac{1}{B_{11}^{(k)}B_{11}^{(k+1)}} g_{k+1} T_c g_{k}^{-1}, \frac{1}{B_{11}^{(k+1)}} g_{k+1} \vec{v}_c, 0],
\end{align}
where we use the normalization $\det g_k = 1$.
We obtain as a simple necessary condition for having an $N$-cycle
\begin{align}
\label{eq:L1L1Tsloccnec}
T_b^N = \pm g_k T_c^N g_k^{-1}
\end{align}
for all $k$ (with a positive sign in case of even $N$, $\pm$ in case of odd $N$). 

Using Eqs. (\ref{eq:L1L1Tslocctrafo}) and (\ref{eq:L1L1Tsloccnec}) it is straightforward to establish that within Figure \ref{fig:L1L1TflowchartSLOCC}, fiducial states that belong to different boxes do allow for a $(b\rightarrow c)$ cycle (see Observation \ref{obs:L1L1TSLOCCobs} in Appendix \ref{app:L1L1TSLOCC}).

Let us now complete the characterization of $(b\rightarrow c)$ cycles. In the case that $T_c^N \not\propto \identity$, considering Eq. (\ref{eq:L1L1Tsloccnec}), the standard form for $T_b,T_c$, and the uniqueness of the Jordan decomposition straightforwardly leads to the fact that all $g_k$ must be in Jordan normal form (special care needs to be taken in case $\operatorname{tr}T_c = 0$). Then, using Eq. (\ref{eq:L1L1Tslocctrafo}) in addition, a tedious calculation shows that $g_k = g$ for all $k$. However, the case $T_c^N \propto \identity$ is more involved as in this case, the condition in Eq. (\ref{eq:L1L1Tsloccnec}) is not helpful. Let us thus take intermediate steps in completing the characterization of $(b\rightarrow c)$ cycles. To this end, we will introduce two lemmata, which we prove in Appendix \ref{app:L1L1TSLOCC}. It is obvious that whenever we have $b$ and $c$ in standard form allowing for a $(b \rightarrow c)$ 1-cycle, it holds that $T_b = T_c$. The first lemma shows that the same is true for $(b \rightarrow c)$ $N$-cycles if $\vec{v}_b \neq \vec{0}$, $\vec{v}_c \neq \vec{0}$.

\begin{lemma}
\label{lemma:L1L1TSLOCCsameT}
Consider $b$ and $c$ in standard form which correspond to normal MPS (i.e., $b_{20} = c_{20} = 1$ and additionally $\vec{v}_b,\vec{v}_c \neq \vec{0}$). If there exists a  $(b\rightarrow c)$ $N$-cycle, then, $T_b = T_c$. 
\end{lemma}

Building on Lemma \ref{lemma:L1L1TSLOCCsameT}, the next lemma shows that whenever such $b$ and $c$ are connected by a $(b \rightarrow c)$ $N$-cycle, there also exists a $(b \rightarrow c)$ $1$-cycle.

\begin{lemma}
\label{lemma:L1L1TSLOCCglobalsuffices}
Consider $b$ and $c$ which correspond to normal MPS (i.e., $b_{20} = c_{20} = 1$ and additionally $\vec{v}_b,\vec{v}_c \neq \vec{0}$). If there exists a  $(b\rightarrow c)$ $N$-cycle, then, there also exists a $(b\rightarrow c)$ $1$-cycle. 
\end{lemma}

We are now in the position to state simple necessary and sufficient conditions for SLOCC equivalence of normal MPS generated by fiducial states within the $LLT$ class.
\begin{theorem}
\label{thm:L1L1TSLOCC}
Consider fiducial states  $\identity \otimes b \otimes \identity \ket{LLT}$ and $\identity \otimes c \otimes \identity \ket{LLT}$ which correspond to normal MPS (i.e., $b_{20} = c_{20} = 1$ and additionally $\vec{v}_b,\vec{v}_c \neq \vec{0}$). Then, the MPS are SLOCC equivalent if and only if they are related via a global operation, i.e., there exists a $g \in SL(2,\mathbb{C})$ such that $\vec{v}_b \propto g \vec{v}_c$  and $T_b = \frac{\vec{v}_b^T \vec{v}_b}{(g \vec{v}_c)^T g \vec{v}_c} g T_c g^{-1}$. 
\end{theorem}

Let us remark here that the operator $g$ in the Theorem is such that $g^{\otimes N}$ transforms one state into the other. 
\begin{proof}
The statement of the theorem follows directly from Lemma \ref{lemma:L1L1TSLOCCglobalsuffices} together with the considerations on global SLOCC operations [$(b\rightarrow c)$ 1-cycles] in Lemma \ref{lemma:L1L1T_global_SLOCC}.
\end{proof}

A straightforward consequence of the theorem is that SLOCC equivalence of the considered MPS is not particle-number dependent, in spite of all the variety within their ($N$-dependent) symmetry group. Note that this is not true in general, see e.g. the SLOCC classes for MPS with bond dimension $D=2$ that are generated by fiducial states within the GHZ class \cite{SaMo19}.

\subsubsection{Representatives and parametrization of SLOCC classes}

Due to Theorem \ref{thm:L1L1TSLOCC} we have that two MPS are SLOCC equivalent iff they are related by a global transformation. Here, we parameterize all SLOCC classes by introducing a more precise standard form for the various $b$'s, i.e. the fiducial states. This standard form is then also useful to obtain representative MPS for the different SLOCC classes. The resulting representatives of the SLOCC classes are presented in the flowchart in Fig.~\ref{fig:L1L1TflowchartSLOCC}. This completes the characterization of all SLOCC classes of TIMPS corresponding to the fiducial states which are SLOCC equivalent to the $\ket{LLT}$ state. 

As mentioned before (see Observation \ref{obs:L1L1TSLOCCobs} in Appendix \ref{app:L1L1TSLOCC}) SLOCC equivalent normal TIMPS must belong to the same box in Fig.~\ref{fig:L1L1Tflowchart}. We obtain the parameterization of the SLOCC classes by introducing a precise standard form of all operators $b$ corresponding to the individual boxes. To this end we consider the operators $b$ and $c$ which have the same standard form as given above. We have seen that $T$ may be brought into Jordan normal form, and normalized to determinant 1. Let us now further specify $T$ and the vector $\vec{v}$ for the various cases (boxes).

Let us first consider the scenario that $T_b$ and $T_c$ are not diagonalizable in more details. Due to the chosen standard form, we then have $T_b= T_c = \begin{pmatrix}1&1\\0&1 \end{pmatrix}$. Moreover, due to $T_b \propto g T_c g^{-1}$ we have that 
\begin{align*}
g = \pm  \begin{pmatrix}1&\beta\\0&1 \end{pmatrix}.
\end{align*}
Note that these are the only global transformations, which map normal MPSs with fiducial states $\one \otimes b \otimes \one \ket{LLT}$, with $b$ such that $T_b$ is non--diagonalizable and in standard from into each other. We hence have that $b$ and $c$ (in standard form) such that $T_b$ and $T_c$ are not diagonalizable lead to MPS that are SLOCC related if and only if 
\begin{align*}
\vec{v}_b = \begin{pmatrix}v_{b,0} \\ v_{b,1}\end{pmatrix} = \pm  \begin{pmatrix}v_{c,0} + \beta v_{c,1} \\ v_{c,1}\end{pmatrix}
\end{align*}
for some $\beta \in \mathbb{C}$. In order to take into account this freedom, we amend the definition of the standard form of $b$ by additionally requiring that either $\vec{v} \propto (1,0)^T$, or $\vec{v} \propto (0,1)^T$ (or $\vec{v} = \vec{0}$). Then, we have that $b$ and $c$ in standard form with non-diagonlizable $T_b$, $T_c$ correspond to MPS that are related by a global SLOCC operation if and only if  $\vec{v}_b = \pm \vec{v}_c$. The standard form for $b$ may now be used to obtain MPS that are representatives for the present SLOCC classes. Contemplating the characterization of symmetries, we have that there is a one-parametric\footnote{Recall that we count complex parameters.} family of SLOCC classes exhibiting a non-trivial global symmetry with $\vec{v}_b \propto (1,0)^T$ (the proportionality factor is the free complex parameter). More precisely, all states which belong to the SLOCC class can be transformed into the standard form with $\vec{v}_b = \pm x (1,0)^T$, and they belong to different SLOCC classes for different values of $x\in \C$.  

Supposing $\vec{v} \neq \vec{0}$, these classes correspond to the box (IIb) in Figure \ref{fig:L1L1Tflowchart}. Moreover, we find a one-parametric family with trivial symmetry group for $\vec{v} \propto (0,1)^T$, $\vec{v}  \neq \vec{0}$, which corresponds to box (V) within Figure \ref{fig:L1L1Tflowchart}. More precisely, all states with $\vec{v}_c = \pm x (\beta ,1)^T$  belong to the SLOCC class with $\vec{v}_b = x (0 ,1)^T$ for arbitrary $\beta \in \C$ and fixed $x\in \C$.

Let us now discuss the case that $T_c$ and $T_b$ are diagonalizable, i.e.,  $T_c = T_b = \operatorname{diag}(\sigma, 1/\sigma)$ for some $\sigma \in \mathbb{C}$. In the case that $\sigma = 1$, the MPS are SLOCC-equivalent if and only if there exists a $g$ with $\det g = 1$ such that $\vec{v}_b = g \vec{v}_c$. We thus choose the standard form $\vec{v} = (1,0)^T$. For $\sigma = i$, we obtain a $(b\rightarrow c)$ 1-cycle if and only if there exists an $g = \operatorname{diag}(\alpha, 1/\alpha)$ such that $\vec{v}_b = \pm g \vec{v}_c$, or $\vec{v}_b = \pm i g \sigma_x \vec{v}_c$. For $\sigma \neq 1,i$, we obtain a $(b\rightarrow c)$ 1-cycle if and only if there exists an $g = \operatorname{diag}(\alpha, 1/\alpha)$ such that $\vec{v}_b = \pm g \vec{v}_c$. In case we have $T \vec{v} \propto \vec{v}$, we thus choose the standard form such that either $\vec{v} = (1,0)^T$, or $\vec{v} = (0,1)^T$. Otherwise, we choose the standard form $\vec{v} = (v_1,1/v_1)^T$ for $v_1 \in \mathbb{C}\setminus \{0\}$.
Due to these considereations, there is a two-parametric family of SLOCC classes corresponding to box (VI) within Figure \ref{fig:L1L1Tflowchart} (one free parameter within $T$ plus one free parameter within $\vec{v}$). For the remaining boxes containing normal MPS, i.e., boxes (IV),  (VII), and (VIII), let us suppose that $T^m \propto \identity$ for some fixed $m$. Then, there is a discrete number of possible $T$s, namely $\lfloor m/2 \rfloor +1$ (see Figure \ref{fig:domain}). For a fixed $m$, there are exactly $m$ different SLOCC classes corresponding to box (IV), as for each of the $\lfloor m/2 \rfloor +1$ possible $T$'s we either have $\vec{v} = (1,0)^T$, or $\vec{v} = (0,1)^T$, except for $T = \identity$ and $T = \operatorname{diag}(-i,i)$, where these possibilities are equivalent and we hence choose $\vec{v} = (1,0)^T$. Particularly interesting will be the subfamily corresponding box (IV) for which $m=N$, as this subfamily encompasses all type-(IV) MPS with non-trivial symmetry group. Due to the reasonings above, there are exactly $N$ of them.   
Finally, there are one-parametric families of SLOCC classes corresponding to boxes (VII) and (VIII), respectively, with the free parameter stemming from $\vec{v}$. We summarize these findings in Figure \ref{fig:L1L1TflowchartSLOCC} following the same structure as in Figure \ref{fig:L1L1Tflowchart}, which displays the corresponding symmetries.

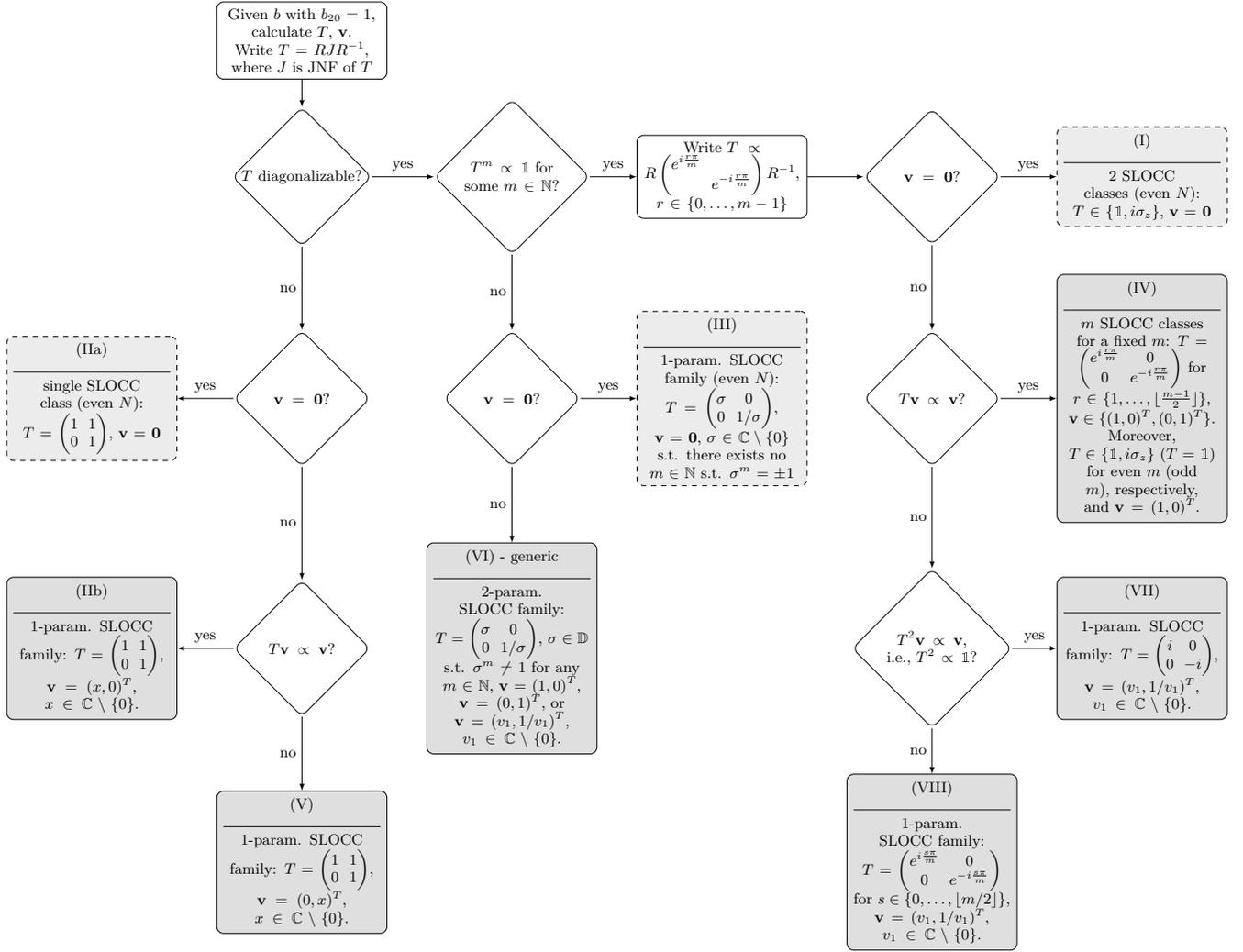
\begin{figure*}[t]
\resizebox{\textwidth}{!}{%
\begin{tikzpicture}[-latex]
  \matrix (chart)
    [
      matrix of nodes,
      column sep      = 2.5em,
      row sep         = 3ex,
      ampersand replacement=\&
    ]
    {
                               \& |[root]| Given $b$ with $b_{20}=1$, calculate $T$, $\vec{v}$.  Write $T = R J R^{-1}$, where $J$ is JNF of $T$     \& \&  \&         \&    \\
                               \& |[decision]| $T$ diagonalizable?                                    \& |[decision]| $T^m \propto \identity$ for some $m\in\mathbb{N}$? \&   |[treenode]| Write $T \propto R \begin{pmatrix} e^{i \frac{r \pi}{m}} & \\ &  e^{-i \frac{r \pi}{m}}\end{pmatrix} R^{-1}$, $r \in \{0, \ldots, m-1\}$  \& |[decision]| $\vec{v}=\vec{0}$?  \&  |[finishnonnormal]| (I) \rule{\textwidth}{0.2pt} 2 SLOCC classes (even $N$): $T \in \{\identity, i \sigma_z\}$,  $\vec{v}= \vec{0}$\\
      |[finishnonnormal]| (IIa) \rule{\textwidth}{0.2pt}   single SLOCC class (even $N$):  $T= \begin{pmatrix}1&1\\0&1 \end{pmatrix}$, $\vec{v}= \vec{0}$    \&  |[decision]| $\vec{v}=\vec{0}$?  \&  |[decision]| $\vec{v}=\vec{0}$? \&  |[finishnonnormal]| (III)   \rule{\textwidth}{0.2pt} 1-param. SLOCC family (even $N$): $T= \begin{pmatrix}\sigma&0\\0&1/\sigma \end{pmatrix}$, $\vec{v}= \vec{0}$, $\sigma \in \mathbb{C}\setminus \{0\}$ s.t. there exists no $m \in \mathbb{N}$ s.t. $\sigma^{m} = \pm 1$ \&  |[decision]| $T \vec{v}\propto \vec{v}$? \&  |[finish]| (IV)  \rule{\textwidth}{0.2pt} $m$ SLOCC classes for a fixed $m$: $T= \begin{pmatrix}e^{i \frac{r \pi}{m}}&0\\0& e^{-i \frac{r \pi}{m}} \end{pmatrix}$ for $r\in\{1, \ldots, \lfloor \frac{m-1}{2} \rfloor\}$, $\vec{v}\in \{(1,0)^T,(0,1)^T\}$. Moreover, $T \in \{\identity, i \sigma_z\}$ ($T=\identity$) for even $m$ (odd $m$), respectively, and $\vec{v}=(1,0)^T$.  \\
          |[finish]| (IIb) \rule{\textwidth}{0.2pt} 1-param. SLOCC family: $T= \begin{pmatrix}1&1\\0&1 \end{pmatrix}$, $\vec{v}= (x,0)^T$, $x \in \mathbb{C} \setminus \{0\}$.                 \&   |[decision]| $T \vec{v}\propto \vec{v}$? \&   |[finish]| (VI) - generic \rule{\textwidth}{0.2pt} 2-param. SLOCC family: $T= \begin{pmatrix}\sigma&0\\0&1/\sigma \end{pmatrix}$, $\sigma \in \mathbb{D}$ s.t. $\sigma^m\neq 1$ for any $m \in \mathbb{N}$, $\vec{v}=(1,0)^T$, $\vec{v}=(0,1)^T$, or $\vec{v} =  (v_1,1/v_1)^T$, $v_1 \in \mathbb{C} \setminus \{0\}$.   \&  \&  |[decision]| $T^2 \vec{v} \propto \vec{v}$, i.e., $T^2 \propto \identity$?  \&  |[finish]| (VII)  \rule{\textwidth}{0.2pt} 1-param. SLOCC family: $T= \begin{pmatrix}i&0\\0&-i \end{pmatrix}$, $\vec{v} =  (v_1,1/v_1)^T$, $v_1 \in \mathbb{C} \setminus \{0\}$. \\
                               \&         |[finish]| (V) \rule{\textwidth}{0.2pt} 1-param. SLOCC family: $T= \begin{pmatrix}1&1\\0&1 \end{pmatrix}$, $\vec{v} =  (0,x)^T$, $x\in \mathbb{C} \setminus \{0\}$.         \& \&  \&  |[finish]| (VIII) \rule{\textwidth}{0.2pt} 1-param. SLOCC family: $T= \begin{pmatrix}e^{i \frac{s \pi}{m}}&0\\0& e^{-i \frac{s \pi}{m}} \end{pmatrix}$ for $s\in\{0, \ldots, \lfloor m/2 \rfloor\}$, $\vec{v} =  (v_1,1/v_1)^T$, $v_1 \in \mathbb{C} \setminus \{0\}$. \&   \\
    };
  \draw
    (chart-1-2) edge (chart-2-2)    
    (chart-2-2) \no (chart-3-2)    
    (chart-2-2) \yes (chart-2-3)  
    (chart-3-2) \yes (chart-3-1)
    (chart-3-2) \no (chart-4-2) 
    (chart-2-3) edge node [above] {yes} (chart-2-4)
    (chart-2-4) edge  (chart-2-5)
    (chart-2-5) edge node [above] {yes} (chart-2-6)
    (chart-2-5) \no (chart-3-5)    
    (chart-2-3) \no (chart-3-3)
    (chart-3-3) \yes (chart-3-4)
    (chart-3-3) \no (chart-4-3)
    (chart-3-5) \no (chart-4-5)    
    (chart-3-5) \yes (chart-3-6)
    (chart-4-2) \yes (chart-4-1)
    (chart-4-5) \yes (chart-4-6)
    (chart-4-2) \no (chart-5-2)    
    (chart-4-5) \no (chart-5-5)    ;
\end{tikzpicture}
}
\caption{Summary of the SLOCC classification of all normal MPS generated by $\identity \otimes b \otimes \identity \ket{LLT}$ (shaded rectangles with solid contour) plus partial results on SLOCC classes of some non-normal MPS (shaded rectangles with dashed contour). We display the number of SLOCC classes corresponding to each type and give representatives for every SLOCC class (we count complex parameters). To simplify the presentation not all redundancies in the representatives of the continuous SLOCC families are avoided. Redundancy may be removed straightforwardly, though. The flowchart is following the same structure as the one in Figure \ref{fig:L1L1Tflowchart}, in particular, the displayed type labels agree. Note that additional non-normal MPS that have not been identified as SLOCC equivalent might in fact be equivalent. The displayed non-normal MPS vanish in case of odd $N$.
}
\label{fig:L1L1TflowchartSLOCC}
\end{figure*}

\subsection{Non-normal MPS}
\label{sec:L1L1Tnonnormal}

We have seen in Observation \ref{obs:L1L1Tnormality} that an MPS associated to $b$ is normal if and only if $b_{20} = 1$ and $\vec{v}_b \neq \vec{0}$. In Observation \ref{obs:L1L1Tb20vansihes} we have analyzed (non-normal) MPS associated to $b$ with $b_{20} = 0$. In this section, we discuss the remaining non-normal MPS, i.e., MPS generated by $b$ such that $b_{20} = 1$ and $\vec{v} = \vec{0}$, i.e. the three boxes (I), (IIa), and (III) in Fig.~\ref{fig:L1L1Tflowchart}. In particular, we present non-normal MPS belonging to box (I) and show that the symmetries determined in this section might indeed only be subgroups of the whole symmetry group for non normal MPS. Moreover, we show that the fact that any possible SLOCC transformation among normal MPS can be performed with a global transformation is no longer true for non normal MPS.

Note that for any odd particle number $N$ we have  $\ket{\Psi_{A_b}} = 0$, i.e., the MPS vanishes. Considering the definition of MPS [see Eq. (\ref{eq:mps})], this can be easily seen as follows.  If $\vec{v}= \vec{0}$, then both $A_b^0$ and $A_b^1$ are matrices of the form $\begin{pmatrix} 0 & 0 & \cdot \\ 0 & 0 & \cdot  \\ \cdot & \cdot & 0 \end{pmatrix}$, where $\cdot$ indicates an arbitrary (vanishing or non-vanishing) entry. This form is retained by any product of matrices of such a form, which is comprised of odd factors (see also Appendix \ref{app:L1L1Tnormality}). In particular, the trace vanishes and thus $\ket{\Psi_{A_b}} = 0$. Conversely, it may be straightforwardly verified that $\ket{\Psi_{A_b}} \neq 0$ for any even $N \geq 4$ (unless $\det b = 0$). In the following, we hence consider even $N \geq 4$.

Let us first consider MPS associated to diagonalizable $T$ as in box (I) in Fig.~\ref{fig:L1L1Tflowchart}.\ and consider the particularly interesting case $m=N$. In standard form, we then have $T = \operatorname{diag}(e^{i\frac{r \pi}{N}},e^{-i\frac{r \pi}{N}})$ for $r\in \{0,\ldots, \frac{N}{2}\}$. Considering only global $(b \rightarrow c)$-cycles these different states appear to be inequivalent. Note, however, that the premises of Lemmata \ref{lemma:L1L1TSLOCCsameT} and \ref{lemma:L1L1TSLOCCglobalsuffices} (which stated that this suffices to conclude that the associated MPS are inequivalent)  are not fulfilled, as $\vec{v} = \vec{0}$. Indeed, additional equivalences become apparent taking non-global $(b \rightarrow c)$-cycles into account. More precisely, considering Eq.~(\ref{eq:L1L1Tslocctrafo}) and $g_k = \operatorname{diag}(e^{-i\frac{r k \pi}{N}}, e^{i\frac{r k \pi}{N}})$ and $B_{11}^{(k)} = 1$ for $k \in \{0,\ldots, N-1\}$, we find that all $T = \operatorname{diag}(e^{i\frac{r \pi}{N}},e^{-i\frac{r \pi}{N}})$ for even $r$ are equivalent to $T = \identity$. For odd $r$ this construction does not work due to a sign mismatch in $B_{11}^{(k)}B_{11}^{(k+1)}$ in Eq.~(\ref{eq:L1L1Tslocctrafo}). Instead, those can be shown to be equivalent to $T = \operatorname{diag}(i,-i) = i\sigma_z$ using $g_k = \operatorname{diag}(e^{-i\frac{(r-N/2) k \pi}{N}}, e^{i\frac{(r-N/2) k \pi}{N}})$. Thus, there are (at most) two SLOCC classes within box (I) for $m=N$. The two representative MPS associated to $T=\identity$ and $T = i\sigma_z$ are in fact the Majumdar-Ghosh states \cite{Perez-Garcia2007} 
\begin{align*}
\ket{\Psi_{A_b}} \propto& \ket{\psi^{-}}_{01} \ldots  \ket{\psi^{-}}_{N-2 \, N-1}   \pm  \ket{\psi^{-}}_{12} \ldots  \ket{\psi^{-}}_{ N-1 \, 0},
\end{align*}
where $+$ (-) corresponds to $T=\identity$ ($T = i\sigma_z$), respectively. Note that $g^{\otimes N}$ is a symmetry of $\ket{\Psi_{A_b}}$ for any $g$. Thus, the study of cycles (see Figure \ref{fig:L1L1Tflowchart}) clearly has revealed only a subgroup of the symmetry group, as may be expected for non-normal tensors.

\section{Fiducial states for a bond dimension $D>3$ corresponding to diagonal matrix pencils}
\label{sec:2DDdiagonal}

We discuss here the generic case of fiducial states with bond dimension larger than 3 (see also Sec. \ref{sec:prelim}). The fiducial states correspond to diagonal matrix pencils \cite{HeGa18}, i.e., we have
\begin{align*}
\ket{A} = \ket{0}_A\ket{\Phi^+_D}_{BC} + \ket{1}_A  (D \otimes \identity) \ket{\Phi^+_D}_{BC},
\end{align*}
where $D = \operatorname{diag}(x_1, \ldots, x_D)$, with $x_1, \ldots, x_D$ being the eigenvalues of the corresponding matrix pencil. Its symmetries are of the form 
\begin{align*}
g \otimes P_\sigma^{-1} D_{g}^{-1} \tilde{D} \otimes  P_\sigma^{-1} \tilde{D}^{-1},
\end{align*}
where $g = \begin{pmatrix}\alpha & \beta \\ \gamma & \delta \end{pmatrix}$ is such that the set of eigenvalues of the pencil is mapped into itself (see Sec. \ref{sec:prelim}).

Depending on the eigenvalues of the diagonal matrix pencil, i.e., the entries of $D$, we have either: 
(i) the eigenvalues are such that the linear fractional transformation given in Eq.~(\ref{eq:EigenvaluesMPsym})  exists, or
(ii) no such transformation exists (which is the generic case). 

In case (ii), which is the generic case, the fiducial state has only the trivial qubit symmetry, which implies that the corresponding TIMPS has only the trivial symmetry. Moreover, any TIMPS which corresponds to a fiducial state which is in such a SLOCC class has only the trivial symmetry. In case (i) non--trivial symmetries exist. 
A simple example of such a state would be the fiducial state given in Eq.~(\theequation) with $x_k=\omega^k$, where $\omega = e^{i \frac{2\pi}{D}}$\footnote{Using the theory of Möbius transformations one might derive necessary and sufficient conditions on the fiducial states to possess non-trivial symmetries.}.
In general we have that in this case, $g$ is such that for all $i \in \{1, \ldots, D\}$, $x_i \mapsto x_{\sigma(i)}$ for some permutation $\sigma$. Moreover, $D_{g} = \operatorname{diag}(\gamma x_1 + \delta, \ldots, \gamma x_D + \delta)$, and $P_{\sigma}= \sum_{i} \ket{\sigma(i)}\bra{\sigma}$, and $\tilde{D}$ is an arbitrary invertible diagonal matrix (see equation above). Then, it is immediate that $\identity \otimes b \otimes \identity \ket{A}$ has the symmetries
\begin{align*}
g \otimes b P_\sigma^{-1} D_{g}^{-1} \tilde{D} b^{-1} \otimes  P_\sigma^{-1} \tilde{D}^{-1} = S \otimes b x b^{-1} \otimes y^T, 
\end{align*}
where the RHS is standard MPS notation for symmetries of the fiducial state.

We outline in the following how all the symmetries of the corresponding TIMPS can be determined. As in Sec. \ref{sec:M10M11M1inf} one would start out by solving the concatenation condition [see Eq. (\ref{eq:ConcatSymb})]. More precisely, in order to determine the physical symmetry of the MPS we have to identify the $N$-cycles within the symmetry group of the fiducial state. That is, we solve the concatenation rules: 
\begin{align*}
 y_k b x_{k+1} b^{-1} &\propto \identity \text{ for } k\in \{0, \ldots, N-2\},\nonumber\\
 y_0 b x_{N-1}b^{-1} &\propto \identity.
\end{align*}

For $\ket{A_b}$ as fiducial state, this condition is equivalent to  
\begin{align*}
 b P_{\sigma_{k+1}}^{-1} D_{{g}_{k+1}}^{-1} \tilde{D}_{k+1} b^{-1} \propto  P_{\sigma_k}^{-1} \tilde{D}_k.
\end{align*}

Thus, we can derive as a necessary condition that $b P_{\sigma_{k+1}}^{-1} D_{{g}_{k+1}}^{-1} \tilde{D}_{k+1} b^{-1}$ must be similar to $P_{\sigma_k}^{-1} \tilde{D}_k$. This implies that these matrices must have the same eigenvalues. Matrices of the form $P_{\sigma} D$, where $P_{\sigma}$ is a permutation matrix and $D$ is a diagonal matrix, are called \emph{generalized permutation matrices} \cite{JoHo}. It turns out that their eigenvalues are easy to calculate, as the following lemma shows.

\begin{lemma}[Eigenvalues of monomial matrices]
\label{lemma:genpermutation}
Let $P_\sigma$ be a permutation matrix and $D$ a diagonal matrix. Then the eigenvalues of $P_\sigma D$ can be determined as follows. Assume $\sigma$ decomposes into $l$ distinct cycles $\pi_1, \ldots, \pi_l$. Let $d_i$ denote the $\operatorname{length}(\pi_i)$-th root of the product of the entries of $D$ associated with the cycle $\pi_i$. Then the eigenvalues of $P_\sigma D$ are $\left(d_1 e^{i\frac{2 k \pi}{\operatorname{length}(\pi_1)}} \right)_{k=0}^{\operatorname{length}(\pi_1)-1} \cup \ldots \cup \left(d_l e^{i\frac{2 k \pi}{\operatorname{length}(\pi_l)}} \right)_{k=0}^{\operatorname{length}(\pi_l)-1}$.
\end{lemma}
\begin{proof}
  Let us fix a cycle $\pi_i$ and restrict $P_\sigma D$ to the subspace spanned by the basis elements that $\pi_i$ does not leave invariant. This matrix then has characteristic polynomial $\lambda^{\operatorname{length}(\pi_i)} - d_i^{\operatorname{length}(\pi_i)} =0$, and thus its eigenvalues are as stated in the lemma. 
\end{proof}

Using now these necessary conditions for the existence of a cycle, similar tools as the ones presented in Sec \ref{sec:M10M11M1inf} can be utilized to determine all symmetries of the corresponding TIMPS.

\section{Entanglement and LOCC transformations}
\label{sec:reventt}

Before concluding, let us briefly discuss the implication of the results derived here in the context of entanglement theory. 
As mentioned in the introduction, if a state, $\ket{\Psi}$ can be transformed deterministically via LOCC into some other state $\ket{\Phi}$, then $E(\ket{\Psi})\geq E(\ket{\Phi}) $ for any entanglement measure $E$. Hence, LOCC transformations induce a partial order on the set of entangled states. As shown in \cite{GoWa11, deSp13,HeEn21} local symmetries play an important role in characterizing all possible LOCC transformations among pure states. As we have characterized all the local symmetries of the TIMPS, it is straightforward to determine possible LOCC transformations (at least in case the number of symmetries is finite). To give a simple example a TIMPS $\ket{\Psi}$ can be transformed deterministically into a state $h_1 \otimes \one  \otimes \one \otimes \ldots \ket{\Psi}$, where $h_1$ is determined by the  symmetries of $\ket{\Psi}$. More precisely, if the symmetries are unitary symmetries on (at least) all but one system, system 1, then the above mentioned transformation is possible if and only if there exists a finite set of probabilities $\{p_k\}$ and symmetries $g^{k}$ such that $\sum_k p_k \left(g^k_1\right)^\dagger h_1^\dagger h_1 g^k_1 \propto \identity$.
Recall that the symmetry group of an MPS may depend on the particle number $N$. In drastic cases, an MPS may exhibit the trivial symmetry group for certain $N$, while the symmetry group is non-trivial for other $N$. A simple example would be the MPS $\Psi_{1\, D^3}$ from Section \ref{sec:examples}. Thus, it can be easily seen that whether $\ket{\Psi}$ can be transformed deterministically into a state $h_1 \otimes \one  \otimes \one \otimes \ldots \ket{\Psi}$ via LOCC may depend on the particle number $N$.
Another assertion concerning reachability of states under LOCC is possible due to knowing the full symmetry group of a TIMPS. Namely, TIMPS which possess non-trivial global symmetries, but no local symmetries, such as, e.g., $\Psi_{G \, \text{fin}}$ from Section \ref{sec:examples}, are not reachable from any other state via an LOCC protocol involving a finite number of rounds of classical communication \cite{sym}.

In case a deterministic transformation is not posssible, one might study the maximal success probability of transforming one TIMPS into another. We denote by $P(\psi \rightarrow \phi)$ the maximal success probability for transforming $\psi$ to $\phi$. It has been shown in \cite{Vidal} that $P(\psi \rightarrow \phi)=\min_\mu \mu(\psi)/\mu(\phi)$, where $\mu$ denotes an arbitrary entanglement monotone. For a generic set of states, this minimum can be easily determined. This set is defined as the union of all SLOCC classes which possess a representative whose single party reduced state is completely mixed (critical state) and whose stabilizer is trivial. It has been shown that this set is a full measure set in case of a homogeneous system, i.e. where all local dimensions coincide \cite{SaWa18}. For those generic multipartite states, we have \cite{GoKr17}
\begin{equation*} 
  P(\psi \rightarrow \phi)= \frac{||\phi||^2}{||\psi||^2}\frac{1}{\lambda_{max}(G^{-1} H)},
\end{equation*}
where $G=g^\dagger g$, $H=h^\dagger h$, and $\psi=g \psi_s$, $\phi=h \psi_s$, with $\psi_s$ the critical representative of the SLOCC class and $g$ and $h$ are local operators. Here, $\lambda_{max}$ denotes the maximal eigenvalue. Note that the maximal success probability can be easily determined as $G$ and $H$ are local operators. Note that for these states it is also possible to determine so--called SLOCC--paths along which on state can be transformed optimally into the other \cite{SaSc18}. Furthermore, for these states a complete set of entanglement monotones, which can be easily computed, is known \cite{SaSc18}. Clearly all these results apply to TIMPS which belong to the above mentioned full--measured set.

\section{Conclusion}
We studied the symmetries of TIMPS with bond dimension $D=3$ and showed that they are in strong contrast to TIMPS with bond dimension $D=2$. Depending on the SLOCC class of the underlying fiducial state, very different symmetry and entanglement properties (regarding SLOCC classes) occur. We illustrate the rich variety of states by presenting TIMPS with particular symmetry groups. 

In a future project it will be interesting to investigate how the stabilizer groups and SLOCC classes presented here relate to the results on the classification of phases of matter presented in \cite{Schuch2011,Chen2011}. Furthermore, the relaxation of locality in (S)LOCC as presented in \cite{Piroli2021} might reveal a more coarse--grained structure of the SLOCC classes presented here. 

\section{Acknowledgments} M. H. and B. K.  acknowledge financial support from the Austrian Science Fund (FWF): the DK-ALM (W1259-N27), the stand alone project: P32273-N27, and the SFB BeyondC (F7017). A.M. acknowledges funding from the European Research Council (ERC) under the European Union's Horizon 2020 research and innovation programme through the ERC-CoG SEQUAM (grant agreement No. 863476) and ERC-CoG GAPS (grant agreement No.  648913).
J. I. Cirac acknowledges funding from the ERC-AdG QUENOCOBA
(No. 742102) and 
the DFG under
Germany's Excellence Strategy (EXC2111-390814868),
and the SFB BeyondC F7104 
Project number 414325145.



\appendix

\section{Characterization of cycles considering fiducial states $\identity \otimes b \otimes \identity \ket{ M(\omega) }$}
\label{app:M10M11M1infcycles}

In this appendix we present Table \ref{tab:M10M11M1infcycleslong} providing details on normal TIMPS generated by fiducial states $\identity \otimes b \otimes \identity \ket{ M(\omega) }$ (see Section \ref{sec:M10M11M1inf}). The table lists all possible cycles and provides parametrizations for all normal fiducial states $\identity \otimes b \otimes \identity \ket{ M(\omega) }$ s.t. $G_b$ exhibits the respective cycles. See Section  \ref{sec:M10M11M1inf} for the methods required to derive the table, as well as an exemplary calculation for the cycle $C_0$.

\LTcapwidth=\textwidth
\begin{longtable*}{c|c|c|c|>{\raggedright\arraybackslash}p{0.52\linewidth}|c }
Label & Subgroup(s)  & Cycle & $N$ & $b$ & \#param.\\ \hline
  $C_0$&- &  $S$& 1 &  $D \operatorname{diag}(1,\omega,\omega^2)^l \begin{pmatrix} b_{00} & b_{01}& b_{02} \\ b_{02} & b_{00} & b_{01} \\ b_{01} & b_{02} & b_{00} \end{pmatrix}D^{-1}$, where $b_{00}, b_{01}, b_{02}\in \mathbb{C}$, $l \in \{0,1,2\}$ & 2\\
 $T_0^\tau$ & - &  $\tau$ & 1& $D \operatorname{diag}(1,\omega,\omega^2) \begin{pmatrix} -b_{00} & -b_{01}& -b_{01} \\ b_{01} & b_{11} & b_{12} \\ b_{01} & b_{12} & b_{11} \end{pmatrix}D^{-1}$, or $D \operatorname{diag}(1,\omega,\omega^2) \begin{pmatrix} b_{00} & b_{01}& i b_{01} \\ -b_{01} & -b_{11} & -b_{12} \\ i b_{01} & b_{12} & -b_{11} \end{pmatrix}D^{-1}$, where $b_{00}, b_{01}, b_{11}, b_{12} \in \mathbb{C}$  &3\\
 $T_0^\epsilon$ & - &  $\epsilon$ & 1&  $D \operatorname{diag}(1,\omega,\omega^2) \begin{pmatrix} b_{00} & b_{01}& b_{02} \\ b_{01} & b_{00} & b_{02} \\ -b_{02} & -b_{02} & -b_{22} \end{pmatrix}D^{-1}$, or $D \operatorname{diag}(1,\omega,\omega^2)\begin{pmatrix} b_{00} & b_{01}& b_{02} \\ -b_{01} & b_{00} &  i b_{02} \\ b_{02} & -i b_{02} & b_{22} \end{pmatrix}D^{-1}$, where $b_{00}, b_{01}, b_{02}, b_{22} \in \mathbb{C}$ &3 \\
 $T_0^\kappa$ & - &  $\kappa$ & 1& $D \operatorname{diag}(1,\omega,\omega^2) \begin{pmatrix} b_{00} & b_{01}& b_{02} \\ -b_{01} & -b_{11} & -b_{01} \\ b_{02} & b_{01} & b_{00} \end{pmatrix}D^{-1}$, where $b_{00}, b_{01}, b_{02}, b_{11} \in \mathbb{C}$&3 \\ \hline
  $C_1$ &-&  $S \otimes S^2$  & 2 & $D \operatorname{diag}(1,\omega,\omega^2) \begin{pmatrix} b_{00} & b_{01}& b_{02} \\ b_{01} & b_{02} & b_{00} \\ b_{02} & b_{00} & b_{01} \end{pmatrix}D^{-1}$, or $D \operatorname{diag}(1,\omega,\omega^2) \begin{pmatrix} b_{00} & b_{01}& b_{02} \\ b_{01} & -b_{02} & -b_{00} \\ b_{02} & -b_{00} & b_{01} \end{pmatrix}D^{-1}$, or $D \operatorname{diag}(1,\omega,\omega^2) \begin{pmatrix} b_{00} & b_{01}& b_{02} \\ -b_{01} & -b_{02} & -i b_{00} \\ b_{02} & i b_{00} & i b_{01} \end{pmatrix}D^{-1}$ where $b_{00}, b_{01}, b_{02} \in \mathbb{C}$&2\\
 $T_1^\tau$ & $T_0^\tau$ &  $\identity \otimes \tau$& 2  & $D \operatorname{diag}(1,\omega,\omega^2) \begin{pmatrix} 0 & b_{01}& b_{01} \\ -b_{01} & -b_{11} & b_{11} \\ -b_{01} & b_{11} & -b_{11} \end{pmatrix}D^{-1}$, where $b_{01}, b_{11} \in \mathbb{C}$ &1\\
 $T_1^\epsilon$ & $T_0^\epsilon$ &  $\identity \otimes \epsilon$  & 2& $D \operatorname{diag}(1,\omega,\omega^2) \begin{pmatrix} -b_{00} & b_{00}& -b_{02} \\ b_{00} & -b_{00} & -b_{02} \\ b_{02} & b_{02} & 0 \end{pmatrix}D^{-1}$, where $b_{00}, b_{02} \in \mathbb{C}$ &1\\
 $T_1^\kappa$ & $T_0^\kappa$ &  $\identity \otimes \kappa$ & 2 & $D \operatorname{diag}(1,\omega,\omega^2) \begin{pmatrix} b_{00} & b_{01}& -b_{00} \\ -b_{01} & 0 & -b_{01} \\ -b_{00} & b_{01} & b_{00} \end{pmatrix}D^{-1}$, where $b_{00}, b_{01} \in \mathbb{C}$ &1\\
 $T_2^\tau$ &  $T_0^\tau$, $C_1$ & $\epsilon \otimes \kappa$ & 2& $D \operatorname{diag}(1,\omega,\omega^2) \begin{pmatrix} b_{00} & b_{01}& b_{01} \\  -b_{01} & -i b_{01} &  b_{00} \\ -b_{01} &  b_{00} & -b_{01} \end{pmatrix}D^{-1}$, or $D \operatorname{diag}(1,\omega,\omega^2) \begin{pmatrix} b_{00} & b_{01}& -i b_{01} \\ -b_{01} & i b_{01} & -i b_{00} \\ -i b_{01} & i b_{00} & i b_{01} \end{pmatrix}D^{-1}$, where $b_{00}, b_{01} \in \mathbb{C}$ &1 \\
 $T_2^\epsilon$ &  $T_0^\epsilon$, $C_1$ & $\tau \otimes \kappa$ & 2&  $D \operatorname{diag}(1,\omega,\omega^2) \begin{pmatrix} -b_{00} & -b_{01}& - i b_{00} \\ -b_{01} & -b_{00} & -i b_{00} \\ i b_{00} & i b_{00} & -b_{01} \end{pmatrix}D^{-1}$, or $D \operatorname{diag}(1,\omega,\omega^2) \begin{pmatrix} b_{00} & b_{01}& -b_{00} \\ -b_{01} & b_{00} & -i b_{00} \\ -b_{00} & i b_{00} & i b_{01} \end{pmatrix}D^{-1}$, where $b_{00}, b_{01} \in \mathbb{C}$ &1\\
 $T_2^\kappa$ &  $T_0^\kappa$, $C_1$ &  $\tau \otimes \epsilon$ & 2& $D \operatorname{diag}(1,\omega,\omega^2) \begin{pmatrix} b_{00} & i b_{00}& b_{02} \\ -i b_{00} & b_{02} & -i b_{00} \\ b_{02} & i b_{00} & b_{00} \end{pmatrix}D^{-1}$, or $D \operatorname{diag}(1,\omega,\omega^2) \begin{pmatrix} b_{00} & -i b_{00}& b_{02} \\ i b_{00} & -b_{02} & -i b_{00} \\ b_{02} & i b_{00} & b_{00} \end{pmatrix}D^{-1}$, where $b_{00},  b_{02} \in \mathbb{C}$ &1\\ \hline
  $C_2$& $C_0$ & $\identity \otimes S \otimes S^2$& 3&  $D  \operatorname{diag}(1,\omega,\omega^2)^l \begin{pmatrix} 1 & 1 & \omega^{m} \\  \omega^{m} & 1 & 1 \\ 1 & \omega^{m} & 1 \end{pmatrix}D^{-1}$, where $l \in \{0,1,2\}$, $m \in \{1,2\}$ &-\\
  $T_3^\tau$& - & $\identity \otimes \tau \otimes \tau$& 3& $D \operatorname{diag}(1,\omega,\omega^2) \begin{pmatrix} 0 & b_{01}& b_{01} \\ -b_{01} & -b_{11} & b_{11} \\ b_{01} & -b_{11} & b_{11} \end{pmatrix}D^{-1}$, where $b_{01}, b_{11} \in \mathbb{C}$ & 1\\
  $T_3^\epsilon$& - &  $\identity \otimes \epsilon \otimes \epsilon$& 3&  $D \operatorname{diag}(1,\omega,\omega^2) \begin{pmatrix} b_{00} & -b_{00}& b_{02} \\ b_{00} & -b_{00} & -b_{02} \\ b_{02} & b_{02} & 0 \end{pmatrix}D^{-1}$, where $b_{00}, b_{02} \in \mathbb{C}$ &1\\
  $T_3^\kappa$& - & $\identity \otimes \kappa \otimes \kappa$& 3& $D \operatorname{diag}(1,\omega,\omega^2) \begin{pmatrix} b_{00} & b_{01}& b_{00} \\ -b_{01} & 0 & b_{01} \\ -b_{00} & b_{01} & -b_{00} \end{pmatrix}D^{-1}$, where $b_{00}, b_{01} \in \mathbb{C}$ &1 \\
  $T_4^{\circlearrowleft}$& $C_0$  &$\tau \otimes \kappa \otimes \epsilon$& 3 & $D \operatorname{diag}(1,\omega,\omega^2) \begin{pmatrix} b_{00} & b_{01}& b_{00} \omega^m \\ b_{00} \omega^m & b_{00} & b_{01} \\ b_{01} & b_{00}\omega^m & b_{00} \end{pmatrix}D^{-1}$, where $b_{00}, b_{01} \in \mathbb{C}$, $m \in \{0,1,2\}$ &1 \\
  $T_4^{\circlearrowright}$& $C_0$ &  $\tau \otimes \epsilon \otimes \kappa$& 3 &$D \operatorname{diag}(1,\omega,\omega^2) \begin{pmatrix} b_{00} & b_{00} \omega^m& b_{02}  \\ b_{02}  & b_{00} & b_{00}\omega^m \\ b_{00} \omega^m & b_{02} & b_{00} \end{pmatrix}D^{-1}$, where $b_{00}, b_{01} \in \mathbb{C}$, $m \in \{0,1,2\}$ &1\\
  $T_5^{\circlearrowleft}$& $C_2$, $C_0$ & $\tau \otimes \tau \otimes \kappa $& 3& $D  \operatorname{diag}(1,\omega,\omega^2)^{m+1} \begin{pmatrix} 1 & 1 & \omega^{m} \\  \omega^{m} & 1 & 1 \\ 1 & \omega^{m} & 1 \end{pmatrix}D^{-1}$, where $m \in \{1,2\}$ &-\\
  $T_5^{\circlearrowright}$& $C_2$, $C_0$ & $\tau \otimes \tau \otimes \epsilon $& 3 & $D  \operatorname{diag}(1,\omega,\omega^2)^{2m+1} \begin{pmatrix} 1 & 1 & \omega^{m} \\  \omega^{m} & 1 & 1 \\ 1 & \omega^{m} & 1 \end{pmatrix}D^{-1}$, where $m \in \{1,2\}$ &-\\ \hline
  $C_3$& - &  $\identity \otimes S \otimes \identity \otimes S^2$ & 4 & $D\operatorname{diag}(1,\omega,\omega^2)^l  \begin{pmatrix} 1 & 1 & 1 \\  1 & \omega^{m} & \omega^{2m} \\ 1 & \omega^{2m} & \omega^{m} \end{pmatrix}D^{-1}$, where $l \in \{0,1,2\}$, $m\in\{1,2\}$&- \\
  $T_6^\tau$& $C_3$, $T_0^\tau$ & $\kappa \otimes \tau \otimes \epsilon \otimes \tau $ & 4 & $D \operatorname{diag}(1,\omega,\omega^2) \begin{pmatrix} 1 & 1 & 1 \\  1 & \omega^{m} & \omega^{2m} \\ 1 & \omega^{2m} & \omega^{m} \end{pmatrix}D^{-1}$, where $m\in\{1,2\}$&-\\
  $T_6^\epsilon$& $C_3$, $T_0^\epsilon$ & $\tau \otimes \epsilon \otimes \kappa \otimes \epsilon $ & 4 & $D\operatorname{diag}(1,\omega,\omega^2)^{2m+1}  \begin{pmatrix} 1 & 1 & 1 \\  1 & \omega^{m} & \omega^{2m} \\ 1 & \omega^{2m} & \omega^{m} \end{pmatrix}D^{-1}$, where $m\in\{1,2\}$&-\\
  $T_6^\kappa$& $C_3$, $T_0^\kappa$ &  $\epsilon \otimes \kappa \otimes \tau \otimes \kappa $ & 4 &$D\operatorname{diag}(1,\omega,\omega^2)^{m+1}  \begin{pmatrix} 1 & 1 & 1 \\  1 & \omega^{m} & \omega^{2m} \\ 1 & \omega^{2m} & \omega^{m} \end{pmatrix}D^{-1}$, where $m \in \{1,2\}$ &- \\ \hline
   $C_4$& $C_1$&  $\identity \otimes S \otimes S \otimes \identity \otimes S^2 \otimes S^2$  & 6 & $D \operatorname{diag}(1,\omega,\omega^2)^l \begin{pmatrix} 1 & 1& \omega^{m} \\  1 & \omega^m & 1 \\ \omega^{m} & 1 & 1 \end{pmatrix}D^{-1}$, where $l \in \{0,1,2\}$, $m \in \{1,2\}$ &-\\
  $T_7^\tau$& $C_4$, $C_1$, $T_2^\tau$ , $T_0^\tau$ &  $\kappa \otimes \tau \otimes \epsilon \otimes \epsilon \otimes \tau \otimes \kappa$  & 6 &  $D \operatorname{diag}(1,\omega,\omega^2)^{m+1} \begin{pmatrix} 1 & 1& \omega^{m} \\  1 & \omega^m & 1 \\ \omega^{m} & 1 & 1 \end{pmatrix}D^{-1}$, where  $m \in \{1,2\}$&-\\
  $T_7^\epsilon$& $C_4$, $C_1$, $T_2^\epsilon$, $T_0^\epsilon$ & $\tau \otimes \epsilon \otimes \kappa \otimes \kappa \otimes \epsilon \otimes \tau$  & 6 & $D \operatorname{diag}(1,\omega,\omega^2)^{2m+1} \begin{pmatrix} 1 & 1& \omega^{m} \\  1 & \omega^m & 1 \\ \omega^{m} & 1 & 1 \end{pmatrix}D^{-1}$, where  $m \in \{1,2\}$&-\\
  $T_7^\kappa$& $C_4$, $C_1$, $T_2^\kappa$, $T_0^\kappa$ & $\epsilon \otimes \kappa \otimes \tau \otimes \tau \otimes \kappa \otimes \epsilon$  & 6  & $D \operatorname{diag}(1,\omega,\omega^2) \begin{pmatrix} 1 & 1& \omega^{m} \\  1 & \omega^m & 1 \\ \omega^{m} & 1 & 1 \end{pmatrix}D^{-1}$, where  $m \in \{1,2\}$&-\\\hline
  
  \caption{Characterization of all $b$ leading to non-trivial cycles for normal MPS generated by the fiducial states $\identity \otimes b \otimes \identity \ket{ M(\omega) }$. The first four columns coincide with Table \ref{tab:M10M11M1infcycles}. 
  The fifth column shows a parametrization (possible alternative parametrizations) of the set of all normal $b$'s s.t. $G_b$ exhibits the respective cycles. Parameter choices such that $b$ becomes a generalized permutation matrix must be excluded in order to have normality.
  The last column indicates the effective (complex) dimension of the respective sets ("-" indicates discrete sets). The matrix $D$ appearing in the fifth column is always an arbitrary diagonal matrix. As explained in Section \ref{sec:M10M11M1inf}, $D$ is actually irrelevant for the generated MPS. See Figure \ref{fig:sets} in the main text for set-theoretic relations between all the displayed parametrized sets.}
\label{tab:M10M11M1infcycleslong}  
\end{longtable*}

\section{Proof of Observation \ref{obs:M10M11M1infnormality} concerning the normality of MPS generated by fiducial states $\identity \otimes b \otimes \identity \ket{ M(\omega) }$}
\label{app:M10M11M1infnormality}

In this appendix we prove Observation \ref{obs:M10M11M1infnormality} from the main text, which concerns the normality of fiducial states $\identity \otimes b \otimes \identity \ket{ M(\omega) }$. Before we do so, let us discuss a few general remarks on the matter.

Recall that a fiducial state $\identity \otimes b \otimes \identity \ket{A}$ is normal if and only if for some fixed $L$, it is possible to build products of $A_b^{0}$ and $A_b^{1}$ comprised of $L$ factors, which form a basis for all $D \times D$ matrices (see Section \ref{sec:prelim}). It suffices to consider $L$ up to a certain upper bound depending on the bond dimension $D$ \cite{Sanz2010}. For any concrete choice of $b$ it is simple to decide normality. To this end, one may proceed as follows. First, one calculates all possible products of $A_b^{0}$ and $A_b^{1}$ of length $L$ leading to $2^L$ $D\times D$ matrices. Then, one rearranges the matrix entries in order to form $D^2$-dimension vectors, which are then used as columns of a $D^2\times 2^L$ matrix $M$. Obviously, the tensor is normal if $\operatorname{rk} M = D^2$ and it is non-normal if  $\operatorname{rk} M < D^2$, hence, normality may be decided by computing the rank\footnote{In order to circumnavigate numeric imprecisions, one may e.g. compute the singular values of $M$ and make sure that $D^2$ of them are sufficiently different from 0.} of $M$.  
Deciding normality of a continuous family of $b$'s is more involved, though. Which products of $A_b^{0}$ and $A_b^{1}$ one needs to consider in order to obtain a basis will typically depend on the parameter choices in $b$. It is still comparably simple to show that a family of $b$'s is normal for generic parameter choices. To this end, one may construct the matrix $M$ as above and then consider the determinant of a submatrix of $M$ obtained by selecting $D^2$ columns of $M$, which will be some polynomial in the entries of $b$. In case the obtained polynomial is not identically zero, this shows that the MPS is normal for generic parameter choices. However, as mentioned above, typically there will be certain particular parameter choices for which the polynomial vanishes. In order to prove that the whole family leads to normal MPS, one thus often needs to consider additional submatrices of $M$ (and their determinants) and show that the obtained polynomials don't have a common root (additionally assuming $\det b \neq 0$). 

We now make a small observation concerning the fact that an MPS is normal if and only if any SLOCC equivalent MPS is, which is an immediate consequence of the concepts introduced in \cite{SaMo19}.

\begin{observation}
\label{obs:SLOCCnormality}
A (not necessarily TI) fiducial state $g_k \otimes \identity \otimes \identity \ket{A}$ is normal with injectivity length $L$ for any $g_1 \ldots g_N$ if and only if $\ket{A}$ is.
\end{observation} 

Let us also recall Observation \ref{obs:permutationnotnormal}, which we have proven in the main text. According to Observation \ref{obs:permutationnotnormal}, fiducial states $\identity \otimes b \otimes \identity \ket{ M(\omega) }$ are not normal if $b$ is a generalized permutation matrix, or such that in any row or column $i$, $b_{ii}$ is the only non-vanishing entry. Let us now restate and prove Observation \ref{obs:M10M11M1infnormality}. 

\noindent {{\bf Observation \ref{obs:M10M11M1infnormality}}{\bf.}}\textit{
All $b$ s.t. there exist non-trivial cycles in  $G_b$ (see Table \ref{tab:M10M11M1infcycleslong}) are normal with injectivty length $L=4$, $5$, or $6$, unless $b$ fulfills the prerequisites of Observation \ref{obs:permutationnotnormal}. 
}
\begin{proof}
We use the parametrizations for $b$ as given in Table \ref{tab:M10M11M1infcycleslong}.
We first argue that in order to prove the observation, it suffices to prove the statement for only some of the families of $b$'s in Table \ref{tab:M10M11M1infcycleslong}. We will then argue that it will be convenient to change the considered representative from $\identity \otimes b \otimes \identity \ket{ M(\omega) }$ to $\identity \otimes b \otimes \identity \ket{ M^1(0) \oplus M^1(1) \oplus M^1(\infty)}$. For each of the relevant families, we then proceed as discussed above.

In order to prove the observation it suffices to consider $b(C_0)$, $b(C_1)$, $b(C_3)$, $b(T_0^{\tau, \epsilon, \kappa})$, and $b(T_3^{\tau, \epsilon, \kappa})$, as all other relevant families are subfamilies of the mentioned ones (see Figure \ref{fig:groupstructure}). Moreover, it suffices to prove the statement for $b(T_0^{\tau})$ and $b(T_3^{\tau})$ instead of all the families $b(T_0^{\tau, \epsilon, \kappa})$ and $b(T_3^{\tau, \epsilon, \kappa})$. The reason for this is that the families are SLOCC equivalent to the families $b(T_0^{\tau})$ and $b(T_3^{\tau})$, respectively, via a global operation. Due to Observation \ref{obs:SLOCCnormality}, normality is retained under an SLOCC operation.
Thus, we need to show the statement of the observation for the families $b(C_0)$,$b(C_1)$, $b(C_3)$, $b(T_0^{\tau})$, and $b(T_3^{\tau})$ only.
In order to do so, we proceed as outlined in the discussion above Observation \ref{obs:SLOCCnormality}. Recall that the matrix $D$ within the parametrizations for $b$ within Table \ref{tab:M10M11M1infcycleslong} does not alter the generated MPS, we will hence disregard $D$ in the following. Note that the representative for the fiducial state $\ket{ M(\omega) }$ is SLOCC equivalent to $\ket{ M^1(0) \oplus M^1(1) \oplus M^1(\infty)}$. While using the former representative was more convenient in the main text, as the physical operators constitute a more natural representation of the symmetric group, the latter representative will be more convenient here, as the tensor $A$ is more sparsely populated in that case. Note that $\ket{ M(\omega) }\propto \begin{pmatrix}1 & \omega^2\\ \omega & \omega^2 \end{pmatrix} \otimes \operatorname{diag}(1,-\omega^2,\omega) \otimes \identity  \ket{ M^1(0) \oplus M^1(1) \oplus M^1(\infty)}$. Thus, $\identity \otimes b \otimes \identity \ket{ M(\omega) }$ is normal if and only if $\identity \otimes b' \otimes \identity \ket{ M^1(0) \oplus M^1(1) \oplus M^1(\infty)}$ is for $b' = \operatorname{diag}(1,-\omega^2,\omega) b$. Moreover, note that $b$ is of the form given in Observation \ref{obs:permutationnotnormal} if and only if $b'$ is. Thus, it suffices to prove the statement of the observation considering the tensor
\begin{align}
A_{b'}^0 = \begin{pmatrix} 0& b'_{01} & b'_{02} \\  0 &b'_{11} &  b'_{12}\\ 0 & b'_{21}& b'_{22}\end{pmatrix}, \ 
A_{b'}^1 = \begin{pmatrix} b'_{00} & b'_{01} & 0\\  b'_{10} &  b'_{11} & 0\\b'_{20}& b'_{21} & 0\end{pmatrix},
\end{align}
where $b'$ is given by the parametrizations of the families $b(C_0)$, $b(C_1)$, $b(C_3)$, $b(T_0^{\tau})$, and $b(T_3^{\tau})$ as in Table \ref{tab:M10M11M1infcycleslong}, multiplied with $\operatorname{diag}(1,-\omega^2,\omega)$ from the left. Let us now show the statement for the individual families. In the following we simply write $b$ instead of $b'$.

Let us start considering the family $b(C_3)$. As this is a discrete family, it is straightforward to show that the corresponding fiducial states are normal with $L=4$. To this end, it suffices to construct the matrix $M$ as described above Observation \ref{obs:SLOCCnormality} and verify that it has rank $9$.

Let us now come to the continuous families. For each of these families, we will provide several alternative choices of nine products of $A_b^0$ and $A_b^1$ comprised of $L$ factors. Then, for each parameter choice within the considered family of $b$'s, at least one of the alternatives will provide a basis for all $3\times 3$ matrices (for generic parameter choices, all of the alternatives do). Before we provide the concrete choices, a few remarks are in order. First, note that the provided choices are by far not unique and, moreover, there might exist choices such that less than the provided number of alternatives suffice to show normality. Second, note that if nine operators $\{O_i\}_i$ form a basis of $3\times3$ matrices then so do the operators $\{X O_i\}_i$ for any invertible $X$. Here it is advantageous to consider products of $A_b^0$ and $A_b^1$ multiplied by $b^{-1}$ from the left instead of the mere products, such as, e.g., $b^{-1}A_b^0 A_b^1 A_b^0A_b^0$ or $b^{-1}A_b^0 A_b^0 A_b^1A_b^1$ instead of $A_b^0 A_b^1 A_b^0A_b^0$ or $A_b^0 A_b^0 A_b^1A_b^1$.
The reason for this is that due to the form of $A^0 = \ket{1}\bra{1}+\ket{2}\bra{2}$ and $A^1 = \ket{0}\bra{0}+\ket{1}\bra{1}$, we obtain matrices with only four non-vanishing entries and it becomes much simpler to identify independent ones. 

Let us start with considering the family $b(T_3^\tau)$. Here, it actually suffices to consider a single choice of nine products of the operators $A_b^0$, $A_b^1$ of length $L=4$, which turn out to be linearly independent for all $b$ within the considered family. Considering
$\{A_b^0A_b^0A_b^0A_b^0,\allowbreak A_b^0A_b^0A_b^0 A_b^1,\allowbreak A_b^0A_b^0 A_b^1 A_b^1,\allowbreak A_b^0 A_b^1A_b^0 A_b^1,\allowbreak A_b^0 A_b^1 A_b^1A_b^0,\allowbreak A_b^0 A_b^1 A_b^1 A_b^1,\allowbreak  A_b^1A_b^0 A_b^1A_b^0,\allowbreak  A_b^1A_b^0 A_b^1 A_b^1,\allowbreak  A_b^1 A_b^1 A_b^1 A_b^1\}$, it may be easily verified that the determinant of the corresponding $9 \times 9$ matrix reads $-2048 b_{01}^{16}b_{11}^{11}$. In order to have an invertible $b$ we have $b_{01}\neq 0$ and $b_{11}^{11}\neq 0$. Hence, the given operator products are linearly independent (and thus form a basis for all $3\times3$ matrices) for all $b$ within the considered family. 

Let us now consider the family $b(C_0)$. We consider the following alternative sets of products with $L=4$,
$\{A_b^0A_b^0A_b^0A_b^0,\allowbreak A_b^0A_b^0A_b^0 A_b^1,\allowbreak A_b^0 A_b^1A_b^0A_b^0,\allowbreak  A_b^1 A_b^1A_b^0 A_b^1,\allowbreak  A_b^1 A_b^1 A_b^1 A_b^1,\allowbreak  A_b^1A_b^0A_b^0A_b^0,\allowbreak  A_b^1A_b^0 A_b^1A_b^0,\allowbreak  A_b^1 A_b^1A_b^0A_b^0,\allowbreak  A_b^1 A_b^1 A_b^1A_b^0\}$,
$\{A_b^0A_b^0 A_b^1A_b^0,\allowbreak A_b^0A_b^0 A_b^1 A_b^1,\allowbreak A_b^0 A_b^1 A_b^1A_b^0,\allowbreak  A_b^1A_b^0A_b^0 A_b^1,\allowbreak  A_b^1A_b^0 A_b^1 A_b^1,\allowbreak  A_b^1A_b^0A_b^0A_b^0,\allowbreak  A_b^1A_b^0 A_b^1A_b^0,\allowbreak  A_b^1 A_b^1A_b^0A_b^0,\allowbreak  A_b^1 A_b^1 A_b^1A_b^0\}$,
$\{A_b^0A_b^0A_b^0A_b^0,\allowbreak A_b^0A_b^0A_b^0 A_b^1,\allowbreak A_b^0A_b^0 A_b^1A_b^0,\allowbreak  A_b^1A_b^0A_b^0 A_b^1,\allowbreak  A_b^1 A_b^1 A_b^1 A_b^1,\allowbreak  A_b^1A_b^0A_b^0A_b^0,\allowbreak  A_b^1A_b^0 A_b^1A_b^0,\allowbreak  A_b^1 A_b^1A_b^0A_b^0,\allowbreak  A_b^1 A_b^1 A_b^1A_b^0\}$,
$\{A_b^0A_b^0A_b^0A_b^0,\allowbreak A_b^0A_b^0A_b^0 A_b^1,\allowbreak A_b^0A_b^0 A_b^1A_b^0,\allowbreak A_b^0 A_b^1A_b^0 A_b^1,\allowbreak A_b^0 A_b^1 A_b^1A_b^0,\allowbreak  A_b^1A_b^0 A_b^1A_b^0,\allowbreak  A_b^1A_b^0A_b^0A_b^0,\allowbreak  A_b^1A_b^0A_b^0 A_b^1,\allowbreak  A_b^1 A_b^1A_b^0 A_b^1\}$, as well as two with $L=6$,
$\{A_b^0A_b^0A_b^0A_b^0A_b^0A_b^0,\allowbreak A_b^0A_b^0A_b^0A_b^0A_b^0 A_b^1,\allowbreak A_b^0A_b^0A_b^0A_b^0 A_b^1A_b^0,\allowbreak A_b^0A_b^0A_b^0 A_b^1A_b^0A_b^0,\allowbreak A_b^0A_b^0 A_b^1A_b^0A_b^0 A_b^1,\allowbreak A_b^0 A_b^1A_b^0A_b^0 A_b^1A_b^0,\allowbreak  A_b^1A_b^0A_b^0A_b^0A_b^0A_b^0,\allowbreak  A_b^1A_b^0A_b^0A_b^0A_b^0 A_b^1,\allowbreak  A_b^1 A_b^1 A_b^1 A_b^1 A_b^1 A_b^1\}$,
$\{A_b^0A_b^0 A_b^1A_b^0A_b^0 A_b^1,\allowbreak A_b^0A_b^0 A_b^1A_b^0 A_b^1A_b^0,\allowbreak A_b^0A_b^0 A_b^1A_b^0 A_b^1 A_b^1,\allowbreak A_b^0 A_b^1A_b^0A_b^0 A_b^1A_b^0,\allowbreak A_b^0 A_b^1A_b^0A_b^0 A_b^1 A_b^1,\allowbreak A_b^0 A_b^1A_b^0 A_b^1A_b^0 A_b^1,\allowbreak  A_b^1 A_b^1A_b^0 A_b^1A_b^0A_b^0,\allowbreak  A_b^1 A_b^1A_b^0 A_b^1A_b^0 A_b^1,\allowbreak  A_b^1 A_b^1A_b^0 A_b^1 A_b^1A_b^0\}$, which certify normality as outlined above. As the corresponding determinants are lengthy expressions, we abstain from displaying them here. Note that although $b$ within the considered family are generically normal with $L=4$, there indeed exist examples that are normal with injectivity length $L=6$ (but not for a smaller $L$).

Let us now consider the family $b(C_1)$. We consider the following two alternative sets of products with $L=4$, $\{A_b^0A_b^0A_b^0A_b^0,\allowbreak A_b^0A_b^0 A_b^1A_b^0,\allowbreak A_b^0 A_b^1A_b^0A_b^0,\allowbreak A_b^0 A_b^1A_b^0 A_b^1,\allowbreak A_b^0 A_b^1 A_b^1A_b^0,\allowbreak  A_b^1A_b^0A_b^0A_b^0,\allowbreak  A_b^1A_b^0A_b^0 A_b^1,\allowbreak  A_b^1A_b^0 A_b^1A_b^0,\allowbreak  A_b^1 A_b^1 A_b^1 A_b^1\}$,
$\{A_b^0A_b^0A_b^0A_b^0,\allowbreak A_b^0A_b^0A_b^0 A_b^1,\allowbreak A_b^0A_b^0 A_b^1A_b^0,\allowbreak A_b^0 A_b^1 A_b^1A_b^0,\allowbreak A_b^0 A_b^1A_b^0 A_b^1,\allowbreak  A_b^1A_b^0A_b^0A_b^0,\allowbreak  A_b^1A_b^0A_b^0 A_b^1,\allowbreak  A_b^1 A_b^1A_b^0 A_b^1,\allowbreak  A_b^1 A_b^1 A_b^1 A_b^1\}$, as well as the following three alternative sets of products with $L=6$,
$\{A_b^0A_b^0A_b^0A_b^0A_b^0A_b^0,\allowbreak A_b^0A_b^0A_b^0A_b^0A_b^0 A_b^1,\allowbreak A_b^0A_b^0A_b^0A_b^0 A_b^1 A_b^1,\allowbreak A_b^0A_b^0A_b^0 A_b^1 A_b^1 A_b^1,\allowbreak A_b^0A_b^0 A_b^1 A_b^1A_b^0A_b^0,\allowbreak A_b^0 A_b^1 A_b^1A_b^0A_b^0A_b^0,\allowbreak  A_b^1A_b^0A_b^0A_b^0A_b^0A_b^0,\allowbreak  A_b^1A_b^0A_b^0A_b^0A_b^0 A_b^1,\allowbreak  A_b^1 A_b^1 A_b^1 A_b^1 A_b^1 A_b^1\}$,
$\{A_b^0A_b^0A_b^0A_b^0A_b^0A_b^0,\allowbreak A_b^0A_b^0A_b^0A_b^0A_b^0 A_b^1,\allowbreak A_b^0A_b^0A_b^0A_b^0 A_b^1 A_b^1,\allowbreak A_b^0A_b^0 A_b^1 A_b^1A_b^0A_b^0,\allowbreak A_b^0 A_b^1 A_b^1A_b^0A_b^0A_b^0,\allowbreak  A_b^1A_b^0A_b^0A_b^0A_b^0A_b^0,\allowbreak  A_b^1A_b^0A_b^0A_b^0A_b^0 A_b^1,\allowbreak  A_b^1 A_b^1A_b^0 A_b^1 A_b^1A_b^0,\allowbreak  A_b^1 A_b^1A_b^0 A_b^1A_b^0 A_b^1\}$,
$\{A_b^0A_b^0 A_b^1A_b^0A_b^0 A_b^1,\allowbreak A_b^0A_b^0 A_b^1A_b^0 A_b^1A_b^0,\allowbreak A_b^0A_b^0 A_b^1A_b^0 A_b^1 A_b^1,\allowbreak A_b^0 A_b^1A_b^0A_b^0 A_b^1A_b^0,\allowbreak A_b^0 A_b^1A_b^0A_b^0 A_b^1 A_b^1,\allowbreak A_b^0 A_b^1A_b^0 A_b^1A_b^0 A_b^1,\allowbreak  A_b^1 A_b^1 A_b^1 A_b^1 A_b^1 A_b^1,\allowbreak  A_b^1 A_b^1 A_b^1 A_b^1 A_b^1A_b^0,\allowbreak  A_b^1 A_b^1 A_b^1 A_b^1A_b^0A_b^0\}$, which certify normality as outlined above.

Let us now consider the family $b(C_1)$. We consider the following two alternative sets of products with $L=4$,
$\{A_b^0A_b^0A_b^0A_b^0,\allowbreak A_b^0A_b^0A_b^0 A_b^1,\allowbreak A_b^0A_b^0 A_b^1A_b^0,\allowbreak A_b^0A_b^0 A_b^1 A_b^1,\allowbreak A_b^0 A_b^1A_b^0A_b^0,\allowbreak A_b^0 A_b^1A_b^0 A_b^1,\allowbreak  A_b^1A_b^0A_b^0A_b^0,\allowbreak  A_b^1A_b^0 A_b^1A_b^0,\allowbreak  A_b^1A_b^0 A_b^1 A_b^1\}$,
$\{A_b^0A_b^0A_b^0A_b^0,\allowbreak A_b^0A_b^0A_b^0 A_b^1,\allowbreak A_b^0A_b^0 A_b^1 A_b^1,\allowbreak A_b^0 A_b^1A_b^0 A_b^1,\allowbreak A_b^0 A_b^1 A_b^1A_b^0,\allowbreak  A_b^1A_b^0A_b^0A_b^0,\allowbreak  A_b^1 A_b^1A_b^0A_b^0,\allowbreak  A_b^1 A_b^1A_b^0 A_b^1,\allowbreak  A_b^1A_b^0 A_b^1A_b^0\}$,
$\{A_b^0A_b^0A_b^0A_b^0,\allowbreak A_b^0A_b^0A_b^0 A_b^1,\allowbreak A_b^0A_b^0 A_b^1A_b^0,\allowbreak A_b^0 A_b^1A_b^0A_b^0,\allowbreak A_b^0 A_b^1 A_b^1A_b^0,\allowbreak  A_b^1A_b^0A_b^0A_b^0,\allowbreak  A_b^1A_b^0A_b^0 A_b^1,\allowbreak  A_b^1 A_b^1A_b^0 A_b^1,\allowbreak  A_b^1 A_b^1 A_b^1 A_b^1\}$,
 as well as the following three alternative sets of products with $L=5$,
$\{A_b^0 A_b^1 A_b^1A_b^0A_b^0,\allowbreak A_b^0 A_b^1 A_b^1 A_b^1A_b^0,\allowbreak A_b^0 A_b^1 A_b^1 A_b^1 A_b^1,\allowbreak  A_b^1A_b^0A_b^0A_b^0A_b^0,\allowbreak  A_b^1A_b^0A_b^0A_b^0 A_b^1,\allowbreak  A_b^1A_b^0A_b^0 A_b^1 A_b^1,\allowbreak A_b^0A_b^0A_b^0A_b^0A_b^0,\allowbreak A_b^0A_b^0A_b^0 A_b^1 A_b^1,\allowbreak A_b^0A_b^0 A_b^1 A_b^1A_b^0\}$,
$\{A_b^0A_b^0 A_b^1A_b^0A_b^0,\allowbreak A_b^0A_b^0 A_b^1A_b^0 A_b^1,\allowbreak A_b^0A_b^0 A_b^1 A_b^1A_b^0,\allowbreak A_b^0 A_b^1A_b^0A_b^0 A_b^1,\allowbreak A_b^0 A_b^1A_b^0 A_b^1A_b^0,\allowbreak A_b^0 A_b^1 A_b^1A_b^0A_b^0,\allowbreak  A_b^1A_b^0A_b^0 A_b^1A_b^0,\allowbreak  A_b^1A_b^0A_b^0 A_b^1 A_b^1,\allowbreak  A_b^1A_b^0 A_b^1A_b^0A_b^0\}$,
 as well as the following three alternative sets of products with $L=6$,
$\{A_b^0A_b^0A_b^0A_b^0A_b^0 A_b^1,\allowbreak A_b^0A_b^0 A_b^1A_b^0 A_b^1A_b^0,\allowbreak A_b^0A_b^0 A_b^1A_b^0 A_b^1 A_b^1,\allowbreak A_b^0 A_b^1A_b^0A_b^0 A_b^1A_b^0,\allowbreak A_b^0 A_b^1A_b^0A_b^0 A_b^1 A_b^1,\allowbreak A_b^0 A_b^1A_b^0 A_b^1A_b^0 A_b^1,\allowbreak  A_b^1 A_b^1A_b^0 A_b^1A_b^0A_b^0,\allowbreak  A_b^1 A_b^1A_b^0 A_b^1A_b^0 A_b^1,\allowbreak  A_b^1 A_b^1A_b^0 A_b^1 A_b^1A_b^0\}$, and 
$\{A_b^0A_b^0A_b^0A_b^0A_b^0A_b^0,\allowbreak A_b^0A_b^0A_b^0A_b^0A_b^0 A_b^1,\allowbreak A_b^0A_b^0A_b^0A_b^0 A_b^1 A_b^1,\allowbreak A_b^0A_b^0A_b^0 A_b^1 A_b^1A_b^0,\allowbreak A_b^0A_b^0 A_b^1 A_b^1A_b^0A_b^0,\allowbreak A_b^0 A_b^1 A_b^1A_b^0A_b^0A_b^0,\allowbreak  A_b^1A_b^0A_b^0A_b^0A_b^0A_b^0,\allowbreak  A_b^1A_b^0A_b^0A_b^0A_b^0 A_b^1,\allowbreak  A_b^1A_b^0A_b^0A_b^0 A_b^1 A_b^1\}$, which certify normality as outlined above. Within the considered family there indeed exist examples that are normal with injectivity lengths $L=4$, $L=5$, and $L=6$, respectively.

\end{proof}

\section{Proof of Observation \ref{obs:L1L1Tnormality} concerning the normality of MPS generated by fiducial states $\identity \otimes b \otimes \identity \ket{LLT}$}
\label{app:L1L1Tnormality}

In this appendix we restate and prove Observation \ref{obs:L1L1Tnormality} from the main text, which characterizes the normality of MPS generated by fiducial states $\identity \otimes b \otimes \identity \ket{LLT}$.

\noindent {{\bf Observation \ref{obs:L1L1Tnormality}}{\bf.}}\textit{
$N$-qubit MPS associated to $\identity\otimes b \otimes \identity \ket{LLT}$ are normal if and only if $b_{20}\neq 0$ and $\vec{v} \neq \vec{0}$.
}
\begin{proof}
We will first show that the tensor cannot be normal if $b_{20} = 0$.
We will then parametrize $b$ in terms of $T$, $\vec{v}$, $b_{00}$ and $b_{10}$ as in Eq. (\ref{eq:L1L1Tparam}). Making use of Observation \ref{obs:SLOCCnormality} we will argue that it suffices to consider $b$ of a restricted form. We will then show that the tensor cannot be normal if $\vec{v} = \vec{0}$. Finally, we will show that $b$ is normal with injectivity length $L = 4$ if $\vec{v} \neq \vec{0}$.

Let us now show that $b_{20} = 0$ leads to non-normal MPS. For $b_{20} = 0$, we obtain
\begin{align}
\{A_b^0, A_b^1\} = \{ \begin{pmatrix}0 & b_{00} & b_{02} \\0 &b_{10} & b_{12}  \\ 0 & 0 & b_{22} \end{pmatrix},  \begin{pmatrix} b_{00} & 0 & b_{01} \\b_{10} & 0 & b_{11}  \\ 0 & 0 & b_{21} \end{pmatrix}\} 
\end{align}
However, products of matrices of the form $ \begin{pmatrix} \cdot & \cdot & \cdot \\ \cdot & \cdot & \cdot  \\ 0 & 0 & \cdot \end{pmatrix}$ remain of the same form (here, $\cdot$ denots entries with arbitrary values). Hence, $A_b^0$ and $A_b^1$ cannot generate all matrices.

Let us now consider $b$ such that $b_{20} \neq 0$. We choose wlog  $b_{20}=1$. As mentioned above, we now parametrize $b$ in terms of $T$, $\vec{v}$, $b_{00}$, and $b_{10}$. We now make use of the fact that any obtained MPS is SLOCC equivalent to an MPS associated to some $b$ with $T_{10} = 0$, $b_{00} = 0$, and $b_{10} = 0$ (see Section \ref{sec:L1L1T_SLOCC}). Thus, it suffices to consider the normality of MPS for such $b$ in order characterize the normality for all remaining $b$ due to Observation \ref{obs:SLOCCnormality}.

Let us now consider the case $\vec{v} = \vec{0}$. Note that in this case both $A_b^0$ and $A_b^1$ are matrices of the form $ \begin{pmatrix} 0 & 0 & \cdot \\ 0 & 0 & \cdot  \\ \cdot & \cdot & 0 \end{pmatrix}$. It may be easily verified that any product of matrices of such a form comprised of an odd number of factors is again of this form. Moreover, any product comprised of an even number of factors is of the form  $\begin{pmatrix} \cdot & \cdot & 0 \\ \cdot & \cdot & 0  \\ 0 & 0 & \cdot \end{pmatrix}$. Hence, such $b$ cannot lead to normal tensors/MPS.

Let us now consider the case $\vec{v} \neq \vec{0}$. Here we distinguish the two sub-cases  $v_1 \neq 0$ and $v_1 = 0$. In both cases we consider nine products of $A_b^0$ and $A_b^1$ comprised of four factors and show that the products are linearly independent.
Let us first consider the case $v_1 \neq 0$. In this case, we find that the nine products 
\begin{align}
  & \{A_b^0 A_b^0 A_b^1 A_b^0, A_b^0 A_b^1 A_b^1 A_b^0, A_b^0 A_b^1 A_b^1 A_b^1, \nonumber\\  & \qquad A_b^1 A_b^0 A_b^0 A_b^1,   A_b^1 A_b^0 A_b^1 A_b^0,
   A_b^1 A_b^0 A_b^1 A_b^1, \nonumber\\  & \qquad  A_b^1 A_b^1 A_b^0 A_b^1, A_b^1 A_b^1 A_b^1 A_b^0, A_b^1 A_b^1 A_b^1  A_b^1 \}
\end{align}
are linearly independent and thus form a basis for all $3 \times 3$ matrices. This may be easily verified by e.g. considering the determinant of a $9 \times 9$ matrix $M$ whose columns are constructed by rewriting the nine matrices in Eq. (\theequation) as 9-dimensional vectors. We obtain $\det M = -T_{00}^7 T_{11}^6 v_1^{10}$, which is non-vanishing for $v_1 \neq 0$ (note that $T_{00},T_{11}\neq 0$ in order to have $\det b \neq 0$).

In case $v_1 = 0$ we instead consider the products
\begin{align}
  & \{A_b^0 A_b^0 A_b^0 A_b^1, A_b^0 A_b^1 A_b^0 A_b^0, A_b^0 A_b^1 A_b^0 A_b^1, \nonumber\\
   & \qquad A_b^0 A_b^1 A_b^1 A_b^0, A_b^1 A_b^0 A_b^0 A_b^0,  A_b^1 A_b^0 A_b^0 A_b^1, \nonumber\\
   & \qquad  A_b^1 A_b^0 A_b^1 A_b^0, A_b^1 A_b^1 A_b^0 A_b^0, A_b^1 A_b^1 A_b^0 A_b^1\}.
\end{align}
Constructing $M$ as above and using $v_1 = 0$ we obtain $\det M = -T_{00}^{11} T_{11}^5 v_0^4$. As $\vec{v} \neq \vec{0}$ and thus $v_0 \neq 0$, the matrices in Eq. (\theequation) form a basis for all $3 \times 3$ matrices. This completes the proof.
\end{proof}

\section{Derivation of the symmetries of MPS generated by \texorpdfstring{$\identity \otimes b \otimes \identity \ket{LLT}$}{I x b x I |L1 + L1T>}}
\label{app:L1L1Tsym}

In this appendix we provide details on the derivations of the symmetries of MPS generated by $\identity \otimes b \otimes \identity \ket{LLT}$. In particular, we provide the proof of Theorem \ref{thm:L1L1Tconditions} and we derive the procedure of deciding the symmetries for a given $b$ as in the flowchart presented in Figure \ref{fig:L1L1Tflowchart}.

For readability we recite Theorem \ref{thm:L1L1Tconditions} here.\\

\noindent {{\bf Theorem \ref{thm:L1L1Tconditions}}{\bf.}}\textit{
${g}_0, \ldots, {g}_{N-1}$ is an $N$-cylce in $G_b$ \footnote{Recall that we consider here $b_{20} = 1$. We have dealt with the case $b_{20} = 0$ in Observation \ref{obs:L1L1Tb20vansihes}.} if and only if there exist $B_{11}^{(k)} \in \mathbb{C}$ such that for all $k \in \{0, \ldots, N-1\}$,
\begin{align}
\label{eqapp:conjugationrule}
{g}_{k+1} &= \frac{1}{B_{11}^{(k)}B_{11}^{(k+1)}} T {g}_k T^{-1},\\
\label{eqapp:vrule}
\left[ {g}_k - {B_{11}^{(k)}} \identity \right]  \vec{v} &= \vec{0},\\
B_{11}^{(k)}B_{11}^{(k+1)} &= \begin{cases} \pm1 & \text{if } x = i \\ 1 &\text{otherwise}\end{cases}
\end{align}
}

\begin{proof}
In the main text, Eq. (\theequation) has already been proven to be a necessary condition, see Observation \ref{obs:L1L1Teigenvalues}. Let us now prove that Eq. (\ref{eqapp:conjugationrule}) is a necessary condition. To this end, let us consider again the concatenation condition as in Eq. (\ref{eq:lhsrhs}) with the normalizations discussed in the main text. Note also that, as discussed in the main text, due to $b_{20} \neq 0$, the proportionality factor within Eq. (\ref{eq:lhsrhs}) is fixed to $1/B_{11}^{(k)}$. It may be easily seen that the condition is equivalent to the following vector equation,

\begin{widetext}
\begin{align}
\label{eq:vectoreq}
\begin{pmatrix}
 -\delta_{k} & -\gamma_k & 0 & 0\\
 -\delta_{k} b_{21}& -\gamma_k b_{21}& B_{11}^{(k)} b_{00} & 0\\
 -\delta_{k} b_{22}& -\gamma_k b_{22}& 0 & B_{11}^{(k)} b_{00}\\
 -\beta_{k} & -\alpha_k & 0 & 0\\
 -\beta_{k} b_{21}& -\alpha_k b_{21}& B_{11}^{(k)} b_{10} & 0\\
 -\beta_{k} b_{22}& -\alpha_k b_{22}& 0 & B_{11}^{(k)} b_{10}\\
0 & 0 &  B_{11}^{(k)} & 0 \\
 0 &  0 & 0 & B_{11}^{(k)}  
\end{pmatrix}
\begin{pmatrix}
B_{01}^{(k)}\\
B_{02}^{(k)}\\
B_{01}^{(k+1)}\\
B_{02}^{(k+1)}\\
\end{pmatrix} =
\begin{pmatrix}
(B_{11}^{(k)}-\delta_{k}) b_{00} - \gamma_{k} b_{10} \\
B_{11}^{(k)}B_{11}^{(k+1)} (\alpha_{k+1} b_{01} -\gamma_{k+1} b_{02}) - \delta_{k} b_{01} - \gamma_{k} b_{11} \\
B_{11}^{(k)}B_{11}^{(k+1)} (-\beta_{k+1} b_{01} +\delta_{k+1} b_{02})  - \delta_{k} b_{02} - \gamma_{k} b_{12} \\
(B_{11}^{(k)}-\alpha_{k}) b_{10} - \beta_{k} b_{00}\\
B_{11}^{(k)}B_{11}^{(k+1)} (\alpha_{k+1} b_{11} -\gamma_{k+1} b_{12}) - \beta_{k} b_{11} - \alpha_{k} b_{11} \\
B_{11}^{(k)}B_{11}^{(k+1)} (-\beta_{k+1} b_{11} +\delta_{k+1} b_{12})  - \beta_{k} b_{02} - \alpha_{k} b_{12} \\
B_{11}^{(k)} \left( (\alpha_{k+1} B_{11}^{(k+1)} -1) b_{21} - B_{11}^{(k+1)} \gamma_{k+1} b_{22} \right) \\
B_{11}^{(k)} \left( (\delta_{k+1} B_{11}^{(k+1)} -1) b_{22} - B_{11}^{(k+1)} \beta_{k+1} b_{21} \right)
\end{pmatrix}.
\end{align}
Considering this form of the equation, one easily obtains the following set of necessary conditions that are independent of $B_{01}$ and $B_{02}$
\begin{align}
\gamma_k (b_{11} - b_{10}b_{21}) + \delta_k (b_{01} - b_{00}b_{21}) - B_{11}^{(k)}B_{11}^{(k+1)}  \left(\alpha_{k+1} (b_{01} - b_{00}b_{21}) + \gamma_{k+1} (b_{02} - b_{00}b_{22}) \right) &= 0 \nonumber\\
\gamma_k (b_{12} - b_{10}b_{22}) + \delta_k (b_{02} - b_{00}b_{22}) + B_{11}^{(k)}B_{11}^{(k+1)}  \left(\beta_{k+1} (b_{01} - b_{00}b_{21}) + \delta_{k+1} (b_{02} - b_{00}b_{22}) \right) &= 0 \nonumber\\
\alpha_k (b_{11} - b_{10}b_{21}) + \beta_k (b_{01} - b_{00}b_{21}) - B_{11}^{(k)}B_{11}^{(k+1)}  \left(\alpha_{k+1} (b_{11} - b_{10}b_{21}) + \gamma_{k+1} (b_{12} - b_{10}b_{22}) \right) &= 0 \nonumber\\
\alpha_k (b_{12} - b_{10}b_{22}) + \beta_k (b_{02} - b_{00}b_{22}) + B_{11}^{(k)}B_{11}^{(k+1)}  \left(\beta_{k+1} (b_{11} - b_{10}b_{21}) + \delta_{k+1} (b_{12} - b_{10}b_{22}) \right) &= 0,
\end{align} 
\end{widetext}
which are equivalent to Eq. (\ref{eqapp:conjugationrule}) for $T$ as in Eq. (\ref{eq:Tdef}), which proves that Eq. (\ref{eqapp:conjugationrule}) is a necessary condition.

Let us now also prove that Eq. (\ref{eqapp:vrule}) is a necessary condition. To this end, let us again consider Eq. (\ref{eq:vectoreq}) and assume that ${g}_0, \ldots, {g}_{N-1}$ satisfy Eq. (\ref{eqapp:conjugationrule}). It may be easily verified that then, the remaining conditions within Eq. (\ref{eq:vectoreq}) read
\begin{align}
&\begin{pmatrix}
-\delta_k & -\gamma_k \\
-\beta_k & -\alpha_k\\
1 & 0\\
0 & 1
\end{pmatrix}
\begin{pmatrix}
B_{01}^{(k)} \\ B_{02}^{(k)}
\end{pmatrix}
= \nonumber\\
& \qquad
\begin{pmatrix}
-(B_{11}^{(k)} - \delta_k) b_{00} + \gamma_k b_{10} \\
\beta_k b_{00} - (B_{11}^{(k)} - \alpha_k)b_{10} \\
-(\alpha_k B_{11}^{(k)} -1) b_{21} + B_{11}^{(k)} \gamma_k b_{22}\\
 B_{11}^{(k)} \beta_k b_{21} - (\delta_k B_{11}^{(k)} -1) b_{22} 
\end{pmatrix}
\end{align}
for all $k$. Eq. (\theequation) has a solution $(B_{01}^{(k)},B_{02}^{(k)})^T$ for all $k$ if and only if Eq. (\ref{eqapp:vrule}) is satisfied for all $k$. 

Finally, note that if all of the conditions in the theorem are satisfied, then the concatenation condition given in Eq. (\ref{eq:lhsrhs}) is satisfied. This completes the proof of the theorem.
\end{proof}

Let us briefly discuss the condition given in  Eq. (\ref{eqapp:conjugationrule}). Considering an $N$-cycle, in particular, it implies that  ${g}_k = T^N {g}_k T^{-N}$ if $\tr{{g}_k} \neq 0$ ($x \neq i$), or $N$ is even. To see this, we have used Eq. (\ref{eqapp:conjugationrule})  iteratively, moreover, we have used that $x \neq i$ implies that $B_{11}^{(q)}B_{11}^{(q+1)} = 1$ for all $q$, and that for an even $N$ we have that $\prod_{q=0}^{N-1} \left(B_{11}^{(q)}\right)^2 = 1$.
If $N$ is odd and $\tr{{g}_k} = 0$ ($x=i$) instead, we obtain ${g}_k = \pm T^N {g}_k T^{-N}$, i.e., an additional sign freedom emerges. In the former case, we obtain equivalently $[ T^N, {g}_k] = 0$. This condition is satisfied if and only if either $[T,{g}_k] = 0$, or $T \propto R \operatorname{diag}(1,e^{i \frac{2 r \pi}{N}}) R^{-1}$ for some $r \in \{0, \ldots, N-1\}$ and some matrix $R$ (as then  $T^N=\identity$). Note that $[T,{g}_k] = 0$ implies a global symmetry, ${g}_k = {g}$. Considering a minus sign instead, ${g}_0 = - T^N {g}_0 T^{-N}$, we additionally obtain solutions of the form $T = S_0 \tilde{D} H  \operatorname{diag}(1,e^{i \frac{(2 r+1)\pi}{N}}) H \tilde{D}^{-1} S_0^{-1}$ for some $r \in \{0, \ldots, N-2\}$ and some diagonal matrix $\tilde{D}$, where $H$ denotes the Hadamard matrix.

In the following discussion, let us now also take the condition in Eq. (\ref{eqapp:vrule}) into account. One possiblility to satisfy this condition is to have $b_{00} = - b_{21}$ as well as $b_{10} = - b_{22}$. Then, $\vec{v} = \vec{0}$ and Eq. (\ref{eqapp:vrule}) is obviously satisfied for all $k$. The second possibility is that all ${g}_k$ share a common eigenvector, $\vec{v}$, corresponding to the eigenvalue $B_{11}^{(k)}$, respectively. We will see that this second option limits the symmetries of the MPS severely. To this end, first note that $B_{11}^{(k)}$ must equal one of the eigenvalues of ${g}_k$, $x$ or $1/x$. Moreover, using Eq. (\ref{eqapp:vrule}) for $k+1$, inserting ${g}_{k+1}$ as in Eq. (\ref{eqapp:conjugationrule}) and using $\left( B_{11}^{(k+1)} \right)^2 = 1/\left( B_{11}^{(k)} \right)^2 $ yields another useful condition,
\begin{align}
T \left[  {g}_k  - \frac{1}{B_{11}^{(k)}} \identity \right]  T^{-1}\vec{v} =  \vec{0}.
\end{align}
In other words, for each $k$, the eigenvectors of ${g}_k$ are given by $\vec{v}$ and $T^{-1}\vec{v}$, corresponding to the eigenvalues $B_{11}^{(k)}$ and $1/B_{11}^{(k)}$, respectively. One more notable consequence is that considering once more Eq. (\ref{eqapp:conjugationrule}) one obtains that either $T \vec{v} \propto \vec{v}$, or $T \vec{v} \not\propto \vec{v}$ and $T^2 \vec{v} \propto \vec{v}$ (here we assumed ${g}_k \not\propto \identity$ to disregard trivial solutions). In the former case we must have that $1/B_{11}^{(k)}= B_{11}^{(k)} = 1/x = x = 1$. Thus, only non-diagonlizable $g_k$ are possible. In the latter case we have alternating symmetries with ${g}_{k+1} = {g}_k^{-1}$ for all $k$ (we obtain a global symmetry in case $x=i$). Having established these useful facts, we are now in the position to derive the process for determining all cycles in $G_b$ for a given $b$.

\subsection{Deriving the process depicted in Figure \ref{fig:L1L1Tflowchart}, which determines the cycles in $G_b$ for a given \texorpdfstring{$b$}{b}}
We will now derive the process of determining the cycles for a given $b$ with $b_{20}=1$ as depicted in Figure \ref{fig:L1L1Tflowchart}. The first step is to calculate $T$ and $\vec{v}$ for the given $b$. We now consider the Jordan decomposition of $T$, $T = R J R^{-1}$, where $J$ is the JNF of $T$. Depending on whether $T$ is diagonalizable, we now distinguish two cases. 

Let us first deal with the case that $T$ is not diagonalizable. Suppose we haven an $N$-cycle ${g}_0, \ldots, {g}_{N-1}$. Then, the only possibility to fulfill Eq. (\ref{eqapp:conjugationrule}) is that $[T,{g}_k] = 0$. To see this, consider a consequence of Eq. (\ref{eqapp:conjugationrule}), ${g}_{k} T^N  {g}_{k}^{-1}  = \frac{1}{\prod_k \left(B_{11}^{(k)}\right)^2} T^{N}$, which shows that $\prod_k \left( B_{11}^{(k)}\right)^2 = 1$. Thus we have $[T^N,{g}_k] = 0$. From this, $[T,{g}_k] = 0$ follows. It immediately follows that only global symmetries are posssible, as ${g}_{k+1} \propto {g}_k$ due to Eq. (\ref{eqapp:conjugationrule}). Moreover, using $[T,{g}] = 0$ we obtain that for all potential symmetries, ${g} = R \begin{pmatrix}1 & y \\ 0 & 1 \end{pmatrix} R^{-1}$ for any $y \in \mathbb{C}$. It yet remains to consider Eq. (\ref{eqapp:vrule}). Recall that both $\vec{v}$ and $T \vec{v}$ are eigenvectors of ${g}$. As ${g}$ is not diagonalizable in the currently considered case, we must have $T \vec{v} \propto \vec{v}$ (or $\vec{v} = 0$). This leads to another case distinction. If $T \vec{v} \not\propto \vec{v}$, then we have no non-trivial cycles. However, if $T \vec{v} \propto \vec{v}$, then with ${g}$ as given above, all conditions in Theorem \ref{thm:L1L1Tconditions} are satisfied. Thus, in this case we obtain the one-parametric family of gloabal symmetries given above. This compeltes the case that $T$ is not diagonalizable and is shown in the left branch of Figure \ref{fig:L1L1Tflowchart}.

Let us now discuss the case that $T$ is diagonalizable. We now additionally distinguish the case that there exists an $m \in \mathbb{N}$ such that $T^m \propto \identity$ from the case that there does not exist such an  $m$. Let us first discuss the latter case. Similarly to before, we suppose that we have an $N$-cycle and obtain ${g}_0 T^N {g}_0^{-1} = \frac{1}{\prod_k \left(B_{11}^{(k)}\right)^2} T^{N}$. We must have $\prod_k \left(B_{11}^{(k)}\right)^2=1$, as otherwise $T^{2N} \propto \identity$, contradicting the assumption. Thus, we obtain $[{g}_0,T^N] = 0$, which implies $[{g}_0,T] = 0$, as $T^N \not\propto \identity$. Hence, we have global symmetries only and moreover, ${g} = R \operatorname{diag}(x,1/x) R^{-1}$ for any $x \in \mathbb{C} \setminus \{0\}$. If $\vec{v} = \vec{0}$ the conditions in Theorem \ref{thm:L1L1Tconditions} are satisfied and one indeed obtains the mentioned one-parametric family of global symmetries. Let us now consider the case $\vec{v} \neq \vec{0}$. Note that if $T \vec{v} \propto \vec{v}$ both eigenvalues of ${g}$ coincide, which leads to trivial cycles only (recall that ${g}$ must be diagonalizable). Recall that $T^2 \vec{v} \propto \vec{v}$ is a necessary condition to have non-trivial cycles. Thus, only the case $T \vec{v} \not\propto \vec{v}$ and $T^2 \vec{v} \propto \vec{v}$ remains. Note, however, that if $T^2 \vec{v} = \lambda \vec{v}$ for some $\lambda \in \mathbb{C}$, then also $T^2 T \vec{v} = \lambda T \vec{v}$. As $\vec{v}$ and $T \vec{v}$ are linearly independent, this implies that $T^2 = \lambda \identity$, which is contradicting the assumption that $T^N\not\propto \identity$ for any $N \in \mathbb{N}$. Hence, in case $\vec{v} \neq \vec{0}$   we have a trivial symmetry only.

Let us finally disucuss the case that $T$ is diagonalizable and moreover, there is an $m \in \mathbb{N}$ such that $T^m \propto \identity$. In this case one can write $T \propto R \begin{pmatrix}e^{i \frac{ r \pi}{m}} & \\ & e^{-i \frac{ r \pi}{m}} \end{pmatrix} R^{-1}$ for some $r\in \mathbb{N}$. 
If $\vec{v} = \vec{0}$, we first obtain the same global symmetries as in case $T^m \not\propto \identity$, ${g} = R \operatorname{diag}(x,1/x) R^{-1}$ for any $x \in \mathbb{C} \setminus \{0\}$. Moreover, we obtain $m$-cycles of the form ${g}_k = T^k {g}_0 T^{-k}$ for any ${g}_0$, which (taking normalization into account) effectively constitutes a three-parametric family of non-global symmetries. Thus, in this case we have diagonalizable as well as non-diagonalizable symmetries. In case $m$ is even, additionally certain $m/2$-cycles emerge, which stem from the fact that in case $x=i$, we may have ${g}_0 = - T^{m/2} {g}_0 T^{m/2}$. As discussed earlier, this admits solutions ${g}_0 = S_0 \operatorname{diag}(i, -i) S_0^{-1}$ such that $T = S_0 \tilde{D} H  \operatorname{diag}(e^{i \frac{(2 r+1)\pi}{m}},e^{-i \frac{(2 r+1)\pi}{m}}) H \tilde{D}^{-1} S_0^{-1}$ for some $r \in \{0, \ldots, m-2\}$ and some diagonal matrix $\tilde{D}$. It may be easily verified that this leads to the one-parametric family of $m/2$-cycles with ${g}_k = R \begin{pmatrix} 0 & i y e^{i \frac{2 k (2 r + 1) \pi}{m}} \\ i/y e^{-i \frac{2 k (2 r + 1) \pi}{m}} & 0\end{pmatrix}R^{-1}$ for $y \in \mathbb{C} \setminus \{0\}$. 
Let uns now discuss the case that $\vec{v} \neq 0$, in which additional restrictions must be satisfied. Let us first discuss the subcase $T \vec{v} \propto \vec{v}$. Recall that in this subcase, only cycles with non-diagonalizable ${g}_k$ are possible (disregarding trivial cycles). We thus write ${g}_0 = S_0 \begin{pmatrix}1 & 1 \\ 0 & 1 \end{pmatrix}S_0^{-1}$, where $S_0 = \left(\vec{v}, \vec{w} \right)$ with a freely choosable generalized eigenvector $\vec{w} \in \mathbb{C}^2$. One then obtains the remaining matrices forming an $m$-cycle ${g}_k = S_0 \begin{pmatrix}1 & e^{\frac{2 k r \pi}{m}} \\ 0 & 1 \end{pmatrix}S_0^{-1}$ for $k \in \{0, \ldots, m-1\}$. Thus, in this case we obtain an effectively one-parametric family of non-diagonalizable, non-global symmetries.
Let us now discuss the remaining subcases. As $\vec{v}$ such that $T^2 \vec{v} \not\propto \vec{v}$ do not allow for any non-trivial cycle, actually the only subcase that remains open is the case that $T \vec{v} \not\propto \vec{v}$, but $T^2 \vec{v} \propto \vec{v}$. Note that in this case we have $T^2 \propto \identity$.
 Clearly, all ${g}_k$ must be diagonalizable, as each ${g}_k$ possesses the two linearly independent vectors $\vec{v}$ and $T \vec{v}$ as eigenvectors. It may be easily seen that a global symmetry with eigenvectors $\vec{v}$ and $T \vec{v}$ and eigenvalues given by $\pm i$ satisfies the conditions given in Theorem \ref{thm:L1L1Tconditions}. Considering cycles with ${g}_k$ such that $x \neq i$, one obtains that ${g}_{k+1} = {g}_k^{-1}$, as the vectors $\vec{v}$ and $T\vec{v}$ correspond to the eigenvalues $x$ and $1/x$ in an alternating manner. Thus, one obtains a one-parametric family of two-cycles with a freely choosable $x \in \mathbb{C} \setminus \{0,\pm i\}$. Clearly, however, for an odd partilcle number only the global symmetry with $x=i$ remains.

\section{Details on the SLOCC classification of $\Psi(LLT)$ and proofs of Lemmata \ref{lemma:L1L1TSLOCCsameT} and \ref{lemma:L1L1TSLOCCglobalsuffices} }
\label{app:L1L1TSLOCC}
In this appendix, we first provide details on the characterization of $(b \rightarrow c)$ cycles for fiducial states $\ket{LLT}$. We first state and prove Observation  \ref{obs:L1L1TSLOCCobs}. We then use the observation in order to prove Lemmata \ref{lemma:L1L1TSLOCCsameT} and \ref{lemma:L1L1TSLOCCglobalsuffices}, which lead to the SLOCC classification of normal MPS (Theorem \ref{thm:L1L1TSLOCC}), as explained in the main text.

From Eq. (\ref{eq:L1L1Tsloccnec}) it follows that $T_b$ and $T_c$ are either both diagonalizable, or both non-diagonalizable and we may deal with SLOCC equivalence for these two cases separately. By considering Eq. (\ref{eq:L1L1Tslocctrafo}) one straightforwardly obtains a few further necessary conditions for having an $(b\rightarrow c)$ $N$-cycle. First, $\vec{v}_b = \vec{0}$ if and only if $\vec{v}_c = \vec{0}$. Second, for any $k \in \mathbb{N}$ it holds that $T_b^k \vec{v}_b \propto \vec{v}_b$ if and only if $T_c^k \vec{v}_c \propto \vec{v}_c$. To see this, suppose that $T_c^k \vec{v}_c \propto \vec{v}_c$ and consider
\begin{align}
T_b^k \vec{v}_b &\propto g_{k} T_c g_{k-1}^{-1} g_{k-1} T_c \ldots T_c g_{0}^{-1} g_0 \vec{v}_c = g_k T_c^k \vec{v}_c \nonumber\\
&\propto g_k \vec{v}_c \propto \vec{v}_b.
\end{align}
Third, there exists an $m_b$ such that $T_b^{m_b} \propto \identity$ if and only if there exists an $m_c$ such that $T_c^{m_c} \propto \identity$. To see this, suppose that there exists an $m_c$ such that $T_c^{m_c} \propto \identity$. Then, consider $T_b^{N m_c}$, where $N$ is the particle number. Due to Eq. (\ref{eq:L1L1Tsloccnec}) we have $T_b^{N m_c} \propto g_k T_c^{N m_c} g_k^{-1} \propto \identity$. Thus, we have $T_b^{m_b} \propto \identity$ with $m_b = N m_c$.  Note, however, that we do not necessarily have $T_b^{m_b} \propto \identity$ for $m_b = m_c$. We present a simple counter example below.
Finally, let us remark that if $T_c \vec{v}_c \not\propto \vec{v}_c$, then we have that for any $m \in \mathbb{N}$, $T_c^{m} \propto \identity$ if and only if $T_b^{m} \propto \identity$ (with the same $m$). This can be seen as follows. Suppose that $T_c^{m} \propto \identity$. Then we have $T_c^{m} \vec{v}_c \propto \vec{v}_c$. As discussed above we thus have $T_b^{m} \vec{v}_b = \lambda \vec{v}_b$ for some $\lambda \in \mathbb{C}$. Furthermore, we have that  $T_b^{m} T_b \vec{v}_b = \lambda T_b \vec{v}_b$. As $\vec{v}_b$ and  $T_b \vec{v}_b$ are linearly independent by the assumption, we have that $T_b^{m} = \lambda \identity$. Let us summarize all of the discussed properties in the following observation. 
\begin{observation}
\label{obs:L1L1TSLOCCobs}
Consider $b$ and $c$ such that there exists an $(b\rightarrow c)$ $N$-cycle\footnote{Recall that we are considering $b_{20} = c_{20} = 1$.}. Then we have
\begin{enumerate}[(i)]
\item $T_b$ is diagonalizable if and only if $T_c$ is.
\item $\vec{v}_b = \vec{0}$ if and only if $\vec{v}_c = \vec{0}$.
\item For any $k \in \mathbb{N}$ we have $T_b^k \vec{v}_b \propto \vec{v}_b$ if and only if $T_c^k \vec{v}_c \propto \vec{v}_c$.
\item There exists an $m_b\in \mathbb{N}$ such that $T_b^{m_b} \propto \identity$ if and only if there exists an $m_c \in \mathbb{N}$ such that $T_c^{m_c} \propto \identity$.
\item If $T_c \vec{v}_c \not\propto \vec{v}_c$, then for any $m \in \mathbb{N}$ we have $T_b^{m} \propto \identity$ if and only if  $T_c^{m} \propto \identity$.
\end{enumerate}
\end{observation}
We formulate the observation in terms of cycles rather than SLOCC equivalence of MPS as the considered family of $b,c$  encompasses non-normal instances, cf. the discussion below Theorem \ref{thm:L1L1Tconditions}.
With this observation, we have established that normal MPS belonging to different boxes within Figure \ref{fig:L1L1Tflowchart} are SLOCC inequivalent (even though they may have compatible symmetry group).

In Observation \ref{obs:L1L1TSLOCCobs} we have shown that if there is a $(b \rightarrow c)$ cycle and $T_c^{m_c} \propto \identity$ for some $m_c$, then $T_b^{m_b} \propto \identity$ for some $m_b$.
However, in the discussion above the observation we have mentioned that we do not necessarily have $T_b^{m_b} \propto \identity$ for $m_b = m_c$. Here we present a simple counter example illustrating this. Consider 
$c$ such that $\vec{v}_c = \vec{0}$ and $T_c = \operatorname{diag}(e^{i \frac{\pi}{m_c}}, e^{-i \frac{\pi}{m_c}})$, i.e., $T_c^{m_c} \propto \identity$. Then, $b$ with $\vec{v}_b = \vec{0}$ and $T_b = \operatorname{diag}( e^{i (\frac{2}{N}+\frac{1}{m_c})\pi},e^{-i (\frac{2}{N}+\frac{1}{m_c})\pi})$ gives rise to an SLOCC equivalent MPS. This can be easily seen by writing $T_b = g_{k+1} T_c g_{k}^{-1}$, where $g_k = \operatorname{diag}(e^{i \frac{2 k \pi}{N}},e^{-i \frac{2 k \pi}{N}})$. Suppose that $N$ is odd, moreover, $m_c$ and $N$ are coprime. Then, $m_b = N m_c$ is the smallest integer such that $T_b^{m_b} \propto \identity$, in particular,  $T_b^{m_c} \not\propto \identity$.
\\

\noindent {{\bf Lemma \ref{lemma:L1L1TSLOCCsameT}}{\bf.}}\textit{
Consider $b$ and $c$ in standard form which correspond to normal MPS (i.e., $b_{20} = c_{20} = 1$ and additionally $\vec{v}_b,\vec{v}_c \neq \vec{0}$). If there exists a  $(b\rightarrow c)$ $N$-cycle, then, $T_b = T_c$. 
}
\begin{proof}
In order to prove the lemma, we distinguish several cases, namely, non-diagonalizable and diagonalizable $T_c$ and, moreover, $T_c$ such that $T_c^m \propto \identity$ for some $m$ and $T_c$ such that $T_c^m \not\propto \identity$ for all $m$. We show the statement of the lemma for each of these cases separately.

Let us first consider non-diagonalizable $T_c$. Due to Observation \ref{obs:L1L1TSLOCCobs} and the chosen standard form for $T$ we have $T_b= T_c = \begin{pmatrix}1&1\\0&1 \end{pmatrix}$. The statement is hence trivial for non-diagonalizable $T$. 

Let us now consider diagonalizable $T_b$, $T_c$. 
Consider first the case that there exists no $m_c$ such that $T_c^{m_c}\propto \identity$. Due to Observation \ref{obs:L1L1TSLOCCobs}, the same must hold for $T_b$. Considering the standard form for $T$ we may write $T_c = \operatorname{diag}(\sigma_c, \sigma_c^{-1})$ and $T_b = \operatorname{diag}(\sigma_b, \sigma_b^{-1})$. Using Eq. (\ref{eq:L1L1Tsloccnec}) we necessarily have that either $\sigma_b^N = \pm \sigma_c^N$, or $\sigma_b^N = \pm \sigma_c^{-N}$ (with a positive sign in case of even $N$). Note that the latter case is only possible if $|\sigma_c| = 1$. Thus, $\sigma_b = e^{i \frac{q \pi}{N}} \sigma_c$ or $\sigma_b = e^{i \frac{q \pi}{N}} \sigma_c^{-1}$ for some $q \in \{0, \ldots, 2N-1\}$ in case of odd $N$ and $\sigma_b = e^{i \frac{2 q \pi}{N}} \sigma_c$ or $\sigma_b = e^{i \frac{2 q \pi}{N}} \sigma_c^{-1}$ for some $q \in \{0, \ldots, N-1\}$ in case of even $N$. Moreover, using Eq. (\ref{eq:L1L1Tsloccnec}) we obtain that either $g_k = \operatorname{diag}(\alpha_k,1/\alpha_k)$ for all $k$, or $g_k = \operatorname{diag}(\alpha_k,1/\alpha_k) \sigma_x$ for all $k$, respectively\footnote{Let us remark that in case $\vec{v}_c = \vec{0}$ we also have $\vec{v}_b = \vec{0}$ and any $\sigma_b$ as in the main text indeed leads to an SLOCC equivalent MPS, which can be seen by considering $g_k = \operatorname{diag}(e^{-i \frac{k q \pi}{N}},e^{i \frac{k q \pi}{N}}) \sigma_x^b$ and $B_{11}^{(k)}= (-1)^k i$, where $b \in \{0,1\}$ for odd $N$ and $g_k = \operatorname{diag}(e^{-i \frac{2 k  q \pi}{N}},e^{i \frac{2 k  q \pi}{N}}) \sigma_x^b$ and $B_{11}^{(k)}=1$ for even $N$. Hence, there is a one-parametric family of SLOCC classes corresponding to box (III) in Figure \ref{fig:L1L1Tflowchart}. }.
Due to the assumption we have $\vec{v}_c \neq \vec{0}$. Due to Eq. (\ref{eq:L1L1Tslocctrafo}) we have
 \begin{align}
 \label{eq:L1L1Tvcomparison}
 \frac{B_{11}^{(k)}}{B_{11}^{(k+1)}} g_{k}^{-1} g_{k+1} \vec{v}_c = \vec{v}_c,
 \end{align}
 and due to the considerations above we moreover have
 \begin{align}
 g_k^{-1} g_{k+1} = {B_{11}^{(k+1)}}{B_{11}^{(k)}} \operatorname{diag}(e^{i \frac{q \pi}{N}}, e^{-i \frac{q \pi}{N}})
 \end{align}
 for $q \in \{0, \ldots, 2N-1\}$.
Thus, $\left(B_{11}^{(k)}\right)^2 \operatorname{diag}(e^{i \frac{q \pi}{N}}, e^{-i \frac{q \pi}{N}})  \vec{v}_c = \vec{v}_c$. Hence, either $q=0$, which implies $T_b = T_c$ ($T_b = T_c^{-1}$ is not possible due to the chosen standard form), or  $\vec{v}_c$ is proportional to a standard basis vector, i.e., an eigenvector of $T_c$. In the latter case we have $\left(B_{11}^{(k)}\right)^2 = e^{\pm i \frac{q \pi}{N}}$ with coinciding sign for all $k$. As we are considering $T$ in standard form ($\det T_b = \det T_c =1$), we have $\left(B_{11}^{(k)}\right)^2 \left(B_{11}^{(k+1)}\right)^2 = 1$. We thus obtain $e^{i \frac{2 q \pi}{N}} = 1$ and in further consequence $\sigma_b = \pm \sigma_c$, or $\sigma_b = \pm  \sigma_c^{-1}$. Hence, $T_b = T_c$ due to the standard form.

Let us now consider the case that there exists an $m_c$ such that $T_c^{m_c}\propto \identity$. Here we distinguish two subcases. First, if $T_c^{N}\not\propto \identity$, then the same conclusions as for the case $T_c^{m_c} \not\propto \identity$ for any $m_c \in \mathbb{N}$ can be drawn. Second, if $T_c^{N} \propto \identity$, then Eq. (\ref{eq:L1L1Tsloccnec}) does not yield a constraint. In particular, it does not imply that $g_k$ must be of either diagonal or counter diagonal form. Thus, in the following, we deal with this case separately.
Due to Eq. (\ref{eq:L1L1Tslocctrafo}) we have
\begin{align}
g_k = \left(\prod_{i=0}^{k-1} B_{11}^{(i)}\right) \left(\prod_{i=1}^{k} B_{11}^{(i)}\right) T_b^k g_0 T_c^{-k}
\end{align}
for all $k$. We now consider $\frac{1}{B_{11}^{(k)}} g_k \vec{v}_c = \frac{1}{B_{11}^{(0)}} g_0 \vec{v}_c$ [an implication of Eq. (\ref{eq:L1L1Tslocctrafo})]. Inserting Eq. (\theequation) we obtain the condition
\begin{align}
\underbrace{\left( \left(  \prod_{i=0}^{k-1}B_{11}^{(i)}\right)^2  T_b^k g_0 T_c^{-k} - g_0 \right)}_{M_k} \vec{v}_c = \vec{0}.
\end{align}
As $\left(B_{11}^{(l)}B_{11}^{(l+1)} \right)^2 = 1$ for any $l$, the product in Eq. (\theequation) equals 1 for any even $k$. As $T_c^{N},T_b^{N} \propto \identity$ and due to the standard form we have $T_b = \operatorname{diag}(e^{i\frac{q_b \pi}{N}},e^{-i\frac{q_b \pi}{N}})$ and $T_c = \operatorname{diag}(e^{i\frac{q_c \pi}{N}},e^{-i\frac{q_c \pi}{N}})$ for some $q_b,q_c \in \{0, \ldots, \lfloor N/2 \rfloor\}$. Due to Eq. (\theequation) we have 
$\det M_k = 0$ for any $k$. Moreover, due to the definition of $M_k$ we have that
\begin{align}
\det M_k = 2 -  2&\left( \alpha_0 \delta_0 \cos \frac{k \pi (q_b-q_c)}{N}\right. \nonumber\\
  &\quad \left. -\beta_0 \gamma_0 \cos \frac{k \pi (q_b + q_c)}{N} \right)
\end{align}
for any even $k$. We have
\begin{align}
\sum_{k \in \{2, 4, \ldots, 2(N-1)\}} \det M_k = \begin{cases}2N & q_b \neq q_c\\0 &q_b = q_c \in \{0,\frac{N}{2}\}\\-2N\beta_0 \gamma_0 & q_b = q_c \not\in \{0,\frac{N}{2}\} \end{cases}
\end{align}
Note that we deliberately also sum over $k$ that are larger than $N$, we use $N+l \equiv l$ for $l \in \{0, \ldots, N-1\}$. 
As $M_k$ must be singular for any $k$ this shows that $q_c = q_b$ and hence $T_c = T_b$. This completes the proof of the lemma.
\end{proof}

\noindent {{\bf Lemma \ref{lemma:L1L1TSLOCCglobalsuffices}}{\bf.}}\textit{
Consider $b$ and $c$ which correspond to normal MPS (i.e., $b_{20} = c_{20} = 1$ and additionally $\vec{v}_b,\vec{v}_c \neq \vec{0}$). If there exists a  $(b\rightarrow c)$ $N$-cycle, then, there also exists a $(b\rightarrow c)$ $1$-cycle. 
}
\begin{proof}
We prove the statement separately for the case of non-diagonalizable $T_c$, diagonalizable $T_c$ such that $T_c^N \propto \identity$ and diagonalizable $T_c$ such that $T_c^N \not\propto \identity$. We make use of the fact that   $T_b = T_c$ due to Lemma \ref{lemma:L1L1TSLOCCsameT} in all of the cases.

Let us start by considering non-diagonalizable $T_c$.  Due to Eq. (\ref{eq:L1L1Tslocctrafo}) we have $g_{k+1}\propto T_c g_k T_c^{-1}$. From the condition in Eq. (\ref{eq:L1L1Tsloccnec}) it follows that $[g_k, T_c]  = 0$. Thus, $g_{k+1}\propto g_k$ for all $k$ and the statement of the lemma follows for non-diagonalizable $T_c$.

Let us now consider the case that $T_c$ is diagonalizable and $T_c^N \not\propto \identity$. Recall that due to the considered standard form of $c$, $T_c$ is diagonal. Due to Eq. (\ref{eq:L1L1Tsloccnec}) we have that $g_k$ is either diagonal or counter-diagonal. Using $g_{k+1} \propto T_c g_k T_c^{-1}$ we have that $g_{k+1} \propto g_k$ and the statement of the lemma follows for the considered $T_c$.

Let us finally consider the case that $T_c$ is diagonalizable and $T_c^N \propto \identity$. Note that if $T_c \propto \identity$ we have $g_{k+1} \propto T_c g_k T_c^{-1} \propto g_k$ and the statement of the lemma follows trivially. In the following we thus assume $T_c \not\propto \identity$. 
We will make use of the relation
\begin{align}
\left(T_c^2 g_0 T_c^{-2} - g_0 \right) \vec{v}_c = \vec{0},
\end{align}
which may be derived as in the proof of Lemma \ref{lemma:L1L1TSLOCCsameT}. We now distinguish several subcases.
We first consider the subcase that $T_c \vec{v}_c \propto \vec{v}_c$. We either have that $\vec{v}_c \propto (1,0)^T$, or $\vec{v}_c \propto (0,1)^T$.  Suppose that  $\vec{v}_c \propto (1,0)^T$. Note that we are considering $c$ corresponding to box (IV) in Figure \ref{fig:L1L1Tflowchart}.
Considering $g_{k+1} \propto T_c g_k T_c^{-1}$ and $g_0 \vec{v}_c \propto g_k \vec{v}_c$ for any $k$ we have that $g_0 \vec{v}_c \propto T_c^k g_0 T_c^{-k} \vec{v}_c \propto T_c^k g_0 \vec{v}_c$. Considering $k=1$ and using $T_c \not\propto \identity$ we have that $g_0 \vec{v}_c$ must be proportional to a standard basis vector. Thus, either $\alpha_0 = 0$ or $\gamma_0 = 0$. Unless $T_c^2 \propto \identity$ we must have $\gamma_0 =0$ due to Eq. (\theequation).
Due to the fact that the symmetry group of the considered fiducial states possesses an $N$-cycle yielding a symmetry as displayed in box (IV) in Figure \ref{fig:L1L1Tflowchart}, there must as well exist an $(b \rightarrow c)$ $N$-cycle with physical operators $g'_0, g'_1, \ldots$, where $g'_0 = g_0 \begin{pmatrix}1 & z \\ 0 & 1 \end{pmatrix}$ for any $z \in \mathbb{C}$. Choosing $z = -\beta_0/\alpha_0$ ($z = -\delta_0/\gamma_0$) if $\gamma_0 = 0$ ($\alpha_0 = 0$), we obtain a $(b \rightarrow c)$ $N$-cycle with $g'_0 =  \begin{pmatrix}\alpha_0 & 0 \\ 0 & \delta_0 \end{pmatrix}$ ($g'_0 =  \begin{pmatrix}0 & \beta_0 \\ \gamma_0 & 0 \end{pmatrix}$), respectively. Recall that the latter case can only occur if $T_c^2 \propto \identity$, which implies $T_c = \operatorname{diag}(i, -i)$ due to the chosen standard form. As $g'_{k+1} \propto T_c g'_{k+1} T_c^{-1}$ we hence have that $g'_{k+1} \propto g'_0$ for all $k$. Thus, we have an $(b\rightarrow c)$ 1-cycle. The proof works similarly if $\vec{v}_c \propto (0,1)^T$, instead.

Let us now consider the subcase that $T_c \vec{v}_c \not\propto \vec{v}_c$ and $T_c^2 \propto \identity$, i.e., we are considering $c$ corresponding to box (VII) in Figure \ref{fig:L1L1Tflowchart}. It will be convenient to introduce $S = (\vec{v}_c, T\vec{v}_c)$ and write $g_0= S \begin{pmatrix}\alpha_{0,S} & \beta_{0,S} \\ \gamma_{0,S} & \delta_{0,S} \end{pmatrix} S^{-1}$. Note that $T_c S \propto S \sigma_x$. Thus, the condition $g_0 \vec{v}_c \propto g_1 \vec{v}_c \propto T_c g_0 T_c^{-1} \vec{v}_c$ yields $\sigma_x (\beta_{0,S},  \delta_{0,S})^T \propto (\alpha_{0,S},  \gamma_{0,S})^T$. Thus, we have $g_0= S \begin{pmatrix}\alpha_{0,S} & \lambda \gamma_{0,S} \\ \gamma_{0,S} & \lambda \alpha_{0,S} \end{pmatrix} S^{-1}$ for some $\lambda \in \mathbb{C}$. Note that the considered subcase can only occur in case of even $N$. Thus, we may untilize the $2$-cycles displayed in box (VII) in Figure \ref{fig:L1L1Tflowchart} in order to conclude that there must exist another $(b \rightarrow c)$ $N$-cycle with $g_0', g_1', \ldots$ such that $g_0'= S \begin{pmatrix}\alpha_{0,S}' &  \gamma_{0,S}' \\ \gamma_{0,S}' &  \alpha_{0,S}' \end{pmatrix} S^{-1}$. Now, as $g_{k+1}' \propto T_c g_k' T_c^{-1} =  S \sigma_x \begin{pmatrix}\alpha_{k,S}' &  \gamma_{k,S}' \\ \gamma_{k,S}' &  \alpha_{k,S}' \end{pmatrix}\sigma_x S^{-1}$, we have $g_k' \propto g_0'$ for all $k$. Thus, there exists a  $(b \rightarrow c)$ $1$-cycle.

Let us finally consider the subcase that $T_c \vec{v}_c \not\propto \vec{v}_c$ and $T_c^2 \not\propto \identity$, i.e., we are considering $c$ corresponding to box (VIII) in Figure \ref{fig:L1L1Tflowchart}. As $T_c^2 \not\propto \identity$, Eq. (\theequation) yields $\begin{pmatrix}0 & \beta_0 \\ \gamma_0 & 0 \end{pmatrix} \vec{v}_c = \vec{0}$. As $\vec{v}_c$ is not an eigenvector of $T_c$, i.e., not proportional to a standard basis vector, this yields $\beta_0 = 0$ and $\gamma_0 = 0$, i.e., $g_0$ is diagonal. As $g_{k+1} \propto T_c g_k T_c^{-1}$ and $T_c$ is diagonal too, we have $g_k \propto g_0$ for all $k$. Thus, there exists a  $(b \rightarrow c)$ $1$-cycle. This completes the proof of the lemma.
\end{proof}

\clearpage
\section{Symmetries of the TIMPS corresponding to the remaining 4 SLOCC classes of the fiducial state}
\label{sec:otherSLOCCclasses}

We briefly discuss here the four remaining SLOCC classes of the fiducial states. The subsequent subsections are all structured in the same way. We present the symmetries of the fiducial states and the concatenation rules, which would then, analogously to the derivation in the main text, allow to determine the symmetries, SLOCC transformations, and SLOCC classes of the corresponding TIMPS. Instead of completely characterizing here all the symmetries (and SLOCC classes), we just highlight some particularly interesting symmetries, which do not occur in the cases studied in the main text. The notation used for the representatives refer to the value and the degeneracy of the eigenvalues of the corresponding matrix pencil (see also Sec \ref{sec:prelim}). 

\subsection{TIMPS corresponding to the fiducial states represented by $M^3(0)$}\label{sec:M30}

We consider here the SLOCC class of the fiducial state which is represented by the state by $\ket{M^3(0)} = \ket{0}(\ket{01}+\ket{12}) + \ket{1}(\ket{00}+\ket{11}+\ket{22})$.

The matrix pencil reads
\begin{align}
\mathcal{P} = M^3(0) = \begin{pmatrix}\lambda & \mu &0\\0& \lambda  & \mu\\ 0 & 0 & \lambda \end{pmatrix}.
\end{align}
\subsubsection{Symmetries of the fiducial state}

As before, the symmetries of the fiducial state can be straight forwardly determined using the corresponding matrix pencil (or directly from the state). Symmetries of the form  $\identity \otimes B \otimes C$ read
\begin{align}
B &=  \begin{pmatrix}1 &  B_{01} & B_{02} \\0 &1 & B_{01}  \\ 0 & 0 & 1 \end{pmatrix} \nonumber \\
C^T &=  \begin{pmatrix}1 & -B_{01} & B_{01}^2 - B_{02} \\0 & 1 & -B_{01} \\ 0 & 0 & 1 \end{pmatrix}.
\end{align} 

The symmetries involving the physical symmetry ${g}$ can be easily determined using Eq. (\ref{eq:EigenvaluesMPsym}). As the matrix pencil possesses the eigenvalue 0, we have $\beta = 0$. Moreover, w.l.o.g. we choose the normalization $\delta=1$. Hence, any physical symmetry can be chosen to be lower triangular. For any given $g$ the corresponding matrices $B_{g}, C_g$ can then easily be determined and we otbain 
\begin{align}
{g} &= \begin{pmatrix}\alpha &  0 \\\gamma &1 \end{pmatrix} \nonumber\\
B_{g} &= \begin{pmatrix}1/\alpha &  0 & 0 \\0 &1 & -\gamma  \\ 0 & 0 & \alpha \end{pmatrix} \nonumber \\
C_{g}^T &=  \begin{pmatrix}\alpha & -\gamma & 0 \\0 & 1 & 0 \\ 0 & 0 & 1/\alpha \end{pmatrix}.
\end{align}
 
The symmetry group of the state $\ket{M^3(0)}$ is thus given by
\begin{align}
{g} \otimes B_{g} B \otimes C_{g} C.
\end{align}
 
As shown in \cite{SaMo19} (see also Sec. \ref{sec:prelim}), symmetries and SLOCC transformations between TIMPS generated by any fiducial state within the SLOCC class of the representativ $\ket{M^3(0)}$ can be determined by focusing on the states of the form  $\identity \otimes b \otimes \identity \ket{M^3(0)}$. The tensor $A$ associated to $\identity \otimes b \otimes \identity \ket{M^3(0)}$ reads
\begin{align}
A_0 = \begin{pmatrix} 0& b_{00} & b_{01} \\  0 &b_{10} &  b_{11}\\ 0 & b_{20}& b_{21}\end{pmatrix}, \ 
A_1 = \begin{pmatrix} b_{00} & b_{01} & b_{02}\\  b_{10} &  b_{11} & b_{12}\\b_{20}& b_{21} & b_{22}\end{pmatrix}.
\end{align}
 
\subsubsection{Concatenation rules}
As explained in the main text, the symmetries of the MPS corresponding to a particular fiducial state can be determined by the corresponding concatenation rules. In case of the fiducial state $\identity \otimes b \otimes \identity \ket{M^3(0)}$ the concatenation rules read 
 \begin{align}
\label{eq:lhsrhsM3}
b \underbrace{\begin{pmatrix}1/\alpha_{k+1} &  B_{01}^{(k+1)}/\alpha_{k+1} & B_{02}^{(k+1)}/\alpha_{k+1} \\0 & 1 & B_{01}^{(k+1)} - \gamma_{k+1} \\ 0 & 0 & \alpha_{k+1} \end{pmatrix}}_{LHS_{k+1}} b^{-1} \propto \nonumber \\
\underbrace{\begin{pmatrix}1/\alpha_{k} &  (B_{01}^{(k)}+\gamma_k)/\alpha_{k} & (B_{02}^{(k)} + B_{01}^{(k)} \gamma_k)/\alpha_{k} \\0 & 1 & B_{01}^{(k)} \\ 0 & 0 & \alpha_{k} \end{pmatrix}}_{RHS_k}
\end{align} 
Observe that the matrices involved in Eq. (\ref{eq:lhsrhsM3}) are upper triangular. Thus, the diagonal elements are the eigenvalues of the matrices. As the determinant of the matrices on both sides equals 1, the proportionality factor must be a thrid root of unity and it follows that either $\alpha_{k+1} = \alpha_k$, or $\alpha_{k+1} = 1/\alpha_k$.
Note that all eigenvalues are distinct as long as $\alpha_k \neq \pm 1$. Similar tools as presented in the main text can be used to determine all the symmetries and to identify the normal MPS corresponding to fiducial states in this SLOCC class. One finds, apart from global symmetries for instance two-cycles. 

\subsection{TIMPS corresponding to the fiducial states represented by \texorpdfstring{$M^2(0) \oplus M^1(0)$}{M2(0)+M1(0)}}

The SLOCC class we consider here is represented by the fiducial state by $\ket{ M^2(0) \oplus M^1(0) } = \ket{0}(\ket{01}) + \ket{1}(\ket{00}+\ket{11}+\ket{22})$.
The corresponding matrix pencil 
\begin{align}
\mathcal{P} = M^2(0) \oplus M^1(0) = \begin{pmatrix}\lambda & \mu &0\\0& \lambda  & 0\\ 0 & 0 & \lambda \end{pmatrix}
\end{align}
has an eingenvalue zero. 

\subsubsection{Symmetries of the fiducial state}

Analogously to before we determine the symmetries by 
symmetries of the form  $\identity \otimes B \otimes C$, which are given by (choosing a convenient normalization)
\begin{align}
B &=  \begin{pmatrix}1 &  B_{01} &  B_{02} \\0 &1 & 0  \\ 0 &  B_{21} & B_{22} \end{pmatrix} \nonumber \\
C^T &=  \begin{pmatrix}1 & -B_{01} + \frac{ B_{02}  B_{21}}{ B_{22}}& -\frac{ B_{02}}{ B_{22}} \\0 & 1 & 0 \\ 0 & -\frac{ B_{21}}{ B_{22}} & 1/B_{22} \end{pmatrix}.
\end{align} 

and symmetries of the form $g \otimes B_g \otimes C_g$. In the case considered here they are given by 
\begin{align}
{g} &= \begin{pmatrix}\alpha &  0 \\ \gamma  &1 \end{pmatrix} \nonumber\\
B_{g} &= \begin{pmatrix}1 &  0 & 0 \\0 &\alpha & 0  \\ 0 & 0 & 1 \end{pmatrix} \nonumber \\
C_{g}^T &=  \begin{pmatrix}1 & -\gamma/\alpha & 0 \\0 & 1/\alpha & 0 \\ 0 & 0 & 1 \end{pmatrix}.
\end{align}
As before we chose w.l.o.g. the normalization $\delta = 1$. 
Note that, as in the previously considered case, the matrix pencil has a single eigenvalue, 0, and thus we have $\beta = 0$ in ${g}$.
The symmetry group of the state $\ket{ M^2(0) \oplus M^1(0) }$ is thus given by
\begin{align}
{g} \otimes B_{g} B \otimes C_{g} C.
\end{align}

The tensor $A$ associated to $\identity \otimes b \otimes \identity \ket{ M^2(0) \oplus M^1(0) }$ reads
\begin{align}
\label{eq:tensorM20M10}
A_0 = \begin{pmatrix} 0& b_{00} &0 \\  0 &b_{10} &  0\\ 0 & b_{20}& 0\end{pmatrix}, \ 
A_1 = \begin{pmatrix} b_{00} & b_{01} & b_{02}\\  b_{10} &  b_{11} & b_{12}\\b_{20}& b_{21} & b_{22}\end{pmatrix}.
\end{align}
 
\subsubsection{Concatenation rules}
For the fiducial state $\identity \otimes b \otimes \identity \ket{ M^2(0) \oplus M^1(0) }$ the concatenation rules read 
 \begin{align}
\label{eq:lhsrhsM20M10}
b \underbrace{\begin{pmatrix}1 &  B_{01}^{(k+1)} & B_{02}^{(k+1)}  \\0 & \alpha_{k+1} & 0 \\ 0 & B_{21}^{(k+1)}  & B_{22}^{(k+1)} \end{pmatrix}}_{LHS_{k+1}} b^{-1} \propto  
\underbrace{\begin{pmatrix} 1 &  B_{01}^{(k)} + \gamma_k &B_{02}^{(k)}\\0 & \alpha_{k} & 0 \\ 0 & B_{21}^{(k)} &B_{22}^{(k)}\end{pmatrix}}_{RHS_k}
\end{align} 

The symmetries can be worked out similarly to the main text. To give an example, one finds for $b$ given by 
 \begin{align}
b = \begin{pmatrix} b_{00} & b_{01} & b_{02} \\ 0 & b_{11} & b_{12} \\ 1 & b_{21} & -b_{00} -b_{11} \end{pmatrix}  \mbox{ with } b_{12} \neq 0  \end{align}

3-cycles with $\alpha_2 = 1/(\alpha_0 \alpha_1)$ for arbitrary $\alpha_0, \alpha_1 \in \mathbb{C}\setminus\{0\}$. Moreover, $\gamma_k = \frac{\tr \{ b^{-1}\det b\}}{b_{12}} (\alpha_k -1)= \frac{b_{00}^2 + b_{02} + b_{00} b_{11} + b_{11}^2 + b_{12} b_{21}}{b_{12}} (\alpha_k -1)$. 

\subsection{TIMPS corresponding to the fiducial states represented by \texorpdfstring{$M^2(0) \oplus M^1(\infty)$}{M2(0)+M1(infinity)}}
\label{app:M20M1inf}

Here, we consider the SLOCC class of fiducial states represented by $\ket{ M^2(0) \oplus M^1(\infty) } = \ket{0}(\ket{01}+\ket{22}) + \ket{1}(\ket{00}+\ket{11})$.

The corresponding matrix pencil reads
\begin{align}
\mathcal{P} = M^2(0) \oplus M^1(\infty) = \begin{pmatrix}\lambda & \mu &0\\0& \lambda  & 0\\ 0 & 0 & \mu \end{pmatrix}.
\end{align}

It has two distinct eigenvalues, 0 and $\infty$.

\subsubsection{Symmetries of the fiducial state}

Symmetries of the form  $\identity \otimes B \otimes C$ (choosing a convenient normalization) read
\begin{align}
B &=  \begin{pmatrix}1 &  B_{01} & 0 \\0 &1 & 0  \\ 0 & 0 & B_{22} \end{pmatrix} \nonumber \\
C^T &=  \begin{pmatrix}1 & -B_{01} & 0 \\0 & 1 & 0 \\ 0 & 0 & 1/B_{22} \end{pmatrix}.
\end{align} 

The matrix pencil has two distinct eigenvalues, $0$ and $\infty$. Thus, in order to satiesfy Eq. (\ref{eq:EigenvaluesMPsym}) we have $\beta = \gamma = 0$ for any symmetry on the physical system, ${g}$. W.l.o.g. we normalize $g$ such that $\delta = 1$. For a  particular $g$ one can then easily determine a particular $B_{g}, C_{g}$, such that ${g} \otimes B_{g} B \otimes C_{g}
$ is a symmetry. The operators are given by 

\begin{align}
{g} &= \begin{pmatrix}\alpha &  0 \\0  &1 \end{pmatrix} \nonumber\\
B_{g} &= \begin{pmatrix}1/\alpha &  0 & 0 \\0 &1 & 0  \\ 0 & 0 & 1/\alpha \end{pmatrix} \nonumber \\
C_{g}^T &=  \begin{pmatrix}\alpha & 0 & 0 \\0 & 1 & 0 \\ 0 & 0 & 1 \end{pmatrix}.
\end{align}

The symmetry group of the state $\ket{ M^2(0) \oplus M^1(\infty) }$ is thus given by
\begin{align}
{g} \otimes B_{g} B \otimes C_{g} C.
\end{align}

The tensor $A$ associated to $\identity \otimes b \otimes \identity \ket{ M^2(0) \oplus M^1(\infty) }$ reads
\begin{align}
A_0 = \begin{pmatrix} 0& b_{00} & b_{02} \\  0 &b_{10} &  b_{12}\\ 0 & b_{20}& b_{22}\end{pmatrix}, \ 
A_1 = \begin{pmatrix} b_{00} & b_{01} & 0\\  b_{10} &  b_{11} & 0\\b_{20}& b_{21} & 0\end{pmatrix}.
\end{align}
 
\subsubsection{Concatenation rules}
The concatenation rules are given by
 \begin{align}
\label{eq:lhsrhsM20M1inf}
b \underbrace{\begin{pmatrix}1/\alpha_{k+1} &  B_{01}^{(k+1)}/\alpha_{k+1} & 0 \\0 & 1 & 0 \\ 0 & 0 & B_{22}^{(k+1)} / \alpha_{k+1} \end{pmatrix}}_{LHS_{k+1}} b^{-1} \propto \nonumber \\
\underbrace{\begin{pmatrix}1/\alpha_{k} &  B_{01}^{(k)}/\alpha_{k} &0\\0 & 1 & 0 \\ 0 & 0 &B_{22}^{(k)}\end{pmatrix}}_{RHS_k}.
\end{align}

\subsubsection{Symmetries of the MPS}

In this class particularly interesting new symmetries occur, which we will exemplary present here. 

Considering for instance 
$b = \begin{pmatrix} 0 &  b_{01}& -b_{12} b_{21} \\ 0 & 0 & b_{12} \\ 1 & b_{21} & 0 \end{pmatrix}$ leads to $N$-cycles such that $\alpha_{k+1}  \alpha_k^2 \alpha_{k-1} = 1$ for all $k$. Now, $\alpha_k$ may be expressed in terms of $\alpha_0$ and $\alpha_1$ as follows. For even $k$ we have $\alpha_k = \frac{1}{\alpha_0^{k-1} \alpha_1^k}$ and for odd $k$ we have $\alpha_k = \alpha_0^{k-1} \alpha_1^k$. It may be easily verified that for odd $N$, one obtains only three non-trivial global symmetries with $\alpha_k = \alpha$ such that $\alpha^4 = 1$, i.e., $\alpha \in \{\pm 1, \pm i\}$. In contrast to that, for even $N$ we obtain only a less stringent condition, $(\alpha_0 \alpha_1)^N = 1$. Thus, we have $N$-cycles with arbitrary $\alpha_0 \in \mathbb{C}$, $\alpha_1 = 1/\alpha_0 e^{i \frac{2 r \pi}{N}}$ for some $r \in \{0, \ldots, N-1\}$, and then $\alpha_k = \begin{cases}\alpha_0 e^{-i \frac{2 r k \pi}{N}} & \text{for even } k \\ \frac{1}{\alpha_0} e^{i \frac{2 r k \pi}{N}} & \text{for odd } k\end{cases}$.
  Interestingly, the eigenvalues of ${g}_k$ do not coincide (up to rescaling) for all $k$.

Another interesting symmetry occurs for instance for 
 $b = \begin{pmatrix} b_{00} & - b_{00}^2& b_{02} \\ 1 & -b_{00} & 0 \\ 0 & b_{21} & 0 \end{pmatrix}.$ 
 
We find 4-cycles with $\alpha_2 = 1/\alpha_0$ and $\alpha_3 = 1/\alpha_1$ and also 2-cycles with $\alpha_1 = 1/\alpha_0$ and 1-cycles with $\alpha_0 = -1$.  

\subsection{TIMPS corresponding to the fiducial states represented by \texorpdfstring{$M^1(0) \oplus M^1(0) \oplus M^1(\infty)$}{M1(0)+M1(0)+M1(infinity)}}
\label{app:M10M10M1inf}

The representing fiducial state of this SLOCC class is $\ket{ M^1(0) \oplus M^1(0) \oplus M^1(\infty) } = \ket{0}\ket{22} + \ket{1}(\ket{00}+\ket{11})$.
The corresponding matrix pencil reads
\begin{align}
\mathcal{P} = M^1(0) \oplus M^1(0) \oplus M^1(\infty) = \begin{pmatrix}\lambda & 0 &0\\0& \lambda  & 0\\ 0 & 0 & \mu \end{pmatrix}.
\end{align}

The matrix pencil has (as in the previous case) two distinct eigenvalues, $0$ and $\infty$. However, the degeneracy is different compared to the case studied before. 

\subsubsection{Symmetries of the fiducial state}

Symmetries of the form  $\identity \otimes B \otimes C$ are given by
\begin{align}
B &=  \begin{pmatrix}B_{00} &  B_{01} & 0 \\B_{10} &B_{11} & 0  \\ 0 & 0 & B_{22} \end{pmatrix} \nonumber \\
C^T &= B^{-1} 
\end{align} 
In the following we use the normalization $B_{22}=1$.

Symmetries involving ${g}$ (and choosing a proper normalization) are of the following simple form
\begin{align}
{g} &= \begin{pmatrix}\alpha &  0 \\0  &1 \end{pmatrix} \nonumber\\
B_{g} &= \begin{pmatrix}1 &  0 & 0 \\0 &1 & 0  \\ 0 & 0 & 1/\alpha \end{pmatrix} \nonumber \\
C_{g} &= \identity.
\end{align}
As the matrix pencil now has two distinct eigenvalues, $0$ and $\infty$, we have $\beta = \gamma = 0$ in ${g}$.
The symmetry group of the state $\ket{ M^1(0) \oplus M^1(0) \oplus M^1(\infty) }$ is thus given by
\begin{align}
{g} \otimes B_{g} B \otimes  C.
\end{align}

The tensor $A$ associated to $\identity \otimes b \otimes \identity \ket{ M^1(0) \oplus M^1(0) \oplus M^1(\infty) }$ reads
\begin{align}
A_0 = \begin{pmatrix} 0& 0 & b_{02} \\  0 &0 &  b_{12}\\ 0 & 0& b_{22}\end{pmatrix}, \ 
A_1 = \begin{pmatrix} b_{00} & b_{01} & 0\\  b_{10} &  b_{11} & 0\\b_{20}& b_{21} & 0\end{pmatrix}.
\end{align}

\subsubsection{Concatenation rules}

As in all the previous cases the symmetries of the correspongind MPS can be determined by solving the concatatation rules, which are given in this case by 

 \begin{align}
\label{eq:lhsrhsM10M10M1inf}
b \underbrace{\begin{pmatrix}B_{00}^{(k+1)} &  B_{01}^{(k+1)} & 0 \\B_{10}^{(k+1)} & B_{11}^{(k+1)} & 0 \\ 0 & 0 & 1 / \alpha_{k+1} \end{pmatrix}}_{LHS_{k+1}} b^{-1} \propto \nonumber \\
\underbrace{\begin{pmatrix}B_{00}^{(k)} &  B_{01}^{(k)} &0\\B_{10}^{(k)} & B_{11}^{(k)} & 0 \\ 0 & 0 &1\end{pmatrix}}_{RHS_k}.
\end{align}

\section{Table combining results on the symmetry group and the SLOCC classification of the TIMPS \texorpdfstring{$\Psi(LLT)$}{Psi(LLT)}}
\label{sec:L1L1Ttable}

\begin{table*}[b]
 \begin{tabular}{c|c||c| >{\raggedright\arraybackslash}m{0.3\linewidth}||c|   >{\raggedright\arraybackslash}m{0.25\linewidth}}
 type & normal & \# symm.& \multicolumn{1}{c||}{cycles} & \makecell{\# SLOCC \\classes} & \multicolumn{1}{c}{representatives}\\
 \hline\hline
  IIb & yes &$1$-param.&      1-cycle: ${g} = R \begin{pmatrix}1&y\\0&1 \end{pmatrix} R^{-1}$ for any $y \in \mathbb{C}$                 & $1$-param.& $T= \begin{pmatrix}1&1\\0&1 \end{pmatrix}$, $\vec{v}= (x,0)^T$, $x \in \mathbb{C} \setminus \{0\}$.\\ \hline 
  IV& yes & \makecell{$1$-param. ($m$\\ divides $N$),\\ 1 (else)}&     $m$-cycle: ${g}_k = S_0 \begin{pmatrix}1&e^{i \frac{2 k r \pi}{m}}\\0&1 \end{pmatrix} S_0^{-1}$, where $S_0 = (\vec{v},\, \vec{w})$ for any $\vec{w} \in \mathbb{C}^2$                &  \makecell{$m$ considering \\ a fixed $m$  \footnote{Note that there exist exactly $m$ different pairs of $r$ and $\vec{v}$ in the most right column. However, the SLOCC classes corresponding to e.g. $m=2$ would also be counted in the case of $m=4$ (as $T^2\propto \identity$ implies $T^4 \propto \identity$).}}& $T= \begin{pmatrix}e^{i \frac{r \pi}{m}}&0\\0& e^{-i \frac{r \pi}{m}} \end{pmatrix}$ for $r\in\{1, \ldots, \lfloor \frac{m-1}{2} \rfloor\}$, $\vec{v}\in \{(1,0)^T,(0,1)^T\}$. Moreover, $T \in \{\identity, i \sigma_z\}$ ($T=\identity$) for even $m$ (odd $m$), respectively, and $\vec{v}=(1,0)^T$.\\ \hline 
  V& yes & $1$ (trivial)  &     none               &$1$-param.   & $T= \begin{pmatrix}1&1\\0&1 \end{pmatrix}$, $\vec{v} =  (0,x)^T$, $x\in \mathbb{C} \setminus \{0\}$. \\ \hline 
  VI& yes &$1$ (trivial) &          none           & \makecell{$2$-param.,\\generic}& $T= \begin{pmatrix}\sigma&0\\0&1/\sigma \end{pmatrix}$, $\sigma \in \mathbb{D}$ s.t. $\sigma^n\neq 1$ for any $n \in \mathbb{N}$, $\vec{v}=(1,0)^T$, $\vec{v}=(0,1)^T$, or$\vec{v} =  (v_1,1/v_1)^T$, $v_1 \in \mathbb{C} \setminus \{0\}$.\\ \hline 
  VII& yes&\makecell{$1$-param. (even \\$N$), 2 (odd $N$)}&  2-cycle: ${g}_0$ s.t. ${g}_0 \vec{v} \propto \vec{v}$, ${g}_0 T\vec{v} \propto T\vec{v}$ with any $x \in \mathbb{C}\setminus\{0,i\}$, ${g}_1 = {g}_0^{-1}$. \newline  1-cylce: ${g}_0$ as above, but with $x=i$.                     & $1$-param.& $T= \begin{pmatrix}i&0\\0&-i \end{pmatrix}$, $\vec{v} =  (v_1,1/v_1)^T$, $v_1 \in \mathbb{C} \setminus \{0\}$.\\  \hline 
  VIII& yes &$1$ (trivial) &    none                  & $1$-param.&  $T= \begin{pmatrix}e^{i \frac{s \pi}{m}}&0\\0& e^{-i \frac{s \pi}{m}} \end{pmatrix}$ for $s\in\{0, \ldots, \lfloor m/2 \rfloor\}$, $\vec{v} =  (v_1,1/v_1)^T$, $v_1 \in \mathbb{C} \setminus \{0\}$.  \\ 
  \hline\hline
  I & no & $3$-param.&           1-cycle: ${g} = R \begin{pmatrix}x&\\&1/x \end{pmatrix} R^{-1}$ for any $x\in \mathbb{C}\setminus\{0\}$. \newline $m$-cycle: ${g}_k= T^k g_0 T^{-k}$ for any ${g}_0$. \newline $m/2$-cycle: ${g}_k = R  \begin{pmatrix}0 & i z\\ i/z & 0 \end{pmatrix} R^{-1}$, $z = y e^{i \frac{2k (2r+1) \pi}{m}}$ for any $y \in \mathbb{C}\setminus\{0\}$            &2 (even $N$)& $T \in \{\identity, i \sigma_z\}$,  $\vec{v}= \vec{0}$\\ \hline
  IIa & no &$1$-param.&     1-cycle: ${g} = R \begin{pmatrix}1&y\\0&1 \end{pmatrix} R^{-1}$ for any $y \in \mathbb{C}$                   & 1 (even $N$) & $T= \begin{pmatrix}1&1\\0&1 \end{pmatrix}$, $\vec{v}= \vec{0}$\\ \hline
  III & no &$1$-param.  &     1-cycle: ${g} = R \begin{pmatrix}x&\\&1/x \end{pmatrix} R^{-1}$ for any $x\in \mathbb{C}\setminus\{0\}$                  &  \makecell{$1$-param.\\ (even $N$)} & $T= \begin{pmatrix}\sigma&0\\0&1/\sigma \end{pmatrix}$, $\vec{v}= \vec{0}$, $\sigma \in \mathbb{C}\setminus \{0\}$ s.t. there exists no $m \in \mathbb{N}$ s.t. $\sigma^{m} = \pm 1$.\\ 
 \end{tabular}
 \caption{Summary of the symmetries (cycles in $G_b$) and SLOCC classification of all normal MPS generated by $\identity \otimes b \otimes \identity \ket{LLT}$ (first part) plus partial results on SLOCC classes of some non-normal MPS (second part of the table). The ``type'' corresponds to the labels displayed in Figure \ref{fig:L1L1Tflowchart}. Note that for non-normal MPS the symmetry group might be larger than displayed, as the utilized methods may fail to identify the full symmetry group, but yield a subgroup instead. Similarly, additional non-normal MPS that have not been identified as SLOCC equivalent might in fact be equivalent. The displayed non-normal MPS vanish in case of odd $N$. As in Figure \ref{fig:L1L1Tflowchart}, by $m$ we denote a number such that $T^{m}\propto \identity$ (if such a number exists) and write $T \propto R \begin{pmatrix}e^{i \frac{r \pi}{m}} & \\ &  e^{-i \frac{r \pi}{m}}\end{pmatrix} R^{-1}$, $r \in \{0, \ldots, N-1\}$ for some matrix $R$ and $r \in \{0, \ldots, \lfloor m/2 \rfloor\}$.
 }
  \label{tab:L1L1Tfull}
\end{table*}

In this appendix we show Table \ref{tab:L1L1Tfull}, which summarizes the results on the symmetry groups and SLOCC classes of MPS corresponding to fiducial states in the $LLT$ class.

\end{document}